\documentclass[a4paper,11pt]{article}

\usepackage[T1]{fontenc}
\IfFileExists{lmodern.sty}{\usepackage{lmodern}}{}
\usepackage[english]{babel}
\usepackage{amsmath,amssymb,amsfonts,amsthm,mathtools}
\usepackage{mathrsfs}
\IfFileExists{lmodern.sty}%
  {\usepackage{microtype}}%
  {\usepackage[expansion=false]{microtype}}
\usepackage{enumitem,needspace}
\usepackage{cite}
\usepackage[colorlinks=true,allcolors=blue]{hyperref}

\numberwithin{equation}{section}
\setlist[enumerate]{leftmargin=2.2em}
\setlist[itemize]{leftmargin=2.0em}

\newcommand{\R}{\mathbb R}
\newcommand{\KerrSpace}{\mathcal K}
\newcommand{\CarterClass}{\mathfrak C_{M,a}}
\newcommand{\Ring}{\mathcal R}
\newcommand{\tdiam}{\operatorname{tdiam}}
\newcommand{\esssup}{\operatorname*{ess\,sup}}
\newcommand{\abs}[1]{\left\lvert #1\right\rvert}

\theoremstyle{plain}
\newtheorem{theorem}{Theorem}[section]
\newtheorem{proposition}[theorem]{Proposition}
\newtheorem{lemma}[theorem]{Lemma}
\newtheorem{corollary}[theorem]{Corollary}
\theoremstyle{definition}
\newtheorem{definition}[theorem]{Definition}
\theoremstyle{remark}
\newtheorem{remark}[theorem]{Remark}

\hypersetup{
  pdftitle={The C0-inextendibility of Kerr spacetimes},
  pdfauthor={Bobby Eka Gunara},
  pdfkeywords={Kerr spacetime, C0-inextendibility, ring singularity, causal homotopy, timelike diameter, Carter block extension},
  bookmarksnumbered,
  pdfstartview={FitH}
}

\begin{document}
\pagestyle{plain}

\title{The $C^0$-inextendibility of Kerr spacetimes}

\author{Bobby Eka Gunara\thanks{Theoretical Physics Laboratory, Theoretical High Energy Physics Research Division, Faculty of Mathematics and Natural Sciences, Institut Teknologi Bandung, Jl. Ganesha no. 10, Bandung 40132, Indonesia. Email: \texttt{bobby@itb.ac.id}.}}
\date{}
\maketitle

\begin{abstract}
We prove the $C^0$-inextendibility of the full Kerr realizations
considered in this paper and all their connected covering spacetimes. The result holds for positive mass
and nonzero rotation in the subextremal, extremal, and overextremal
regimes. The obstruction at the singular ring is expressed by unbounded
proper time in a fixed-endpoint causal homotopy class. It persists under
every positive continuous conformal rescaling and under continuous
widening of the Kerr causal cones, without any global causality
assumption. We give sharp bounds in a fixed coordinate norm for
stationary and axisymmetric perturbations and describe complete timelike geodesics that encounter
arbitrarily large curvature near the ring.

\end{abstract}


\tableofcontents

\section{Introduction and Main Results}\label{sec:introduction}

The Kerr solution describes a rotating vacuum gravitational field
\cite{Kerr1963}. The metric extends analytically across its regular
horizons, whereas the curvature diverges at the singular ring.
Since curvature depends on derivatives of the metric, this divergence
does not by itself exclude an extension with a continuous metric.
In this paper we prove the continuous inextendibility of the full Kerr
spacetimes specified below.

It is essential here to distinguish the full spacetime from its
exterior and its maximal globally hyperbolic development. The exterior extends through the event horizon, and the
exact subextremal maximal globally hyperbolic development extends
through its Cauchy horizon. We include all regular horizon crossings
and both regular branches of the $r=0$ disc, and remove only the
curvature rings. Carter's global construction provides the starting
point \cite{Carter1966,Carter,BoyerLindquist,ONeillKerr}; the two disc
branches must remain distinct \cite{MacCallum2026}.

\subsection{The main results}

A connected smooth manifold with a continuous Lorentzian metric is
\emph{$C^0$-inextendible} if it has no smooth isometric embedding as a
proper open subset of a connected smooth manifold carrying a continuous,
nondegenerate Lorentzian metric. Fix $M>0$ and $a\ne0$. We consider:
\begin{enumerate}[label=(\roman*)]
 \item for $0<|a|<M$, any member of the horizon-saturated Carter class
 $\CarterClass$ in Definition~\ref{def:carter-spacetime};
 \item for $|a|=M$, the full ingoing--outgoing chain
 $\mathcal K^{\mathrm{ext}}_{M,a}$ in
 Section~\ref{subsec:extremal-kerr};
 \item for $|a|>M$, the full two-sheet spacetime
 $\mathcal K^{\mathrm{over}}_{M,a}$ in
 \eqref{eq:overextremal-kerr-manifold}.
\end{enumerate}
The subextremal class consists of connected, Hausdorff, time-oriented
realizations from full blocks and full horizon crossings, glued by
stationary-axial translations. The extremal chain uses complete
ingoing and outgoing collars and has no bifurcation sphere.
Both regular disc branches are included and distinct in every model.
In each realization all horizon gluings are full, and no additional
disc identification is made. Set $\Sigma\equiv r^2+a^2\cos^2\theta$, with the same
notation for its pullback to a cover.

\begin{theorem}[$C^0$-inextendibility]\label{thm:main-c0}
Every connected covering spacetime of any of the above Kerr
realizations, with the pulled-back metric, is $C^0$-inextendible.
Every inextendible unit-speed timelike geodesic with a finite
proper-time endpoint satisfies $\Sigma\to0$ there.
\end{theorem}

The extension metric is an arbitrary continuous nondegenerate
Lorentzian metric; it is subject to no symmetry, time-orientation,
coordinate, or global causality requirement. The subextremal assertion includes
stationary-axial screw quotients and all their connected covers.

A causal
homotopy is a continuous deformation of a curve, fixing both endpoints,
whose intermediate curves are all future causal; constant pauses are
allowed. This requirement is stronger than ordinary path homotopy
\cite{MoralesSanchez2015}. 
For a future timelike curve $\gamma$ and $s<t$, define
\begin{equation}\label{eq:intro-causal-length}
 \mathcal L^h_\gamma(s,t)\equiv
 \sup\{L_h(\sigma):\sigma\simeq_c\gamma|_{[s,t]}\}.
\end{equation}
The supremum is taken in one fixed-endpoint causal homotopy class.
Writing $\mu\equiv\cos\theta$, the condition
$\Sigma=r^2+a^2\mu^2\to0$ expresses approach to the ring.

\Needspace{14\baselineskip}
\begin{theorem}[Infinite lengths at the Kerr ring]
\label{thm:main-terminal-lengths}
Let $\gamma:[0,\Lambda)\to\KerrSpace$ be a future-directed $C^1$
timelike curve in a spacetime of Theorem~\ref{thm:main-c0}, equipped
with $\Omega^2g$ for an arbitrary positive continuous function
$\Omega$. Suppose that $\Sigma(\gamma(t))\to0$ as $t\uparrow\Lambda$.
For every $s\in[0,\Lambda)$ there is $t_0\in(s,\Lambda)$ such that
\begin{equation}\label{eq:main-infinite-lengths}
 \mathcal L^{\Omega^2g}_\gamma(s,t)=+\infty\qquad(t_0<t<\Lambda).
\end{equation}
The statement holds in either time direction, locally on every full
punctured Kerr ring collar, and on every connected cover of it.
On a cover, $\Omega$ need not descend to the base. No symmetry or
asymptotic bound on $\Omega$ is imposed.
Moreover, every future-directed locally Lipschitz causal curve with
$\Sigma\to0$, including a null curve, has no endpoint in a $C^0$
extension of the conformally rescaled spacetime.
\end{theorem}

\subsection{Sketch of the proof}

Near the Kerr ring, a short transverse excursion reaches two points
with the following property. Two independent linear combinations of
the stationary and axial Killing fields are future timelike at one
point, and their negatives are future timelike at the other.
Four timelike orbit segments
then cancel the excursion's stationary and angular displacement.
A small variation of an earlier timelike segment keeps the endpoints
fixed throughout the deformation. Thus a nonconstant closed causal
curve $c$ can be inserted into a reference segment $a$:
\begin{equation}\label{eq:intro-loop-insertion}
 a\simeq_c a*c.
\end{equation}
The homotopy fixes the endpoints of $a$ and deforms the concatenation
as a whole; it does not contract $c$ through closed causal curves.

A continuous nondegenerate Lorentzian metric admits small coordinate
neighborhoods with a time coordinate and bounded causal speeds.
The resulting length bound applies to curves contained in the chart.
To compare it with the ring construction, we must also keep the causal
homotopy inside the chart. If $a$ lies sufficiently near an extension
endpoint, the closure of its causal diamond, formed by curves in the chart,
is compactly contained in the chart. An open-and-closed argument then confines
every slice of the fixed-endpoint homotopy to the same neighborhood. The inserted closed curve would then have
both constant and strictly increasing coordinate time, a contradiction.
This uses only local causality in the extension
\cite{Sbierski2018,GallowayLing2017,Guenther2021,ChruscielGrant2012}.
It applies even though the full Kerr spacetime contains closed timelike
curves. Repeating $c$ also makes the lengths in the same homotopy class
unbounded. Every positive continuous conformal factor has a positive
minimum on the compact image of $c$, so this length conclusion survives
arbitrary conformal rescaling.

To pass from the local obstruction to the global theorem, suppose
that a proper continuous extension exists. The boundary theorem of
Minguzzi and Suhr supplies a causal maximizer of finite positive length
ending at its boundary
\cite[Lemma~2.2 and Theorem~2.3]{MinguzziSuhr2019}.
In the smooth interior it is a timelike geodesic. Carter separation
and regular horizon continuation show that every finite-proper-time
terminal timelike geodesic approaches the ring. The local contradiction
therefore excludes every proper extension of the full realizations
described above.

At a continuous extension endpoint, local confinement instead gives
\begin{equation}\label{eq:intro-vanishing-lengths}
 \lim_{u\uparrow\Lambda}\sup_{u<s<t<\Lambda}
 \mathcal L^h_\gamma(s,t)=0
\end{equation}
(Theorem~\ref{thm:causal-length-criterion}). The ring construction
uses the full stationary range but arbitrarily small transverse
collars. Its homotopy lifts with fixed endpoints to every cover, so
the lifted inserted curve is closed. For a null approach, an auxiliary
timelike prefix supplies the endpoint correction; the infinite-length
assertion retains the timelike hypothesis on the reference curve.
Cone inclusion preserves the same argument for continuous Lorentzian
metrics whose future cones contain the Kerr cones
(Corollaries~\ref{thm:ring-cone-widening} and
\ref{cor:causal-ring-approaches}).

\subsection{Perturbations and geometric consequences}

The same construction gives quantitative perturbation results.
For continuous stationary and axisymmetric Lorentzian perturbations
$g+h$, the construction persists under the sharp strict estimate
\begin{equation}\label{eq:intro-sharp-conformal-bound}
 \limsup_{\Sigma\to0}\frac{\|h\|_{\mathrm E}}{\Sigma}
 <\frac1{1+a^2},
\end{equation}
in the specified coordinate norm. Rank-one examples above this
coefficient admit a conformal representative with finite length in
every fixed-endpoint causal homotopy class
(Theorem~\ref{thm:sharp-conformal-perturbation}). 

The terminal-geodesic analysis also yields results on generic
completeness, curvature near the ring, and polar Jacobi fields
(Section~\ref{sec:global-proof}). Section~\ref{sec:symmetry} treats
axial widths and the behavior of Killing fields at boundaries of
Kerr subdomains. The topology of the Carter gluings is described
in Appendix~\ref{app:explicit-carter-atlas}.

\subsection{Relation to other inextendibility results}

The ring obstruction concerns a different boundary from the regular
exact Cauchy horizon. For dynamical vacuum black holes, it has been proved 
nonlinear $C^0$-stability of the Cauchy horizon for
prescribed interior data \cite{DafermosLuk2025};
under nonlinear Price-law assumptions, Lipschitz
inextendibility has also been proved at a dynamical Cauchy horizon near timelike infinity
\cite{Gurriaran2026}. The timelike-completeness criterion by assuming global hyperbolicity has been discussed in
\cite{GLS2018}. The author established the warped-product theorem in the static spacetimes \cite{Gunara2026}, separately with other author \cite{Mosani2026} that  do not cover
the two-sheet Kerr ring .

Curvature and holonomy obstructions at higher regularity
\cite{SbierskiHolonomy,SbierskiWeakNull}, and questions of extension
uniqueness \cite{SbierskiUniqueness2024}, require different arguments.
Here only local cone bounds in a continuous extension are used.
The global theorems concern exact Kerr geometry. The perturbation
theorems below concern stationary and axisymmetric metrics near the
ring.
Conformal Kerr geometry has also been studied through curvature
regularity and geodesic completeness \cite{BambiModestoRachwal2017};
conformal length is considered in the synthetic setting in
\cite{ManzanoMosaniSamannZoghlami2025}. Our local conclusion allows
every positive continuous conformal factor and requires no completeness
assumption. The continuous $4\pi$ cone completion of the transverse
orbit space \cite{ChruscielMaliborskiYunes} is compatible with the
absence of a nondegenerate four-dimensional Lorentzian ring extension.

\subsection{Outline of the paper}

In Section~\ref{sec:setting} we define the Kerr realizations and fixes
the notation, while in Section~\ref{sec:c0geometry} we prove the local
bound for lengths in a causal homotopy class at a continuous boundary.
We provide the constructions of the causal homotopies
near the ring in Section~\ref{sec:two-sheet-fan}. We continue to 
show that a proper extension would produce a timelike geodesic ending
at the ring in Sections~\ref{sec:boundary-tools}-\ref{sec:geodesics}, 
and Section~\ref{sec:global-inextendibility} completes
the proofs of the main theorems in all three horizon regimes.
The perturbation estimates are proved in
Section~\ref{sec:ring-perturbations}. Sections~\ref{sec:global-proof}
and~\ref{sec:symmetry} discuss terminal geodesics, curvature,
Jacobi fields, and axial boundary geometry. Appendix~\ref{app:explicit-carter-atlas} is devoted to
describe the topology of the Carter realizations; the second appendix
derives the scalar expansion used for the higher-order terminal terms.

\section{Kerr geometry and Carter realizations}\label{sec:setting}

We first fix the local Kerr geometry and then specify the global
realizations used in the main theorem. The full horizon crossings and
the distinction between the two regular disc branches will be needed
in the continuation of timelike geodesics.

\subsection{The Kerr metric}

We use signature $(-,+,+,+)$ and first treat the subextremal range
\begin{equation}\label{eq:parameters}
 M>0,\qquad 0<|a|<M.
\end{equation}
The local ring arguments require only $M>0$, $a\ne0$;
Section~\ref{sec:global-inextendibility} gives the extremal and
overextremal models.
Our curvature convention is
\begin{equation}\label{eq:curvature-convention}
 R^\rho{}_{\sigma\mu\nu}
 =\partial_\mu\Gamma^\rho_{\nu\sigma}
 -\partial_\nu\Gamma^\rho_{\mu\sigma}
 +\Gamma^\rho_{\mu\lambda}\Gamma^\lambda_{\nu\sigma}
 -\Gamma^\rho_{\nu\lambda}\Gamma^\lambda_{\mu\sigma}.
\end{equation}
For the local Hodge dual near the ring we choose $e^0\wedge e^1\wedge e^2\wedge e^3$ in the coframe
below. This choice requires no global spacetime orientation;
reversing it changes only the sign of $\mathcal W_2$.
Our Newman-Penrose convention is $\Psi_2=-M(r-ia\cos\theta)^{-3}$.
Kerr's vacuum metric~\cite{Kerr1963}, written in a Boyer-Lindquist block, is
\begin{align}
g_{M,a}={}&-\left(1-\frac{2Mr}{\Sigma}\right)dt^2
-\frac{4Mar\sin^2\theta}{\Sigma}\,dt\,d\phi
+\frac{\Sigma}{\Delta}\,dr^2+\Sigma\,d\theta^2 \nonumber\\
&+\sin^2\theta\left(r^2+a^2+\frac{2Ma^2r\sin^2\theta}{\Sigma}\right)d\phi^2,
\label{eq:kerr-metric}
\end{align}
where
\begin{equation}\label{eq:sigma-delta}
 \Sigma\equiv r^2+a^2\cos^2\theta,
 \qquad
 \Delta\equiv r^2-2Mr+a^2.
\end{equation}
In the subextremal range $0<|a|<M$, the two roots are
\begin{equation}\label{eq:horizon-radii}
 r_\pm\equiv M\pm\sqrt{M^2-a^2}.
\end{equation}
The coefficient of $dt\,d\phi$ in \eqref{eq:kerr-metric} is
$2g_{t\phi}$. Thus
\begin{equation}\label{eq:orbit-components}
 g_{tt}=-1+\frac{2Mr}{\Sigma},
 \quad
 g_{t\phi}=-\frac{2Mar\sin^2\theta}{\Sigma},
 \quad
 g_{\phi\phi}=\sin^2\theta\left(r^2+a^2+\frac{2Ma^2r\sin^2\theta}{\Sigma}\right).
\end{equation}
Direct expansion gives
\begin{equation}\label{eq:determinants}
 g_{tt}g_{\phi\phi}-g_{t\phi}^2=-\Delta\sin^2\theta,
 \qquad
 \det(g_{\mu\nu})=-\Sigma^2\sin^2\theta.
\end{equation}
The vanishing of this coordinate determinant at the poles is due to
the polar coordinates. The geometric singular set is
\begin{equation}
 \Ring\equiv\{\Sigma=0\}=\left\{r=0,\ \theta=\frac\pi2\right\}.
\end{equation}
At $r=0$ with $\theta\ne\pi/2$, one has $\Sigma>0$ and $\Delta=a^2>0$; these are regular disc points of the analytic extension.

\begin{lemma}[Orthonormal coframe near the ring]\label{lem:coframe}
Where $\Sigma>0$ and $\Delta>0$, the one-forms
\begin{align}
 e^0&\equiv\sqrt{\frac{\Delta}{\Sigma}}\,(dt-a\sin^2\theta\,d\phi),
 &e^1&\equiv\sqrt{\frac{\Sigma}{\Delta}}\,dr,\nonumber\\
 e^2&\equiv\sqrt\Sigma\,d\theta,
 &e^3&\equiv\frac{\sin\theta}{\sqrt\Sigma}\bigl((r^2+a^2)d\phi-a\,dt\bigr)
\end{align}
satisfy
\begin{equation}
 g_{M,a}=-(e^0)^2+(e^1)^2+(e^2)^2+(e^3)^2.
\end{equation}
The metric is Lorentzian on every block and extends analytically
through the horizon charts in \textup{(C4)}.
\end{lemma}

\begin{proof}
The $dr^2$ and $d\theta^2$ terms are immediate.  For the orbit block,
put $\mathfrak s\equiv\sin\theta$ and $\rho^2\equiv r^2+a^2$.  Expanding
\begin{equation}
 -\frac{\Delta}{\Sigma}(dt-a\mathfrak s^2d\phi)^2
 +\frac{\mathfrak s^2}{\Sigma}(\rho^2d\phi-a\,dt)^2
\end{equation}
gives the three coefficients in \eqref{eq:orbit-components}, using
$\rho^2-\Delta=2Mr$ and $\rho^2-a^2\mathfrak s^2=\Sigma$.  The determinant of the
coframe matrix is $\Sigma\sin\theta$, hence
$\det g=-\Sigma^2\sin^2\theta$.  On $r_-<r<r_+$, one has $g_{rr}<0$, $g_{\theta\theta}>0$,
and a positive definite orbit form, since $g_{\phi\phi}>0$ and
$-\Delta\sin^2\theta>0$.  Smooth polar charts cover the axis, and
Lemma~\ref{lem:app-crossing-metric} preserves the signature across
both horizons.
\end{proof}

\subsection{Horizon-saturated Carter realizations}

Put
\begin{equation}
 I_1\equiv(-\infty,r_-),\qquad I_2\equiv(r_-,r_+),\qquad
 I_3\equiv(r_+,\infty).
\end{equation}
Let $\mu:\mathbb S^2\to[-1,1]$ be the standard axial height function,
equal to $\cos\theta$ away from the poles, and define the model blocks
\begin{equation}\label{eq:model-blocks}
 \mathcal B_j\equiv
 \bigl(\mathbb R_t\times I_j\times\mathbb S^2\bigr)
 \setminus\{(t,r,\omega):r=0,\ \mu(\omega)=0\}.
\end{equation}
The deleted set occurs only in $\mathcal B_1$.  Each block carries the
metric \eqref{eq:kerr-metric}, interpreted in smooth polar charts at the
axis.

For $h\in\{r_-,r_+\}$, let $\bar h$ be the other root and put
\begin{equation}\label{eq:horizon-normalizations}
 A_h\equiv h^2+a^2,
 \qquad \kappa_h\equiv\frac{h-\bar h}{2A_h},
 \qquad \Omega_h\equiv\frac{a}{A_h}.
\end{equation}
Let $\mathscr R_\zeta$ denote axial rotation of $\mathbb S^2$ through
angle $\zeta$, and write
\begin{equation}
 \mathsf A\equiv\mathbb R\times(\mathbb R/2\pi\mathbb Z)
\end{equation}
for the stationary-axial group.  If
$\boldsymbol\lambda=(\tau,\zeta)\in\mathsf A$, its actions on a model
block and on a crossing model are
\begin{align}
 \mathscr S_{\boldsymbol\lambda}(t,r,\omega)
 &\equiv\bigl(t+\tau,r,\mathscr R_\zeta\omega\bigr),
 \label{eq:graph-block-action}\\
 \widehat{\mathscr S}_{h,\boldsymbol\lambda}(U,V,\omega)
 &\equiv\bigl(e^{-\kappa_h\tau}U,e^{\kappa_h\tau}V,
 \mathscr R_{\zeta-\Omega_h\tau}\omega\bigr).
 \label{eq:graph-crossing-action}
\end{align}
The sector maps \eqref{eq:app-sector-map} satisfy
\begin{equation}\label{eq:graph-equivariance}
 \Xi_h^{\epsilon\eta}\circ\mathscr S_{\boldsymbol\lambda}
 =\widehat{\mathscr S}_{h,\boldsymbol\lambda}
 \circ\Xi_h^{\epsilon\eta}.
\end{equation}
These actions preserve the metric, $r$, $\mu$, and the Kerr symmetry
fields and tensors.

\begin{definition}[Horizon-saturated Carter realization]\label{def:carter-spacetime}
For fixed parameters \eqref{eq:parameters}, a member of
$\CarterClass$ is a marked tuple
\begin{equation}
 (\KerrSpace,g_{M,a},r,\mu,\mathfrak A)
\end{equation}
with the following properties.
\begin{enumerate}[label=(C\arabic*)]
 \item $(\KerrSpace,g_{M,a})$ is a connected, Hausdorff, second-countable,
 time-oriented analytic Lorentzian four-manifold.  The functions $r$ and
 $\mu$ are analytic, and
 \begin{equation}
  \Sigma=r^2+a^2\mu^2>0
 \end{equation}
 everywhere.  Thus every copy of the curvature ring is absent.
 \item The set
 \begin{equation}
  \mathcal H\equiv\{p\in\KerrSpace:\Delta(r(p))=0\}
 \end{equation}
 is the union of the regular future and past Kerr horizon branches
 together with their bifurcation two-spheres.  Every connected
 component $B_\alpha$ of $\KerrSpace\setminus\mathcal H$ comes with a marked
 analytic isometry
 \begin{equation}
  \Psi_\alpha:\mathcal B_{j(\alpha)}\longrightarrow B_\alpha
 \end{equation}
 preserving $r$ and $\mu$.  A block of type $\mathcal B_1$ contains the
 two distinct
 regular discs $\{r=0,\mu>0\}$ and $\{r=0,\mu<0\}$; they are not
 identified.
 \item The local fields $\partial_t$, $\partial_\phi$ and the Carter
 tensor \eqref{eq:app-carter-tensor} are the standard Kerr objects.
 Proposition~\ref{prop:canonical-kerr-fields} derives their global
 continuation from \textup{(C4)} and \textup{(C6)}; they are not
 additional data.

 \item Let $\mathscr E$ be the set of finite block endpoints,
 \begin{equation}
  \mathscr E\equiv\{(\alpha,h):h\in\{r_-,r_+\}
  \text{ is an endpoint of }I_{j(\alpha)}\}.
 \end{equation}
 There is a finite or countable set $\mathscr C$ of crossing labels.
 Each $\beta\in\mathscr C$ has a type
 $h_\beta\in\{r_-,r_+\}$ and a distinguished analytic isometric open
 embedding
 \begin{equation}
  \jmath_\beta:\mathcal X_{h_\beta}\longrightarrow\KerrSpace.
 \end{equation}
 The crossing images are pairwise disjoint.  Their axes exhaust the
 regular horizon set:
 \begin{equation}\label{eq:def-horizon-exhaustion}
  \mathcal H=\coprod_{\beta\in\mathscr C}
  \jmath_\beta\bigl(\{UV=0\}\times\mathbb S^2\bigr).
 \end{equation}
 The restrictions to $\{V\ne0\}$ and $\{U\ne0\}$ give the full
 ingoing and outgoing Eddington-Finkelstein charts
 \eqref{eq:app-ingoing-from-kruskal}-\eqref{eq:app-outgoing-from-kruskal}.

 There is a bijection
 \begin{equation}\label{eq:def-port-bijection}
  \vartheta:\mathscr E\longrightarrow
  \{(\beta,Q):\beta\in\mathscr C,
  Q\in\{Q_{++},Q_{--},Q_{+-},Q_{-+}\}\}.
 \end{equation}
 If $\vartheta(\alpha,h)=(\beta,Q_{\epsilon\eta})$, then
 $h=h_\beta$ and
 \begin{equation}\label{eq:def-radial-port-rule}
 \begin{array}{c|cc}
  &h=r_+&h=r_-\\ \hline
  \mathcal B_1&\text{--}&\epsilon\eta=-1\\
  \mathcal B_2&\epsilon\eta=-1&\epsilon\eta=+1\\
  \mathcal B_3&\epsilon\eta=+1&\text{--}.
 \end{array}
 \end{equation}
 For every endpoint $e=(\alpha,h)$ there is a parameter
 $\boldsymbol\lambda_e=(\tau_e,\zeta_e)\in\mathsf A$.  On the full
 collar
 \begin{equation}
  \mathcal C_e\equiv\mathbb R_t\times
  (I_{j(\alpha)}\cap J_h)\times\mathbb S^2
 \end{equation}
 the block and crossing embeddings satisfy the exact identity
 \begin{equation}\label{eq:def-exact-sector-identity}
  \Psi_\alpha(x)=\jmath_\beta\!\left(
  \widehat{\mathscr S}_{h,\boldsymbol\lambda_e}
  \bigl(\Xi_h^{\epsilon\eta}(x)\bigr)\right),
  \qquad x\in\mathcal C_e.
 \end{equation}
 In particular, all four compatible sectors occur at each crossing;
 no regular horizon generator is an end of $\KerrSpace$.

 \item Form the bipartite incidence multigraph $\mathscr G$ with the
 block labels $\alpha$ and crossing labels $\beta$ as vertices, and one
 edge $e=(\alpha,h)$ joining $\alpha$ to the crossing specified by
 $\vartheta(e)$.  The graph is connected and countable.  The marked
 families
 \begin{equation}
  \{\overline{B_\alpha}^{\,\KerrSpace}\}_\alpha,
  \qquad
  \{\jmath_\beta(\mathcal X_{h_\beta})\}_\beta
 \end{equation}
 are locally finite, and distinct block closures meet only in a common
 crossing image.  Consequently a continuous curve changes block only
 at $r=r_-$ or $r=r_+$, and a compact set meets only finitely many
 blocks.

 \item The images of the block and crossing embeddings cover
 $\KerrSpace$.  Two model points have the same image if and only if they
 are the same point of one model or are paired by
 \eqref{eq:def-exact-sector-identity}.  The atlas $\mathfrak A$ is
 generated by the Boyer-Lindquist and polar-cap charts on the blocks,
 the double-null and polar-cap charts on the crossing models, the
 stationary-axial actions
 \eqref{eq:graph-block-action}-\eqref{eq:graph-crossing-action}, the
 sector maps \eqref{eq:app-sector-map}-\eqref{eq:app-sector-inverse},
 and their restrictions.  Thus every regular disc and horizon point is
 included, and neither a partial horizon gluing nor an additional
 identification is allowed.
\end{enumerate}
\end{definition}

\begin{proposition}[Canonical global Kerr fields]
\label{prop:canonical-kerr-fields}
The marked block and crossing atlas of every member of $\CarterClass$
determines unique commuting analytic Killing fields $T$ and $\Phi$ and
a unique analytic Killing tensor $\mathsf K$ with the following
properties.  On every Boyer-Lindquist block,
\begin{equation}
 T=\partial_t,\qquad \Phi=\partial_\phi,
\end{equation}
the $\Phi$-flow has period $2\pi$ away from the axis, and $\mathsf K$ is
normalized so that
$\mathsf K(\dot\gamma,\dot\gamma)=\mathcal Q$ for every unit timelike
geodesic. 
\end{proposition}

\begin{proof}
On a crossing model $\mathcal X_h$, define
\begin{equation}\label{eq:canonical-crossing-fields}
 T_h\equiv-\kappa_hU\partial_U+\kappa_hV\partial_V
 -\Omega_h\partial_\varphi,
 \qquad \Phi_h\equiv\partial_\varphi,
\end{equation}
where $\partial_\varphi$ is the smooth axial rotation field on
$\mathbb S^2$.  Differentiating \eqref{eq:app-sector-map} gives
\begin{equation}
 (\Xi_h^{\epsilon\eta})_*\partial_t=T_h,
 \qquad
 (\Xi_h^{\epsilon\eta})_*\partial_\phi=\Phi_h.
\end{equation}
The actions \eqref{eq:graph-block-action}-
\eqref{eq:graph-crossing-action} push these fields to themselves.
They therefore glue to analytic Killing fields.  Their bracket
vanishes on the dense block region, hence everywhere, and the axial
period is that of the standard sphere rotation.

The invariant block and crossing formulas
\eqref{eq:app-ky-form}, \eqref{eq:app-ky-crossing} similarly give a
global analytic Killing-Yano form $f$.  Its square
\begin{equation}
 \mathsf K_{\alpha\beta}
 \equiv f_{\alpha\lambda}f_\beta{}^\lambda
\end{equation}
gives an analytic Killing tensor.  The computation
\eqref{eq:app-carter-normalization} proves the asserted normalization.

Uniqueness follows from agreement on the dense set $\{\Delta\ne0\}$.
\end{proof}

\begin{definition}[Marked isometry]\label{def:marked-isometry}
Let
\begin{equation}
 \mathbf K\equiv(\KerrSpace,g,r,\mu,\mathfrak A),
 \qquad
 \mathbf K'\equiv(\KerrSpace',g',r',\mu',\mathfrak A')
\end{equation}
be members of $\CarterClass$.  A \emph{marked isometry}
$F:\mathbf K\to\mathbf K'$ is an analytic isometry of the underlying
Lorentzian manifolds such that
\begin{equation}
 r'\circ F=r,\qquad \mu'\circ F=\mu,\qquad
 F_*T=T',\qquad F_*\Phi=\Phi',\qquad F^*\mathsf K'=\mathsf K,
\end{equation}
and $F$ carries the marked block and crossing families onto those of
$\mathbf K'$, preserving the horizon type and the full-collar
incidence data.  It therefore induces an isomorphism of the underlying
incidence multigraphs.  At each crossing the induced port map is either
the identity or the simultaneous involution
\begin{equation}
 Q_{\epsilon\eta}\longmapsto Q_{-\epsilon,-\eta};
\end{equation}
the latter is called a discrete port relabelling.  No preservation of
the chosen time orientation is included in this definition unless it
is stated separately.  For a fixed port-labelled graph, the induced
vertex and edge maps and every port map are required to be the
identity; in particular no discrete port relabelling is then allowed.
\end{definition}

\begin{lemma}[Rigidity of marked isometries]
\label{lem:marked-isometry-rigidity}
Let $F:\mathbf K\to\mathbf K'$ be a marked isometry.  After relabelling
the pieces by its induced graph isomorphism, the restriction of $F$ to
each model block is a stationary-axial translation
$\mathscr S_{\boldsymbol b}$, and its restriction to a crossing model
of type $h$ is
\begin{equation}
 \widehat{\mathscr S}_{h,\boldsymbol b}
 \quad\hbox{or}\quad
 \mathscr J_h\circ
 \widehat{\mathscr S}_{h,\boldsymbol b},
 \qquad
 \mathscr J_h(U,V,\omega)\equiv(-U,-V,\omega).
\end{equation}
The second alternative simultaneously exchanges the two ports on
each radial side of the crossing.  After any discrete port
relabellings have been chosen, two parameter families on the resulting
fixed port-labelled graph define
marked-isometric realizations if and only if there are vertex
parameters $\boldsymbol b_v\in\mathsf A$ for which
\begin{equation}\label{eq:marked-isometry-gauge}
 \boldsymbol\lambda_e'
 =\boldsymbol b_{c(e)}+\boldsymbol\lambda_e
  -\boldsymbol b_{b(e)}.
\end{equation}
Consequently the signed holonomy in
\eqref{eq:graph-holonomy} is a complete continuous invariant for a
fixed port-labelled cycle, up to sign when its orientation is
reversed.
\end{lemma}

\begin{proof}
Consider first a block and work away from the rotation axis.  Since
$F$ preserves $r$ and $\mu=\cos\theta$, while it pushes $T$ and $\Phi$
to their primed counterparts, its expression in Boyer-Lindquist
coordinates is
\begin{equation}\label{eq:marked-block-form}
 F(t,r,\theta,\phi)
 =\bigl(t+f(r,\theta),r,\theta,
        \phi+g(r,\theta)\bigr).
\end{equation}
Here $f$ is real-valued and $g$ is
$\mathbb R/2\pi\mathbb Z$-valued.  On a simply connected base
neighborhood choose a real lift of $g$.
Let $X$ be either $\partial_r$ or $\partial_\theta$.  The Kerr metric
has no orbit-base cross terms.  Applying the isometry identity to the
pairs $(X,T)$ and $(X,\Phi)$ gives
\begin{equation}
 \begin{pmatrix}g_{tt}&g_{t\phi}\\
                 g_{t\phi}&g_{\phi\phi}
 \end{pmatrix}
 \begin{pmatrix}Xf\\Xg\end{pmatrix}=0.
\end{equation}
On a block away from the axis,
\begin{equation}
 g_{tt}g_{\phi\phi}-g_{t\phi}^2
 =-\Delta\sin^2\theta\ne0.
\end{equation}
It follows that $df=dg=0$ for every local lift, so $f$ and the circle-valued
function $g$ are locally constant.  The base away from the axis is connected, including the punctured
base of $\mathcal B_1$.  Thus both functions are constant.  Continuity
in polar charts gives $F=\mathscr S_{\boldsymbol b}$ on the full block.

On a crossing quadrant, conjugation by the sector maps and
\eqref{eq:graph-equivariance} gives
$\widehat{\mathscr S}_{h,\boldsymbol b}$.  If the ports are exchanged,
first compose with $\mathscr J_h$, which preserves
$r,\mu,T,\Phi,\mathsf K$.  Analytic continuation from the quadrant
gives the stated form on the connected crossing model.

Fix the discrete port exchanges.  On an
edge $e$ from a block $b(e)$ to a crossing $c(e)$, compatibility of
the two piece maps with the exact gluing identity, followed by
\eqref{eq:graph-equivariance}, is precisely
\eqref{eq:marked-isometry-gauge}.  Conversely, vertex translations
satisfying that equation agree on every collar and therefore descend
to a marked isometry of the quotients.  Summing
\eqref{eq:marked-isometry-gauge} with the circuit signs cancels every
vertex parameter.  A spanning-tree gauge removes all edge parameters
except this sum, while reversal changes every circuit sign.  The final
assertion follows.
\end{proof}

\begin{proposition}[Existence of a horizon-saturated realization]\label{prop:carter-existence}
For every $M>0$ and $0<|a|<M$, the class $\CarterClass$ is nonempty.
\end{proposition}

\begin{proof}
Proposition~\ref{prop:explicit-carter-realization} in
Appendix~\ref{app:explicit-carter-atlas} constructs a marked member of
$\CarterClass$ directly from the Kerr formulas.
\end{proof}

In the subextremal case we fix an arbitrary marked tuple
$\mathbf K\in\CarterClass$ and write $(\KerrSpace,g_{M,a})$ for its
underlying spacetime.

\begin{lemma}[Block properness away from the ring]\label{lem:block-properness}
Let $B_\alpha$ be a block of a member of $\CarterClass$.  If $A>0$,
$\rho>0$, and $J\subset I_{j(\alpha)}$ is a compact interval, then the
image in $B_\alpha$ of
\begin{equation}
 \{\,|t|\le A,\ r\in J,\ \Sigma\ge\rho\,\}
\end{equation}
is compact in $\KerrSpace$.
\end{lemma}

\begin{proof}
The set is closed in the compact space
$[-A,A]\times J\times\mathbb S^2$ and avoids the deleted ring.
Its image under the block embedding is compact.
\end{proof}

\begin{remark}
Conditions (C2), (C4), and (C6) exclude an unfilled regular disc or
horizon sector, while (C5) prevents escape through an accumulation of
block gluings.  No maximality or uniqueness property is included in the
definition.
\end{remark}

\begin{proposition}[Analytically extendible Kerr subdomains]\label{prop:false-domains}
Assume $0<|a|<M$.  The domain of outer communication $r>r_+$ is analytically extendible across $r=r_+$.  The standard two-ended Kerr maximal globally hyperbolic development is analytically extendible across each component of $r=r_-$.  Neither proper subdomain is therefore $C^0$-inextendible.
\end{proposition}

\begin{proof}
The changes of variables
\begin{equation}\label{eq:ef-coordinates}
 v\equiv t+r_*(r),\qquad \bar\phi\equiv\phi+\phi_*(r),
 \qquad
 r_*'=\frac{r^2+a^2}{\Delta},\qquad \phi_*'=\frac{a}{\Delta}
\end{equation}
put the metric in a nondegenerate analytic form across either simple root when the corresponding ingoing branch is used; the outgoing branch is obtained by reversing the signs of $r_*$ and $\phi_*$.  Hence $r=r_+$ and $r=r_-$ are regular null hypersurfaces rather than singular boundaries.  The global placement of these charts is part of the Carter construction.
\end{proof}

\subsection{Extension regularity}

\begin{definition}[Lorentzian extensions]\label{def:extension}
Let $N$ be a connected smooth manifold equipped with a continuous
nondegenerate Lorentzian metric $h$.  Smoothness of $h$ will be stated
when it is needed.
\begin{enumerate}[label=(\alph*)]
 \item A $C^0$ extension is a connected smooth manifold $\widetilde N$ of the same dimension, an everywhere nondegenerate continuous Lorentzian metric $\widetilde h$, and a smooth isometric embedding $\iota:N\to\widetilde N$ whose image is a proper open subset.
 \item A $C^{1,1}_{\mathrm{loc}}$ extension is a $C^0$ extension for which, in every smooth coordinate chart, the components of $\widetilde h$ are $C^1$ and their first derivatives are locally Lipschitz.
 \item For $0<\alpha\le1$, a $C^{0,\alpha}_{\mathrm{loc}}$ extension is a
 $C^0$ extension for which the coordinate components of $\widetilde h$
 are locally H\"older continuous with exponent $\alpha$ in every smooth
 chart.  Thus $C^{0,1}_{\mathrm{loc}}$ is the locally Lipschitz class.
\end{enumerate}
The extension need not preserve a Killing field unless this is stated explicitly.
\end{definition}

An extension metric must be nondegenerate at every point. In
particular, a coordinate formula that is continuous but has zero
determinant at the ring does not define an extension in the sense of
Definition~\ref{def:extension}.

\section{Causal homotopies at a continuous boundary}\label{sec:c0geometry}

The purpose of this section is to bound the lengths in a fixed-endpoint
causal homotopy class near a continuous extension boundary. The
essential point is that a homotopy beginning inside a sufficiently
small coordinate diamond cannot leave the chart.

A future-causal path is locally Lipschitz, with future-causal or zero
tangent almost everywhere. Constant pauses are allowed. All homotopies
are jointly continuous and have fixed endpoints. In this section the
four-dimensional manifolds and extension embeddings are smooth; their
Lorentzian metrics may be continuous.

\subsection{Finite timelike diameter in a continuous coordinate chart}

\begin{definition}[Future timelike diameter]
For an open subset $U$ of a time-oriented Lorentzian manifold, let
$\mathscr T^+(U)$ denote the locally Lipschitz curves
$\sigma:[a,b]\to U$, $a<b<\infty$, with future-timelike tangent almost
everywhere. Define
\begin{equation}
 \tdiam^+(U,g)\equiv\sup_{\sigma\in\mathscr T^+(U)}L_g(\sigma).
\end{equation}
The supremum may be infinite. A sufficiently small coordinate
neighborhood admits a time orientation even if the ambient manifold
does not.
\end{definition}

\begin{proposition}[Finite local timelike diameter]\label{prop:finite-tdiam}
Every point of a continuous nondegenerate Lorentzian manifold has a
coordinate neighborhood $W$ with $\tdiam^+(W,g)<\infty$.
The diameters may be chosen to tend to zero along nested coordinate
boxes shrinking to the point.
\end{proposition}

\begin{proof}
Choose coordinates with $g(0)=\eta\equiv\operatorname{diag}(-1,1,1,1)$.
By continuity, on $W\equiv(-\rho,\rho)^4$ we have
$|g_{\mu\nu}-\eta_{\mu\nu}|<\delta$, where $\delta>0$ is arbitrarily
small, and $\partial_0$ belongs to a fixed time cone. For a causal vector
$v=(v^0,\vec v)$,
\begin{align}
 |(g-\eta)(v,v)|
 &\le\delta\left(\sum_{\mu=0}^3|v^\mu|\right)^2
 \le4\delta|v|_{\mathrm E}^2,\label{eq:component-bound}\\
 -|v^0|^2+|\vec v|^2
 &=\eta(v,v)\le4\delta\bigl(|v^0|^2+|\vec v|^2\bigr).
\end{align}
For $\delta<1/8$, this gives $|\vec v|\le C_0|v^0|$.
Choose the future cone with $v^0>0$. Another application of the
coefficient bound gives $\sqrt{-g(v,v)}\le C_1v^0$, and hence
\begin{equation}
 L_g(\sigma)\le C_1\int\dot\sigma^0\,ds
 =C_1\bigl(\sigma^0_{\mathrm{end}}-\sigma^0_{\mathrm{start}}\bigr)
 \le2C_1\rho
\end{equation}
for every $\sigma\in\mathscr T^+(W)$. The constants remain bounded
as $\rho,\delta\downarrow0$.
\end{proof}

\subsection{Local trapping in a continuous extension}

The following lemmas give the required local confinement. Their
hypotheses concern only the causal cones in a coordinate neighborhood.

\begin{lemma}[Precompact local causal diamonds]
\label{lem:precompact-local-diamond}
Let $(\widetilde M,\widetilde g)$ be a continuous nondegenerate
Lorentzian manifold and $q\in\widetilde M$. There are coordinate
neighborhoods $q\in V\Subset D\Subset W$ and a time orientation on $W$
such that every future-directed $W$-causal curve with endpoints
$P,Q\in V$ lies in $D$. In particular,
\begin{equation}\label{eq:fan-local-diamond-closure}
 \overline{J_W^+(P)\cap J_W^-(Q)}\Subset W,
\end{equation}
where the closure is taken in $W$.
\end{lemma}

\begin{proof}
Use a coordinate box centered at $q$ as in
Proposition~\ref{prop:finite-tdiam}. Its nonzero future-causal vectors
satisfy
\begin{equation}\label{eq:fan-cone-speed}
 v^0>0,\qquad |\vec v|\le C_0v^0.
\end{equation}
The strict inequality holds first for timelike vectors, then for causal
vectors by closure and the speed bound. Take $V$ to be the concentric
box of radius $r$. Along a $W$-causal curve from $P$ to $Q$ in $V$,
$x^0$ is nondecreasing and changes by at most $2r$. Thus every point
has $|x^0|\le r$ and spatial distance at most
$(\sqrt3+2C_0)r$ from $q$. Choose $r$ small enough that this closed
set lies in a box $D$ with $V\Subset D\Subset W$.
Every point of the diamond belongs to a concatenated $W$-causal curve
from $P$ to $Q$, so its closure lies in the same closed set. This proves
\eqref{eq:fan-local-diamond-closure}.
\end{proof}

\begin{lemma}[Fixed-endpoint causal-homotopy trapping]
\label{lem:causal-homotopy-trapping}
Let $(N,h)$ have a continuous time-oriented Lorentzian metric and a
$C^0$ extension $(\widetilde M,\widetilde g,\iota)$. Let $p,q\in N$
and let $W$ be a coordinate neighborhood whose time orientation agrees
at $\iota(p)$ with the pushed-forward future cone of $N$. Suppose
\begin{equation}\label{eq:fan-trapping-diamond-hypothesis}
 K\equiv\overline{J_W^+(\iota(p))\cap J_W^-(\iota(q))}\Subset W.
\end{equation}
If $\Gamma:[0,1]^2\to N$ is a continuous homotopy through future-causal
paths from $p$ to $q$ and
$\iota\circ\Gamma(0,[0,1])\subset W$, then
$\iota\circ\Gamma([0,1]^2)\subset W$.
\end{lemma}

\begin{proof}
Set
\begin{equation}
 \mathcal A\equiv\{\lambda\in[0,1]:
 \iota\circ\Gamma(\lambda,[0,1])\subset W\}.
\end{equation}
It is nonempty and open, since each contained slice has compact image
and the homotopy is continuous. Every such slice is $W$-future causal:
the two time orientations agree at $\iota(p)$, and their relative sign
is locally constant on the component of $\iota(N)\cap W$ containing
the slice. Its image therefore lies in $K$.
If $\lambda_n\in\mathcal A$ and $\lambda_n\to\lambda$, continuity
gives $\iota(\Gamma(\lambda,t))\in K\subset W$ for every $t$.
Thus $\mathcal A$ is also closed and equals $[0,1]$.
This argument includes constant slices.
\end{proof}

\subsection{Causal homotopy lengths at a continuous boundary}

Let $(N,h)$ have a continuous time-oriented Lorentzian metric, and let
$\gamma:[0,\Lambda)\to N$ be locally Lipschitz and future causal.
For $s<t<\Lambda$, denote by $\mathscr H_\gamma(s,t)$ the future-causal
paths from $\gamma(s)$ to $\gamma(t)$ connected to
$\gamma|_{[s,t]}$ by a fixed-endpoint causal homotopy. Set
\begin{equation}\label{eq:causal-class-length}
 \mathcal L_\gamma(s,t)\equiv
 \sup_{\sigma\in\mathscr H_\gamma(s,t)}L_h(\sigma).
\end{equation}
We now combine the homotopy and diameter argument of
\cite[Section~2]{GallowayLing2017} with the local confinement principle
of \cite[Lemma~2.3]{Guenther2021}. The resulting coordinate-time
estimate shows that the supremum of the lengths tends uniformly to
zero at the boundary.

\begin{theorem}[Causal homotopy lengths at a continuous boundary]
\label{thm:causal-length-criterion}
Let $(N,h)$ have a continuous time-oriented Lorentzian metric, and let
$\gamma:[0,\Lambda)\to N$ be locally Lipschitz and future causal.
Suppose that a $C^0$ extension $(\widetilde N,\widetilde h,\iota)$
satisfies $\iota(\gamma(t))\to q\in\partial\iota(N)$ as
$t\uparrow\Lambda$. There are a coordinate time $x^0$ near $q$,
a constant $C>0$, and $s_0<\Lambda$ such that
\begin{equation}\label{eq:causal-class-local-bound}
 \mathcal L_\gamma(s,t)
 \le C\bigl[x^0(\iota(\gamma(t)))
             -x^0(\iota(\gamma(s)))\bigr]
 \qquad(s_0<s<t<\Lambda).
\end{equation}
Consequently,
\begin{equation}\label{eq:causal-class-vanishing}
 \lim_{u\uparrow\Lambda}
 \sup_{u<s<t<\Lambda}\mathcal L_\gamma(s,t)=0.
\end{equation}
For $s_0<s<t<\Lambda$, no nonconstant future-causal closed curve $c$
based at $\gamma(t)$ can satisfy
\begin{equation}\label{eq:closed-curve-boundary}
 \gamma|_{[s,t]}\simeq_c\gamma|_{[s,t]}*c.
\end{equation}
Thus an extension endpoint is impossible if some $\eta>0$ satisfies
$\mathcal L_\gamma(s,t)\ge\eta$ for a pair $s<t$ in every terminal
interval, or if a nonconstant closed curve can be inserted as in
\eqref{eq:closed-curve-boundary} arbitrarily late along $\gamma$.
\end{theorem}

\begin{proof}
Choose $q\in V\Subset D\Subset W$ as in
Lemma~\ref{lem:precompact-local-diamond}, with the cone bounds of
Proposition~\ref{prop:finite-tdiam}. Orient $W$ consistently with the
terminal curve and choose $s_0$ so that
$\iota(\gamma((s_0,\Lambda)))\subset V$.
The reference segments lie in $V$, and their $W$-diamonds have
precompact closure. Lemma~\ref{lem:causal-homotopy-trapping} confines
every curve in $\mathscr H_\gamma(s,t)$ to $W$.
The cone comparison gives
$\sqrt{-\widetilde h(v,v)}\le Cv^0$ for every future-causal vector,
including zero. Integration gives
\begin{equation}
 L_h(\sigma)=L_{\widetilde h}(\iota\circ\sigma)
 \le C\bigl[x^0(\iota(\gamma(t)))
             -x^0(\iota(\gamma(s)))\bigr],
\end{equation}
and taking the supremum proves \eqref{eq:causal-class-local-bound}.
For $u<s<t$, the right-hand side is at most
$C[x^0(q)-x^0(\iota(\gamma(u)))]$, which tends to zero.

If \eqref{eq:closed-curve-boundary} held, $\iota\circ c$ would lie
in $W$. Its coordinate time would be nondecreasing with equal endpoint
values, hence constant. The cone-speed bound would make its spatial
coordinates constant as well. Local Lipschitz regularity gives absolute
continuity on its compact parameter interval, so $c$ would be constant.
This proves the closed-curve assertion and the stated obstructions.
\end{proof}

\subsection{Passage to covering spacetimes}

\begin{lemma}[Causal homotopy lengths under coverings]
\label{lem:covering-causal-length}
Let $\pi:(\widehat N,\widehat h)\to(N,h)$ be a smooth covering with
$\widehat h=\pi^*h$ and the lifted time orientation. If
$\widehat\gamma:[0,\Lambda)\to\widehat N$ is locally Lipschitz and
future causal, and $\gamma\equiv\pi\circ\widehat\gamma$, projection is
a length-preserving bijection
\begin{equation}
 \mathscr H_{\widehat\gamma}(s,t)
 \longrightarrow\mathscr H_\gamma(s,t),\qquad
 \widehat\sigma\longmapsto\pi\circ\widehat\sigma.
\end{equation}
In particular,
\begin{equation}\label{eq:covering-causal-length}
 \mathcal L_{\widehat\gamma}(s,t)=\mathcal L_\gamma(s,t)
 \qquad(0\le s<t<\Lambda).
\end{equation}
\end{lemma}

\begin{proof}
Projection preserves causal homotopies and length. Conversely,
parametrize the paths on $[0,1]$ and lift a homotopy from
$\gamma|_{[s,t]}$ using its reference lift
$\widehat\gamma|_{[s,t]}$ and the covering homotopy theorem
\cite[Section~1.3]{Hatcher2002}. The endpoint edges remain constant,
being continuous maps into discrete fibers. Inverse covering charts
are time-orientation-preserving local isometries, so every lifted slice
is future causal and has the original length. Uniqueness of path
lifting gives injectivity.
\end{proof}

\section{Causal homotopies at the Kerr ring}
\label{sec:two-sheet-fan}

In this section we construct causal homotopies with fixed endpoints
whose curves have arbitrarily large proper time near the Kerr ring.
Together with Section~\ref{sec:c0geometry}, this rules out a continuous
extension through the ring. The construction takes place entirely in
the smooth Kerr block and uses its Killing fields only there. We
parametrize a collapsed segment by a constant curve.

\subsection{Transverse timelike detours}

On a Boyer-Lindquist region with $\Delta>0$, put
\begin{equation}\label{eq:fan-alpha-beta-q}
 \alpha\equiv dt-a\sin^2\theta\,d\phi,\qquad
 \beta\equiv(r^2+a^2)d\phi-a\,dt,\qquad
 \mathfrak q\equiv\frac{\Sigma}{\Delta}\,dr^2+\Sigma\,d\theta^2.
\end{equation}
Lemma~\ref{lem:coframe} gives
\begin{equation}\label{eq:fan-coframe-decomposition}
 g_{M,a}=-\frac{\Delta}{\Sigma}\alpha^2
 +\frac{\sin^2\theta}{\Sigma}\beta^2+\mathfrak q,
 \qquad g^{-1}_{M,a}(\alpha,\alpha)=-\frac{\Sigma}{\Delta}<0.
\end{equation}
It follows that $\alpha$ is timelike on a sufficiently small punctured
ring collar. Since this collar is connected, a choice of time
orientation determines a unique $\varepsilon\in\{-1,1\}$ such that
a nonzero causal vector $V$ is future-directed if and only if
\begin{equation}\label{eq:fan-future-alpha}
 \varepsilon\alpha(V)>0.
\end{equation}
Fix $c>1$. The stationary-axial translations
\begin{equation}\label{eq:fan-stationary-axial-map}
 \mathscr S_{(\xi,\eta)}(t,r,\theta,\phi)
 \equiv(t+\xi,r,\theta,\phi+\eta),
 \qquad (\xi,\eta)\in\mathbb R\times\mathbb S^1,
\end{equation}
preserve the metric, $\alpha$, and the future cone.
On the negative equator, the vector
\begin{equation}\label{eq:fan-Y-minus}
 Y_-\equiv-\varepsilon\operatorname{sgn}(a)\Phi
\end{equation}
is future timelike for small $\rho>0$, since
\begin{equation}
 \varepsilon\alpha(Y_-)=|a|>0,\qquad
 g(Y_-,Y_-)=\rho^2+a^2-\frac{2Ma^2}{\rho}<0.
\end{equation}

\begin{lemma}[Timelike detours near the ring]
\label{lem:radial-attachment}
There are $r_0,C>0$ with the following property. Let $x$ lie in a
marked inner block with $R\equiv\sqrt{\Sigma(x)}<r_0$, and fix
$A_0\in(0,\pi)$. For every $0<\rho<R$ there is a continuous family
of future-causal paths $P_{x,\rho,\lambda}$, $0\le\lambda\le1$, from
$x$ to $\mathscr S_{h_{x,\rho,\lambda}}(x)$ such that:
\begin{enumerate}[label=\textup{(\roman*)}]
 \item $P_{x,\rho,0}$ is constant and $h_{x,\rho,0}=0$;
 \item $P_{x,\rho,1}$ contains a negative-equatorial axial arc at
 $r=-\rho$ of angular width $A_0$;
 \item every nonconstant piece is strictly future timelike, all paths
 lie in $0<\Sigma\le R^2$, and
 \begin{equation}\label{eq:radial-attachment-shift}
  \sup_{0\le\lambda\le1}|h_{x,\rho,\lambda}|\le C(R+A_0).
 \end{equation}
\end{enumerate}
Angular displacements have continuous real representatives, and
$|(\xi,\eta)|\equiv|\xi|+|\eta|$. The constants are independent of
$x,\rho,A_0$.
\end{lemma}

\begin{proof}
Put $y\equiv a\cos\theta$. Choose $0<r_0<|a|/2$ so that $\Delta>0$
for $|r|\le r_0$ and $Y_-$ is future timelike at $r=-\rho$ whenever
$0<\rho<r_0$. On $0<r^2+y^2<r_0^2$, both $\Delta(r)$ and $a^2-y^2$
have positive upper and lower bounds, and
\begin{equation}\label{eq:transverse-metric}
 \Sigma=r^2+y^2,\qquad
 \mathfrak q=\Sigma\left(\frac{dr^2}{\Delta}
                  +\frac{dy^2}{a^2-y^2}\right).
\end{equation}
For a transverse velocity $W\equiv(\dot r,\dot y)$, define
\begin{equation}
 k(W)\equiv\left(\frac{\dot r^2}{\Delta}
                  +\frac{\dot y^2}{a^2-y^2}\right)^{1/2},
\end{equation}
and choose its lift by
\begin{equation}\label{eq:transverse-lift}
 \beta(V)=0,\qquad
 \alpha(V)=\varepsilon c\frac{\Sigma}{\sqrt\Delta}k(W).
\end{equation}
Then
\begin{align}
 g(V,V)&=-(c^2-1)\Sigma k(W)^2,\nonumber\\
 (\dot t,\dot\phi)
 &=\varepsilon c\frac{(r^2+a^2,a)}{\sqrt\Delta}k(W),
 \qquad |\dot t|+|\dot\phi|\le C_0|W|_{\mathbb R^2}.
 \label{eq:transverse-orbit-bound}
\end{align}
Every nonconstant transverse segment therefore has a future-timelike
lift whose stationary-axial variation is bounded by $C_0$ times
its Euclidean length. The same estimate holds when the transverse
segment is traversed in the opposite direction.

Join the transverse point of $x$ to radius $\rho$ by a radial
segment, and then by a circular arc of angle at most $\pi$ to
$(-\rho,0)$; when the two choices have equal length, either may be
used. The resulting path $w$ lies in
$0<r^2+y^2\le R^2$ and has length at most $(1+\pi)R$.
For $0\le\lambda\le1/2$, lift a progressively longer initial portion
of $w$ and its reverse to future-timelike paths. For
$1/2\le\lambda\le1$, retain the two full lifts and insert between
them a $Y_-$ arc of parameter length $(2\lambda-1)A_0$. Translate
the return lift so that the pieces join.
The total shift is at most $2C_0(1+\pi)R+A_0$.
Using fixed parameter intervals and constant curves for collapsed
pieces makes the family jointly continuous, also at $\lambda=1/2$.
The translations preserve $\Sigma$ and the future cone. The axial
arc in the final path has length
\begin{equation}
 A_0\sqrt{\frac{2Ma^2}{\rho}-a^2-\rho^2}.
\end{equation}
\end{proof}

\subsection{Fixed-endpoint lengths at the ring}

\begin{theorem}[Infinite lengths near the Kerr ring]
\label{thm:fixed-endpoint-infinite-length}
Let $\gamma:[0,\Lambda)\to\KerrSpace$ be a future-directed $C^1$
timelike curve in a member of $\CarterClass$, with
$\Sigma(\gamma(t))\to0$ as $t\uparrow\Lambda$. For every $s<\Lambda$
there is $t_0(s)\in(s,\Lambda)$ such that
\begin{equation}\label{eq:fixed-endpoint-infinite-length}
 \mathcal L_\gamma(s,t)=+\infty\qquad(t_0(s)<t<\Lambda).
\end{equation}
The curve may be arbitrary within the stated timelike class. Given
$\delta>0$, $t_0$ may be chosen so that all varying portions of the
homotopies lie in one marked inner block with $0<\Sigma<\delta^2$.
\end{theorem}

\begin{proof}
Fix $s<\Lambda$ and $\delta>0$. For all sufficiently late parameters,
$|r|<r_-/2$. The connected tail therefore avoids the horizons and lies
in one marked inner block.
Choose $s<s_w<s_0<\Lambda$ so that $\gamma([s_w,\Lambda))$ lies
there with $0<\Sigma<\min\{\delta,r_0\}^2$. Put
$w\equiv\gamma(s_w)$, $z_0\equiv\gamma(s_0)$, and choose smooth
$b:[s_w,s_0]\to[0,1]$ with $b(s_w)=0$, $b(s_0)=1$.
For a small real stationary-axial displacement $h\equiv(\xi,\eta)$,
set
\begin{equation}\label{eq:timelike-endpoint-variation}
 \eta_h(v)\equiv\mathscr S_{-b(v)h}(\gamma(v)),
 \qquad s_w\le v\le s_0.
\end{equation}
With $K_h\equiv\xi T+\eta\Phi$,
\begin{equation}
 \dot\eta_h(v)
 =d\mathscr S_{-b(v)h}\bigl(\dot\gamma(v)-b'(v)K_h\bigr).
\end{equation}
Since the reference segment is compact and strictly timelike, the
openness of the future cone and its invariance under $\mathscr S$ give
a neighborhood $\mathcal U$ of zero for which $\eta_h$ is future
timelike whenever $h\in\mathcal U$. This curve joins $w$ to
$\mathscr S_h^{-1}(z_0)$, agrees with the reference segment when
$h=0$, and stays in the same collar.

Choose $A_0>0$ and then $R_0>0$ sufficiently small that
\eqref{eq:radial-attachment-shift} places all the shifts in
$\mathcal U$ whenever $\sqrt{\Sigma(x)}<R_0$ and
$0<\rho<\sqrt{\Sigma(x)}$.
Choose $t_0>s_0$ with
$\sqrt{\Sigma(\gamma(t))}<\min\{R_0,r_0,\delta\}$ for $t>t_0$.
Fix such $t$, put $q\equiv\gamma(t)$, $R\equiv\sqrt{\Sigma(q)}$,
and take $0<\rho<R$. Abbreviate the paths and shifts of
Lemma~\ref{lem:radial-attachment} by $P_\lambda,h_\lambda$. Then
\begin{equation}\label{eq:fixed-endpoint-radial-concatenation}
 \gamma|_{[s,s_w]}*\eta_{h_\lambda}
 *\mathscr S_{h_\lambda}^{-1}(\gamma|_{[s_0,t]})
 *\mathscr S_{h_\lambda}^{-1}(P_\lambda)
\end{equation}
is a future-causal path from $\gamma(s)$ to $q$. The adjacent pieces
join, and the final endpoint is
$\mathscr S_{h_\lambda}^{-1}\mathscr S_{h_\lambda}(q)=q$.
At $\lambda=0$, this is the reference segment with a constant pause,
which can be removed by a homotopy of monotone reparametrizations.
We have therefore obtained a causal homotopy with fixed endpoints.
The first segment is fixed throughout; the other pieces lie in the
chosen collar. At $\lambda=1$, the negative-equatorial arc has length
$A_0\sqrt{2Ma^2/\rho-a^2-\rho^2}$. Since $A_0$ is independent of
$\rho$, these lengths tend to infinity as $\rho\downarrow0$.
\end{proof}

For sufficiently small $\delta>0$, define the full punctured ring collar
\begin{equation}\label{eq:punctured-ring-collar}
 \mathcal N_\delta\equiv
 \{(t,r,\theta,\phi):0<r^2+a^2\cos^2\theta<\delta^2\}
 \subset\mathcal B_1,
\end{equation}
where the stationary and axial coordinates have their full ranges.
The preceding construction applies within this collar to every $C^1$
timelike curve with $\Sigma\to0$. By
Theorem~\ref{thm:causal-length-criterion}, such a curve cannot end in
a continuous extension. The outer boundary $\Sigma=\delta^2$ is
smoothly extendible, so the obstruction is confined to the ring.
Since $\Delta(0)=a^2>0$ and the negative-equatorial axial norm tends
to $-\infty$, the argument applies to every $M>0$, $a\ne0$, including
the extremal and overextremal models considered below.

\subsection{Conformal ring obstruction}

We now extend the ring obstruction to positive continuous conformal
rescalings. At one transverse point, two independent
stationary-axial directions are future timelike; at a second point,
their negatives are future timelike.
Their orbit segments cancel the displacement of the connecting paths
and give a closed causal curve that can be inserted into the reference
curve by a causal homotopy. Repeating this closed curve yields
arbitrarily large proper time for the rescaled metric. The construction
uses the full stationary range and can be carried out at arbitrarily
small transverse radii.

\Needspace{12\baselineskip}
\begin{theorem}[Conformal ring obstruction]
\label{thm:conformal-class-ring}
Let $g\equiv g_{M,a}$, $M>0$, $a\ne0$, on a full punctured ring collar
$\mathcal N_{\delta_0}$ with $t\in\mathbb R$ and $\phi\in\mathbb S^1$.
Let $\pi:\widehat{\mathcal N}\to\mathcal N_{\delta_0}$ be a connected
covering, let $\widehat\Omega>0$ be continuous, and put
$\widehat g\equiv\widehat\Omega^2\pi^*g$ with the lifted time
orientation. Suppose that a future-directed $C^1$ timelike curve
$\gamma:[0,\Lambda)\to\widehat{\mathcal N}$ satisfies
$\Sigma\circ\pi\circ\gamma\to0$.

For every $s<\Lambda$, there is $t_0\in(s,\Lambda)$ such that for
each $t\in(t_0,\Lambda)$ a future-causal closed curve $c$ based at
$\gamma(t)$ satisfies
\begin{equation}\label{eq:conformal-closed-curve-insertion}
 L_{\widehat g}(c)>0,\qquad
 \gamma|_{[s,t]}\simeq_c\gamma|_{[s,t]}*c.
\end{equation}
Consequently,
\begin{equation}\label{eq:conformal-class-infinite-lengths}
 \mathcal L^{\widehat g}_\gamma(s,t)=+\infty
 \qquad(t_0<t<\Lambda),
\end{equation}
and $\gamma$ has no endpoint in a $C^0$ extension.
For every $0<\delta<\delta_0$, the varying portions of the homotopies
can be confined to $\pi^{-1}(\mathcal N_\delta)$.
The factor $\widehat\Omega$ may vary arbitrarily on the cover,
without symmetry, asymptotic restrictions, or descent to the base.
\end{theorem}

\begin{proof}
We first carry out the construction on the base with the Kerr
metric. Put $A\equiv|a|$
and choose $\varepsilon$ by \eqref{eq:fan-future-alpha}.
For sufficiently small $\ell>0$, set
\begin{gather}\label{eq:closed-orbit-parameters}
 d\equiv\frac{\ell^2}{4A},\qquad
 \rho\equiv\min\left\{1,\frac{M}{4A}\right\}d,\qquad
 d_1\equiv d,\quad d_2\equiv\frac32d,\nonumber\\
 K_i\equiv\varepsilon\bigl((A-d_i)T+\operatorname{sgn}(a)\Phi\bigr).
\end{gather}
At $X_+\equiv(0,\ell)$ in the $(r,y)$-plane,
\begin{align}\label{eq:closed-orbits-positive}
 \varepsilon\alpha(K_i)&=4d-d_i>0,\nonumber\\
 g(K_1,K_1)&=-\frac{\ell^2}{2}-d^2<0,
 &g(K_2,K_2)&=-\frac{\ell^2}{4}-\frac94d^2<0.
\end{align}
At $X_-\equiv(-\rho,0)$,
\begin{align}\label{eq:closed-orbits-negative}
 \varepsilon\alpha(-K_i)&=d_i>0,\nonumber\\
 g(-K_i,-K_i)&=2Ad_i-d_i^2+\rho^2-\frac{2Md_i^2}{\rho}.
\end{align}
Since $\rho\le d$ and $\rho\le Md/(4A)$, the last expression is
at most $-6Ad$ for $i=1$ and $-15Ad$ for $i=2$.
Thus both $K_i$ are future timelike at $X_+$, whereas both $-K_i$
are future timelike at $X_-$.

Let $x$ have transverse radius $R$. Choose $0<\ell<R/4$ with
$\rho<\ell$. Radial segments and circular arcs join its transverse
point to $X_+$ and then $X_-$ by a path $w$ of length at most $CR$
in $0<r^2+y^2\le R^2$. Lift $w$ and its reverse by
\eqref{eq:transverse-lift}. Their total stationary-axial displacement
$H\in\mathbb R^2$ satisfies $|H|\le C_1R$.
As in Lemma~\ref{lem:radial-attachment}, progressively longer initial
portions of $w$, followed by their reverses, deform a constant path
into this path. The same estimate holds for every intermediate
displacement.

The unit $K_i$ orbits have linearly independent real displacements
$v_i\equiv\varepsilon(A-d_i,\operatorname{sgn}(a))$. Choose
$c_1,c_2\in\mathbb R$ satisfying
\begin{equation}\label{eq:closed-orbit-cancellation}
 c_1v_1+c_2v_2=-H,\qquad
 c_i^+\equiv\max\{c_i,0\},\quad c_i^-\equiv\max\{-c_i,0\}.
\end{equation}
Once the deformation includes the whole path $w$, insert the $K_i$
orbits of parameter lengths $uc_i^+$ at the first visit to $X_+$ and
the $-K_i$ orbits of lengths $uc_i^-$ at $X_-$, for $0\le u\le1$.
Translate each subsequent piece so that the pieces join. The resulting
paths are causal, and their endpoint displacement is
\begin{equation}\label{eq:closed-orbit-endpoint-displacement}
 H+u\sum_{i=1}^2(c_i^+-c_i^-)v_i=(1-u)H.
\end{equation}
The cancellation takes place in the real angular coordinate before
projection to $\mathbb S^1$. Parametrizing the pieces on fixed
intervals, with constant curves where a piece has collapsed, gives a
jointly continuous family $P_\lambda$ from $x$ to
$\mathscr S_{h_\lambda}(x)$, with $P_0$ constant,
$h_0=h_1=0$, and $|h_\lambda|\le C_1R$.
The final path $D\equiv P_1$ is closed. Its nonconstant transverse
pieces are strictly timelike, so its Kerr length is positive. Throughout
the construction, $0<\Sigma\le R^2$.

Fix $s<\Lambda$ and a smaller collar. An earlier compact timelike
segment of $\pi\circ\gamma$, after $s$ and inside that collar,
admits \eqref{eq:timelike-endpoint-variation} for shifts in a fixed
neighborhood of zero. For all sufficiently late $t$, $C_1R$ at
$x\equiv\pi(\gamma(t))$ lies in this neighborhood. Formula
\eqref{eq:fixed-endpoint-radial-concatenation} gives a causal homotopy
from $(\pi\circ\gamma)|_{[s,t]}$ to that segment followed by $D$.
Since $h_1=0$, the original initial segment and tail agree with those
of the final curve. Every varying piece stays in the prescribed
collar.

Lift this homotopy starting from the reference segment
$\gamma|_{[s,t]}$. Its endpoint curves are constant, since the fibers
are discrete. The final path is therefore $\gamma|_{[s,t]}*c$, where
$c$ is a \emph{closed} lift of $D$.
Conformal rescaling preserves the future cones, so the lifted homotopy
is $\widehat g$-causal. Compactness of $c$ and positivity of
$\widehat\Omega$ give
$L_{\widehat g}(c)\ge(\min_c\widehat\Omega)L_g(D)>0$.
This establishes \eqref{eq:conformal-closed-curve-insertion} for
every factor on the cover, including those that do not descend to
the base.

If $a\simeq_c a*c$, appending $c^n$ to each slice gives
$a*c^n\simeq_c a*c^{n+1}$. Concatenation of the homotopies shows
that every $a*c^n$ belongs to the original class. Its length is
$L_{\widehat g}(a)+nL_{\widehat g}(c)\to\infty$.
This proves \eqref{eq:conformal-class-infinite-lengths} with the same
localization. The closed-curve assertion of
Theorem~\ref{thm:causal-length-criterion} excludes an extension endpoint.
\end{proof}

\begin{remark}
The homotopy $a\simeq_c a*c$ varies an earlier timelike segment of
$a$. It is not a deformation of $c$ through closed causal curves,
since the prefix $a$ cannot in general be cancelled in a causal homotopy;
see \cite[Example~6.1]{Guenther2021}.
\end{remark}

In particular, the theorem applies to arbitrary powers of $\Sigma$,
to factors such as $e^{-1/\Sigma}$, and to factors with unrestricted
dependence on the stationary and angular coordinates. Its conclusion
concerns approaches to the ring; global geodesic localization for the
rescaled metric is a separate question.

\begin{corollary}[Persistence under cone widening]
\label{thm:ring-cone-widening}
Let $\pi:\widehat{\mathcal N}\to\mathcal N_{\delta_0}$ be a connected
cover of a full punctured Kerr ring collar, with $M>0$, $a\ne0$.
Let $k$ be a continuous time-oriented Lorentzian metric on
$\widehat{\mathcal N}$ whose future causal cones contain those of
$\pi^*g_{M,a}$. Suppose that $\gamma:[0,\Lambda)\to\widehat{\mathcal N}$
is future-directed $C^1$ timelike for $\pi^*g_{M,a}$ and
$\Sigma\circ\pi\circ\gamma\to0$.

For every $s<\Lambda$, there is $t_0\in(s,\Lambda)$ such that for
each $t\in(t_0,\Lambda)$ a nonconstant future-causal closed curve
$c$ based at $\gamma(t)$ satisfies
\begin{equation}\label{eq:cone-widening-insertion}
 \gamma|_{[s,t]}\simeq_{c,k}\gamma|_{[s,t]}*c,\qquad L_k(c)>0.
\end{equation}
Consequently $\mathcal L^k_\gamma(s,t)=+\infty$ for these $t$, and
$\gamma$ has no endpoint in a $C^0$ extension of
$(\widehat{\mathcal N},k)$. The varying portions can be confined to
the preimage of any prescribed smaller transverse collar.
The metric $k$ may vary on the cover without invariance or descent
to the base.
\end{corollary}

\begin{proof}
Apply the construction of Theorem~\ref{thm:conformal-class-ring}
with conformal factor one. Cone inclusion preserves causality of
every slice. It also preserves timelike vectors, since an interior
point of the Kerr future cone is interior to the larger $k$-cone.
This is the pointwise observation underlying
\cite[Proposition~3.2]{GarciaParradoSenovilla2003}, valid here for
continuous metrics. Hence the nonconstant Kerr-timelike piece of $c$
has positive $k$-length. Repetition yields arbitrarily large lengths, and
Theorem~\ref{thm:causal-length-criterion} excludes an endpoint. The
transverse localization follows from the same construction.
\end{proof}

For example,
\[
 k\equiv\Omega^2\bigl(\pi^*g_{M,a}-f(\pi^*\alpha)^2\bigr),
 \qquad \Omega>0,\quad f\ge0,
\]
has the required cones for arbitrary continuous $\Omega,f$ on the
cover: $\alpha$ is timelike, so subtracting its nonnegative square
preserves Lorentzian signature and widens the future cones. We next
extend the conclusion to Kerr-causal approaches.

\begin{corollary}[Causal approaches to the ring]
\label{cor:causal-ring-approaches}
Let $\pi:\widehat{\mathcal N}\to\mathcal N_{\delta_0}$ be a connected
cover of a full punctured Kerr ring collar, with $M>0$, $a\ne0$.
Let $k$ be a continuous time-oriented Lorentzian metric whose future
causal cones contain those of $\pi^*g_{M,a}$.
No future-directed, locally Lipschitz $\pi^*g_{M,a}$-causal curve
$\gamma:[0,\Lambda)\to\widehat{\mathcal N}$ with
$\Sigma\circ\pi\circ\gamma\to0$ has an endpoint in a $C^0$ extension
of $(\widehat{\mathcal N},k)$. In particular, this holds for every
positive continuous conformal factor and for null approaches.
\end{corollary}

\begin{proof}
Suppose $\iota\circ\gamma\to q$ in an extension. Choose
$q\in V\Subset D\Subset W$ as in
Lemma~\ref{lem:precompact-local-diamond}, with
$\iota(\gamma([u,\Lambda)))\subset V$ for some sufficiently late $u$.
Choose a short $C^1$ future Kerr-timelike path $\beta$ ending at
$\gamma(u)$, contained in $\iota^{-1}(V)$ and a smaller transverse
collar. Such a path exists because these sets and the Kerr timelike cone
are open. A compact interior subsegment of $\beta$ admits
\eqref{eq:timelike-endpoint-variation} for shifts in a fixed
neighborhood of zero.

For sufficiently late $t$, construct the closed curve at
$\pi(\gamma(t))$ using this subsegment for the correction.
The displacement is $O(\sqrt{\Sigma(\pi(\gamma(t)))})$ and lies in
the admissible neighborhood. The translated tail is locally Lipschitz and future causal, because
stationary-axial translations preserve the Kerr cones. Strict
timelikeness is required only on the earlier segment used for the
endpoint correction. Lifting gives a fixed-endpoint causal homotopy
from $a_t\equiv\beta*\gamma|_{[u,t]}$ to $a_t*c_t$, with $c_t$
nonconstant and closed. Every slice is also $k$-causal.
The reference path lies in $V$, so
Lemma~\ref{lem:causal-homotopy-trapping} confines the homotopy and
$\iota\circ c_t$ to $W$. Its local time and then all its spatial
coordinates must be constant, a contradiction.
\end{proof}

For a null approach, the auxiliary timelike initial segment permits
the endpoint variation. The conclusion for such an approach is the
absence of an extension endpoint; infinite length in its original
causal homotopy class is not asserted.

\section{Timelike geodesics at an extension boundary}\label{sec:boundary-tools}

\subsection{Existence of a geodesic ending at the boundary}

\begin{theorem}[Boundary geodesic]\label{thm:boundary-geodesic}
If a connected smooth time-oriented spacetime $(N,h)$ has a proper
$C^0$ extension, then, after reversing time orientation if necessary,
there are $0<\Lambda<\infty$, a future-inextendible unit-speed timelike
$h$-geodesic $\gamma:[0,\Lambda)\to N$, and
$q\in\partial\iota(N)$ such that
\begin{equation}
 \iota(\gamma(s))\longrightarrow q\qquad(s\uparrow\Lambda).
\end{equation}
The statement applies to $N$ without causality or global-hyperbolicity
restrictions.
\end{theorem}

\begin{proof}
Pass to the time-orientation double cover
$\pi:\widehat N\to\widetilde N$ of the extension. The orientation of
$N$ lifts $\iota$ to an isometric open embedding $\widehat\iota$.
Its image is proper in its connected component. Indeed, if the cover
is disconnected, each component maps diffeomorphically onto
$\widetilde N$; if connected, surjectivity of $\widehat\iota$ would
contradict injectivity of $\iota=\pi\circ\widehat\iota$.
Every boundary point of the lifted image projects to
$\partial\iota(N)$. Indeed, an evenly covered neighborhood inside
$\iota(N)$ has one sheet in the image and the other disjoint from its
closure. We may therefore work in the time-oriented cover.

The extension metric gives closed, sharp, convex future cones with
nonempty interior and continuous homogeneous length function
$F(v)\equiv\sqrt{-\widetilde h(v,v)}$. These are Lorentz-Finsler data
in the sense of \cite[Section~2]{MinguzziSuhr2019}, with $\iota(N)$ a
Lorentz-Finsler subspace. By \cite[Lemma~2.2]{MinguzziSuhr2019}, its
future or past boundary is nonempty. After reversing time orientation if
necessary, \cite[Theorem~2.3]{MinguzziSuhr2019} gives a future causal
maximizer in a relatively compact coordinate neighborhood
$\sigma:[0,1]\to\widetilde N$ satisfying
\begin{equation}
 \sigma([0,1))\subset\iota(N),\qquad
 \sigma(1)=q\in\partial\iota(N),\qquad 0<L_F(\sigma)<\infty.
\end{equation}
Pull the interior back to the original spacetime, remove constant
subarcs, and parametrize compact subarcs by auxiliary Riemannian
arclength. A sufficiently short subarc maximizes length in a convex
normal neighborhood. By the Lorentzian Gauss lemma, it is a causal
geodesic up to increasing reparametrization. The subarcs agree on
overlaps and hence have the same causal character. Since the total
length is positive, this character is timelike. We obtain a smooth
timelike pregeodesic. Proper time gives $[0,\Lambda)$ with
$\Lambda\equiv L_F(\sigma)<\infty$, and the parameter is affine because
\begin{equation}
 h(\dot\gamma,\dot\gamma)=-1,\qquad
 \nabla_{\dot\gamma}\dot\gamma=\kappa\dot\gamma
 \quad\Longrightarrow\quad
 0=\frac d{ds}h(\dot\gamma,\dot\gamma)=-2\kappa.
\end{equation}
If the geodesic continued through $\Lambda$ in $N$, it would have
both an interior limit and a boundary limit in the Hausdorff extension.
This is impossible.
For causal character of maximizers at Lipschitz regularity, see also
\cite{GrafLing2018}.
\end{proof}

\begin{remark}
By Theorem~\ref{thm:boundary-geodesic}, a proper extension gives an
incomplete timelike geodesic ending at its boundary. In
Section~\ref{sec:geodesics}, we determine the terminal behavior of
these geodesics in $\CarterClass$.
\end{remark}

\section{Regular horizons and the Carter tensor}\label{sec:regular-horizons}

We now describe the coordinates across the regular horizons and
verify that the Kerr Killing fields and Carter tensor extend through
them. The quotient construction using these coordinates is given in
Appendix~\ref{app:explicit-carter-atlas}.

\subsection{Analytic crossing coordinates at a simple horizon}

Fix $h\in\{r_-,r_+\}$, use the normalizations
\eqref{eq:horizon-normalizations}, and set
\begin{equation}\label{eq:app-horizon-data}
 A(r)\equiv r^2+a^2.
\end{equation}
Here $\kappa_h\ne0$, with $\kappa_h<0$ when $h=r_-$.  On each component
of $\mathbb R\setminus\{r_-,r_+\}$, define
\begin{align}
 r_\star(r)&\equiv r+\frac{r_+^2+a^2}{r_+-r_-}\log|r-r_+|
 -\frac{r_-^2+a^2}{r_+-r_-}\log|r-r_-|,
 \label{eq:app-tortoise}\\
 \phi_\star(r)&\equiv\frac{a}{r_+-r_-}
 \bigl(\log|r-r_+|-\log|r-r_-|\bigr).
 \label{eq:app-angle-primitive}
\end{align}
Then
\begin{equation}\label{eq:app-primitive-derivatives}
 r_\star'=\frac{A(r)}{\Delta(r)},\qquad
 \phi_\star'=\frac{a}{\Delta(r)}.
\end{equation}
The functions
\begin{equation}\label{eq:app-regular-parts}
 b_h(r)\equiv r_\star(r)-\frac{1}{2\kappa_h}\log|r-h|,
 \qquad
 \chi_h(r)\equiv\phi_\star(r)-\Omega_h r_\star(r)
\end{equation}
extend analytically across $r=h$: the logarithmic coefficients are
$1/(2\kappa_h)$ and $\Omega_h/(2\kappa_h)$, respectively.  Moreover,
\begin{equation}\label{eq:app-chi-derivative}
 \chi_h'(r)=\frac{a-\Omega_h A(r)}{\Delta(r)}
 =-\Omega_h\frac{r+h}{r-\bar h}.
\end{equation}

Define
\begin{equation}\label{eq:app-Fh}
 F_h(r)\equiv(r-h)e^{2\kappa_h b_h(r)}.
\end{equation}
Because $F_h(h)=0$ and $F_h'(h)=e^{2\kappa_h b_h(h)}>0$, we may choose
an interval $J_h\ni h$ and a number $\delta_h>0$ for which
\begin{equation}
 F_h:J_h\longrightarrow(-\delta_h,\delta_h)
\end{equation}
is an analytic diffeomorphism.  Shrinking the intervals if necessary,
we require
\begin{equation}\label{eq:app-disjoint-collars}
 J_{r_-}\cap J_{r_+}=\varnothing,\qquad
 \overline{J_{r_-}}\subset(0,r_+),\qquad
 \overline{J_{r_+}}\subset(r_-,\infty).
\end{equation}
Write $R_h\equiv F_h^{-1}$ and consider the full crossing model
\begin{equation}\label{eq:app-crossing-model}
 \mathcal X_h\equiv\left\{
 (U,V,\omega)\in\mathbb R^2\times\mathbb S^2:
 |UV|<\delta_h\right\},\qquad r\equiv R_h(UV).
\end{equation}
Allowing the full $U,V$ ranges includes the entire stationary
collars and the complete horizon generators.

Let $\mathscr R_\lambda$ be rotation through angle $\lambda$ about the
axis of $\mathbb S^2$.  For $\epsilon,\eta\in\{+1,-1\}$, set
\begin{equation}
 Q_{\epsilon\eta}\equiv\left\{(U,V,\omega)\in\mathcal X_h:
 \epsilon U>0,\ \eta V>0\right\}.
\end{equation}
If $\epsilon\eta=\operatorname{sgn}(r-h)$, define
\begin{align}
 \Xi_h^{\epsilon\eta}(t,r,\omega)
 \equiv\bigl(&\epsilon e^{-\kappa_h(t-r_\star(r))},
 \eta e^{\kappa_h(t+r_\star(r))},
 \mathscr R_{-\Omega_h t+\chi_h(r)}\omega\bigr).
 \label{eq:app-sector-map}
\end{align}
This map is an analytic diffeomorphism from the corresponding side of
the block collar onto $Q_{\epsilon\eta}$, since
\begin{equation}
 UV=\epsilon\eta e^{2\kappa_h r_\star(r)}
 =(r-h)e^{2\kappa_hb_h(r)}=F_h(r).
\end{equation}
Writing $\omega_h$ for the third coordinate, the inverse is
\begin{equation}\label{eq:app-sector-inverse}
 r=R_h(UV),\qquad
 t=\frac{1}{2\kappa_h}\log\frac{|V|}{|U|},\qquad
 \omega=\mathscr R_{\Omega_h t-\chi_h(r)}\omega_h.
\end{equation}
Because the angular changes are sphere rotations, the maps are
analytic at both poles.

\begin{lemma}[Analytic Lorentz metric on the crossing model]
\label{lem:app-crossing-metric}
The Kerr metrics transported by the four maps
$\Xi_h^{\epsilon\eta}$ are restrictions of one real-analytic Lorentzian
metric $g_h$ on all of $\mathcal X_h$.
\end{lemma}

\begin{proof}
We begin away from the axis.  Set
\begin{equation}
 \mathfrak s\equiv\sin\theta,\qquad \Sigma=A-a^2\mathfrak s^2,
 \qquad c_h\equiv\chi_h'=-\Omega_h\frac{r+h}{r-\bar h},
 \qquad q_h\equiv\Omega_h\frac{r+h}{r-\bar h},
\end{equation}
and define
\begin{align}
 G&\equiv\frac{\mathfrak s^2}{\Sigma}\bigl(A^2-a^2\Delta \mathfrak s^2\bigr),
 \label{eq:app-G}\\
 \alpha_h&\equiv\frac{1}{\Sigma}
 \left[-(1-a\Omega_h\mathfrak s^2)^2
 +\Delta \mathfrak s^2q_h^2\right],
 \label{eq:app-alpha}\\
 \beta_h&\equiv\frac{\mathfrak s^2}{\Sigma}
 \left[a(1-a\Omega_h\mathfrak s^2)+Aq_h\right].\label{eq:app-beta}
\end{align}
For $K_h\equiv T+\Omega_h\Phi$, the two-square form of the metric in
Lemma~\ref{lem:coframe} gives
\begin{equation}\label{eq:app-orbit-factorization}
 g(K_h,K_h)=\Delta\alpha_h,\qquad
 g(K_h,\Phi)=\Delta\beta_h,\qquad
 g(\Phi,\Phi)=G.
\end{equation}

On a quadrant, write
\begin{equation}
 \varphi\equiv\phi-\Omega_h t+\chi_h(r).
\end{equation}
Thus $d\phi=d\varphi+\Omega_hdt-c_h\,dr$, and
\begin{align}
 g={}&\Delta\alpha_hdt^2+2\Delta\beta_hdt\,d\varphi
 +Gd\varphi^2-2\Delta\beta_hc_hdt\,dr-2Gc_hd\varphi\,dr\nonumber\\
 &+\left(Gc_h^2+\frac{\Sigma}{\Delta}\right)dr^2
 +\Sigma d\theta^2.\label{eq:app-corotating-metric}
\end{align}
Set
\begin{equation}\label{eq:app-regular-coefficients}
 d_h(r)\equiv\frac{\Delta(r)}{F_h(r)}
 =(r-\bar h)e^{-2\kappa_hb_h(r)},\qquad
 p_h(r)\equiv\frac{1}{F_h'(r)}.
\end{equation}
Since $r=R_h(UV)$,
\begin{equation}\label{eq:app-differentials}
 dr=p_h(V\,dU+U\,dV),\qquad
 dt=\frac{1}{2\kappa_h}\left(\frac{dV}{V}-\frac{dU}{U}\right),
 \qquad \Delta=UVd_h.
\end{equation}
Define
\begin{equation}
 \mathfrak p_h\equiv\frac{d_h\alpha_h}{4\kappa_h^2},\qquad
 \mathfrak q_h\equiv\frac{\Sigma p_h^2}{d_h},\qquad
 \mathfrak h_h\equiv\mathfrak p_h+\mathfrak q_h.
\end{equation}
At $r=h$ one has
\begin{equation}
 1-a\Omega_h\mathfrak s^2=\frac{\Sigma_h}{A_h},\qquad
 \alpha_h(h,\theta)=-\frac{\Sigma_h}{A_h^2},
\end{equation}
where $\Sigma_h\equiv h^2+a^2\cos^2\theta$.  Moreover,
\begin{equation}
 d_h(h)=(h-\bar h)e^{-2\kappa_hb_h(h)},\qquad
 p_h(h)=e^{-2\kappa_hb_h(h)},\qquad
 4\kappa_h^2=\frac{(h-\bar h)^2}{A_h^2}.
\end{equation}
Consequently
\begin{equation}\label{eq:app-h-cancellation}
 \mathfrak p_h(h,\theta)
 =-\frac{\Sigma_h e^{-2\kappa_hb_h(h)}}{h-\bar h},\qquad
 \mathfrak q_h(h,\theta)
 =\frac{\Sigma_h e^{-2\kappa_hb_h(h)}}{h-\bar h}.
\end{equation}
Since $\mathfrak h_h(h,\theta)=0$ and $F_h$ has a simple zero,
\begin{equation}
 \widehat{\mathfrak h}_h\equiv\frac{\mathfrak h_h}{F_h}
\end{equation}
is analytic on $J_h\times\mathbb S^2$.

Substitution gives the analytic expression
\begin{align}
 g_h={}&\widehat{\mathfrak h}_h
 (V^2dU^2+U^2dV^2)
 +2(\mathfrak q_h-\mathfrak p_h)dU\,dV\nonumber\\
 &+\frac{d_h\beta_h}{\kappa_h}(U\,dV-V\,dU)d\varphi
 -\frac{d_h\beta_hc_hp_h}{\kappa_h}
 (U^2dV^2-V^2dU^2)\nonumber\\
 &-2Gc_hp_h(V\,dU+U\,dV)d\varphi
 +Gc_h^2p_h^2(V\,dU+U\,dV)^2\nonumber\\
 &+Gd\varphi^2+\Sigma d\theta^2.
 \label{eq:app-analytic-crossing-metric}
\end{align}
This expression restricts to the Kerr metric on each of the four
quadrants.

At the bifurcation sphere $U=V=0$, equations
\eqref{eq:app-h-cancellation} and
\eqref{eq:app-analytic-crossing-metric} give
\begin{equation}\label{eq:app-bifurcation-metric}
 g_h=\frac{4\Sigma_h e^{-2\kappa_hb_h(h)}}{h-\bar h}
 dU\,dV+\Sigma_h d\theta^2
 +\frac{A_h^2\mathfrak s^2}{\Sigma_h}d\varphi^2.
\end{equation}
The $UV$ block is indefinite and the spherical block is positive
definite. Thus the metric is Lorentzian at the bifurcation sphere.  At $U=0$, $V\ne0$, the
coordinates
\begin{equation}\label{eq:app-ingoing-from-kruskal}
 v\equiv\kappa_h^{-1}\log|V|,\qquad
 \psi\equiv\varphi+\Omega_hv,\qquad r=R_h(UV)
\end{equation}
put $g_h$ into the ingoing Kerr form obtained from
$v=t+r_\star$, $\psi=\phi+\phi_\star$.  At $V=0$, $U\ne0$, the
coordinates
\begin{equation}\label{eq:app-outgoing-from-kruskal}
 u_h\equiv-\kappa_h^{-1}\log|U|,\qquad
 \bar\psi\equiv\varphi+\Omega_hu_h-2\chi_h(r),\qquad r=R_h(UV)
\end{equation}
give the outgoing Kerr form obtained from
$u_h=t-r_\star$, $\bar\psi=\phi-\phi_\star$.  Away from the polar
coordinate singularity, both forms have determinant
$-\Sigma^2\sin^2\theta$.  Nondegeneracy therefore holds at every
nonbifurcate horizon point as well.

We finally check regularity at the axis. In either analytic polar
cap, use
$x\equiv\mathfrak s\cos\varphi$, $y\equiv\mathfrak s\sin\varphi$.  Then
\begin{equation}
 \mathfrak s^2d\varphi=x\,dy-y\,dx,\qquad
 \sin\theta\,d\theta=-d(\cos\theta),
\end{equation}
and
\begin{align}
 \Sigma d\theta^2+Gd\varphi^2
 ={}&\Sigma(d\theta^2+\mathfrak s^2d\varphi^2)\nonumber\\
 &+\frac{a^2\mathfrak s^4(2A-\Delta-a^2\mathfrak s^2)}{\Sigma}d\varphi^2.
 \label{eq:app-polar-cancellation}
\end{align}
The right side is analytic in polar caps, since
$\mathfrak s^4d\varphi^2=(x\,dy-y\,dx)^2$.  Each of the other $d\varphi$ terms contains a factor
$\mathfrak s^2$, through $G$ or $\beta_h$. Consequently $g_h$ is
analytic at the axis.  In an ingoing or outgoing polar-cap chart, with
$x\equiv\sin\theta\cos\psi$ and $y\equiv\sin\theta\sin\psi$, the Jacobian
$\partial(x,y)/\partial(\theta,\psi)=\sin\theta\cos\theta$ gives
\begin{equation}\label{eq:polar-cap-determinant}
 \det g_{(v,r,x,y)}=-\frac{\Sigma^2}{\mu^2}.
\end{equation}
The outgoing determinant is identical.  At a pole it tends to
$-(r^2+a^2)^2\ne0$.  The double-null changes have nonzero Jacobian
on nonbifurcate branches; at the bifurcation poles, use
\eqref{eq:app-bifurcation-metric} in polar caps.
\end{proof}

On $\mathcal X_h$ the Kerr Killing fields are
\begin{equation}\label{eq:app-killing-fields}
 T=-\kappa_hU\partial_U+\kappa_hV\partial_V
 -\Omega_h\partial_\varphi,\qquad
 \Phi=\partial_\varphi.
\end{equation}
Here $\partial_\varphi$ denotes the globally defined axial rotation
field on $\mathbb S^2$, which vanishes at the poles.

To continue the Carter tensor through the horizons, we use the
following explicit two-form. On each Boyer-Lindquist block, set
\begin{equation}\label{eq:app-ky-form}
 f\equiv a\cos\theta\,dr\wedge(dt-a\sin^2\theta\,d\phi)
 +r\sin\theta\,d\theta\wedge(A\,d\phi-a\,dt).
\end{equation}
In the coframe of Lemma~\ref{lem:coframe},
\begin{equation}\label{eq:app-ky-frame}
 f=-a\cos\theta\,e^0\wedge e^1+r\,e^2\wedge e^3.
\end{equation}
Write $\nabla_af_{bc}\equiv(\nabla_{E_a}f)(E_b,E_c)$.  On $\Delta>0$, introduce the following coefficients for Cartan's
equations:
\begin{align}
 q_0&\equiv\frac{r-M}{\sqrt{\Delta\Sigma}}
     -\frac{r\sqrt\Delta}{\Sigma^{3/2}},
 &q_1&\equiv\frac{r\sqrt\Delta}{\Sigma^{3/2}},
 &q_2&\equiv\frac{ar\sin\theta}{\Sigma^{3/2}},\nonumber\\
 q_3&\equiv\frac{a^2\sin\theta\cos\theta}{\Sigma^{3/2}},
 &q_4&\equiv\frac{a\sqrt\Delta\cos\theta}{\Sigma^{3/2}},
 &q_5&\equiv\frac{(r^2+a^2)\cos\theta}
 {\sin\theta\,\Sigma^{3/2}}.
 \label{eq:app-connection-coefficients}
\end{align}
Differentiating the coframe in Lemma~\ref{lem:coframe} gives
\begin{align}
 de^0&=-q_0e^0\wedge e^1+q_3e^0\wedge e^2
       -2q_4e^2\wedge e^3,\nonumber\\
 de^1&=q_3e^1\wedge e^2,
 &de^2&=q_1e^1\wedge e^2,\nonumber\\
 de^3&=-2q_2e^0\wedge e^1+q_1e^1\wedge e^3
       +q_5e^2\wedge e^3.
 \label{eq:app-coframe-derivatives}
\end{align}
With $\eta\equiv\operatorname{diag}(-1,1,1,1)$ and
$\omega_{ab}\equiv\eta_{ac}\omega^c{}_b$, the six independent connection
forms solving Cartan's first equation are
\begin{align}
 \omega_{01}&=-q_0e^0+q_2e^3,
 &\omega_{02}&=q_3e^0+q_4e^3,
 &\omega_{03}&=q_2e^1-q_4e^2,\nonumber\\
 \omega_{12}&=-q_3e^1-q_1e^2,
 &\omega_{13}&=-q_2e^0-q_1e^3,
 &\omega_{23}&=-q_4e^0-q_5e^3.
 \label{eq:app-connection-forms}
\end{align}
With these connection forms,
\begin{equation}
 \nabla_af_{bc}=E_a(f_{bc})
 -\omega^d{}_b(E_a)f_{dc}
 -\omega^d{}_c(E_a)f_{bd},
 \label{eq:app-ky-component-formula}
\end{equation}
and
\begin{equation}
 E_1(r)=\sqrt{\frac\Delta\Sigma},
 \qquad
 E_2(a\cos\theta)=-\frac{a\sin\theta}{\sqrt\Sigma}.
\end{equation}
Using \eqref{eq:app-connection-forms}-
\eqref{eq:app-ky-component-formula}, we obtain
\begin{align}
 \nabla_0f_{12}
 &=q_2r+q_3a\cos\theta
 =\frac{a\sin\theta}{\sqrt\Sigma},\nonumber\\
 \nabla_2f_{01}
 &=E_2(-a\cos\theta)
 =\frac{a\sin\theta}{\sqrt\Sigma},\nonumber\\
 \nabla_1f_{02}
 &=-q_2r-q_3a\cos\theta
 =-\frac{a\sin\theta}{\sqrt\Sigma}.
 \label{eq:app-ky-first-triple}
\end{align}
The other nonzero triple is
\begin{align}
 \nabla_1f_{23}
 &=E_1(r)=\sqrt{\frac\Delta\Sigma},\nonumber\\
 \nabla_2f_{13}
 &=-q_1r-q_4a\cos\theta
 =-\sqrt{\frac\Delta\Sigma},\nonumber\\
 \nabla_3f_{12}
 &=q_1r+q_4a\cos\theta
 =\sqrt{\frac\Delta\Sigma}.
 \label{eq:app-ky-second-triple}
\end{align}
All other components with $b<c$ vanish.  Thus
\begin{align}
 \nabla_0f_{12}=\nabla_2f_{01}=-\nabla_1f_{02}
 &=\frac{a\sin\theta}{\sqrt\Sigma},\nonumber\\
 \nabla_1f_{23}=-\nabla_2f_{13}=\nabla_3f_{12}
 &=\sqrt{\frac\Delta\Sigma}.
 \label{eq:app-ky-derivatives}
\end{align}
It follows that $\nabla_af_{bc}$ is totally alternating, or
equivalently,
\begin{align}
 \nabla_af_{bc}&=\frac13(df)_{abc},\nonumber\\
 df&=\frac{3a\sin\theta}{\sqrt\Sigma}
 e^0\wedge e^1\wedge e^2
 +3\sqrt{\frac\Delta\Sigma}
 e^1\wedge e^2\wedge e^3\nonumber\\
 &=3\Sigma\sin\theta\,dr\wedge d\theta\wedge d\phi.
 \label{eq:app-ky-exterior-derivative}
\end{align}
Hence $\nabla_{(\alpha}f_{\beta)\gamma}=0$ on $\{\Delta>0\}$.
After substituting \eqref{eq:app-ky-form}, its coordinate numerators
are analytic identities in $r,\cos\theta$, so the equation holds
wherever $\Sigma>0$, $\Delta\ne0$, including the axis by continuity.
To check horizon regularity, substitute
\eqref{eq:app-differentials}, $A\Omega_h-a=\Delta q_h$, and
$d\phi=d\varphi+\Omega_hdt-c_hdr$. This gives
\begin{align}
 f={}&\frac{a\cos\theta\,p_h}{\kappa_h}
 (1-a\Omega_h\mathfrak s^2)dU\wedge dV
 +a^2\cos\theta\,\mathfrak s^2p_h\,d\varphi\wedge(VdU+UdV)
 \nonumber\\
 &+r\sin\theta\,d\theta\wedge
 \left\{
 \frac{d_hq_h}{2\kappa_h}(UdV-VdU)+A\,d\varphi
 -Ac_hp_h(VdU+UdV)
 \right\}.
 \label{eq:app-ky-crossing}
\end{align}
This is analytic, including the poles by the identities
$\sin\theta\,d\theta=-d(\cos\theta)$,
$\mathfrak s^2d\varphi=xdy-ydx$ and the smooth spherical area form.
Consequently
\begin{equation}\label{eq:app-carter-tensor}
 \mathsf K_{\alpha\beta}\equiv f_{\alpha\lambda}f_{\beta}{}^{\lambda}
\end{equation}
is analytic on every crossing model.  The Killing-Yano equation holds on the whole model by analytic
continuation from the dense union of quadrants.  Set
\begin{equation}
 A_{\alpha\beta\delta}\equiv \nabla_\alpha f_{\beta\delta}
 =A_{[\alpha\beta\delta]}.
\end{equation}
Metric compatibility gives
\begin{equation}
 \nabla_\alpha\mathsf K_{\beta\gamma}
 =A_{\alpha\beta\delta}f_\gamma{}^\delta
 +f_\beta{}^\delta A_{\alpha\gamma\delta}.
\end{equation}
Cyclic symmetrization then gives
\begin{align}
 3\nabla_{(\alpha}\mathsf K_{\beta\gamma)}
 ={}&
 A_{\alpha\beta\delta}f_\gamma{}^\delta
 +f_\beta{}^\delta A_{\alpha\gamma\delta}
 +A_{\beta\gamma\delta}f_\alpha{}^\delta\nonumber\\
 &+f_\gamma{}^\delta A_{\beta\alpha\delta}
 +A_{\gamma\alpha\delta}f_\beta{}^\delta
 +f_\alpha{}^\delta A_{\gamma\beta\delta}=0.
\end{align}
The terms cancel in pairs by alternation, which proves the
Killing-tensor identity.  In the orthonormal frame,
\begin{equation}\label{eq:app-carter-frame-components}
 (\mathsf K_{ab})
 =\operatorname{diag}
 \bigl(a^2\cos^2\theta,-a^2\cos^2\theta,r^2,r^2\bigr).
\end{equation}
For a unit timelike vector $\xi$, the normalization identity
$-(\xi^0)^2+(\xi^1)^2+(\xi^2)^2+(\xi^3)^2=-1$ and the formulas
\begin{equation}
 \xi^2=\frac{p_\theta}{\sqrt\Sigma},\qquad
 \xi^3=\frac{L-aE\sin^2\theta}{\sqrt\Sigma\sin\theta}
\end{equation}
give
\begin{align}
 \mathsf K(\xi,\xi)
 &=a^2\cos^2\theta\bigl((\xi^0)^2-(\xi^1)^2\bigr)
 +r^2\bigl((\xi^2)^2+(\xi^3)^2\bigr)\nonumber\\
 &=a^2\cos^2\theta+p_\theta^2
 +\frac{(L-aE\sin^2\theta)^2}{\sin^2\theta}
 =\mathcal Q.
 \label{eq:app-carter-normalization}
\end{align}
Thus the global tensor has the normalization used in
Proposition~\ref{prop:separation}.

\section{Finite-proper-time timelike geodesics}\label{sec:geodesics}

The aim of this section is to show that every timelike geodesic
which is inextendible in finite proper time approaches the ring. We
use the horizon coordinates of the preceding section and Carter
separation \cite{CarterSeparation}. Completeness away from the rings
and equatorial causal ring impact were studied in
\cite{BoyerLindquist,Carter}; further analyses of geodesic endpoints
include \cite{Semerak1997,RudnickiZieba2000}. The null geodesics in
\cite{HelliwellMallinckrodt1975,MaliborskiSutter2026} cross regular
points of the disc. Explicit subextremal horizon-crossing formulas are
given in \cite{BakunEtAl2025}. Here we prove the continuation and
localization statements for every marked realization, in the form
needed for a geodesic ending at a continuous extension boundary.

\subsection{First integrals and separation}

\begin{lemma}[Inverse metric and orbit determinant]\label{lem:inverse}
Away from the polar-coordinate axis and from a Boyer-Lindquist horizon, the nonzero inverse components are
\begin{align}
 g^{tt}&=-\frac{(r^2+a^2)^2-a^2\Delta\sin^2\theta}{\Sigma\Delta},
 &g^{t\phi}&=-\frac{2Mar}{\Sigma\Delta},\nonumber\\
 g^{rr}&=\frac{\Delta}{\Sigma},
 &g^{\theta\theta}&=\frac1\Sigma,
 &g^{\phi\phi}&=\frac{\Delta-a^2\sin^2\theta}{\Sigma\Delta\sin^2\theta}.
\label{eq:inverse}
\end{align}
The algebraic identity $g_{tt}g_{\phi\phi}-g_{t\phi}^2=-\Delta\sin^2\theta$ extends across every regular point at which the polar coordinates are valid.
\end{lemma}

\begin{proof}
The dual of the coframe in Lemma~\ref{lem:coframe} is
\begin{align}
 E_0&=\frac{(r^2+a^2)\partial_t+a\partial_\phi}{\sqrt{\Sigma\Delta}},
 &E_1&=\sqrt{\frac\Delta\Sigma}\,\partial_r,\nonumber\\
 E_2&=\frac1{\sqrt\Sigma}\,\partial_\theta,
 &E_3&=\frac{a\sin\theta\,\partial_t+\sin^{-1}\theta\,\partial_\phi}{\sqrt\Sigma}.
\end{align}
The identity $g^{-1}=-E_0^2+E_1^2+E_2^2+E_3^2$ gives
\eqref{eq:inverse} when $\Delta>0$.  Both sides are rational algebraic
expressions, so the formula continues to every Boyer-Lindquist region
with $\Delta\ne0$.  The determinant follows either by inverting the
$2\times2$ orbit block or by expanding \eqref{eq:orbit-components}.
\end{proof}

Let $\gamma$ be a unit-speed timelike geodesic, with affine parameter denoted by $s$.  The Killing constants are
\begin{equation}\label{eq:E-L}
 E\equiv-g(\dot\gamma,T),
 \qquad
 L\equiv g(\dot\gamma,\Phi).
\end{equation}
Since $T$ and $\Phi$ are Killing fields, these quantities are
constant. Put
\begin{equation}\label{eq:P}
 P(r)\equiv E(r^2+a^2)-aL.
\end{equation}

\begin{proposition}[Carter constant and separated equations]\label{prop:separation}
Every unit-speed timelike Kerr geodesic has a constant $\mathcal Q\ge0$ for which
\begin{align}
 \Sigma^2\dot r^2&=R(r)\equiv P(r)^2-\Delta(r^2+\mathcal Q),
 \label{eq:R}\\
 \Sigma^2\dot\theta^2&=\Theta(\theta)\equiv \mathcal Q-a^2\cos^2\theta
 -\frac{(L-aE\sin^2\theta)^2}{\sin^2\theta}.
 \label{eq:Theta}
\end{align}
Where $\Delta\ne0$ and $\sin\theta\ne0$,
\begin{align}
 \dot t&=\frac1\Sigma\left[\frac{(r^2+a^2)P}{\Delta}
 +a(L-aE\sin^2\theta)\right],
 \label{eq:tdot}\\
 \dot\phi&=\frac1\Sigma\left[\frac{aP}{\Delta}
 +\frac{L}{\sin^2\theta}-aE\right].
 \label{eq:phidot}
\end{align}
The exact radial polynomial is
\begin{align}
 R(r)={}&(E^2-1)r^4+2Mr^3
 +\bigl(2E(Ea^2-aL)-a^2-\mathcal Q\bigr)r^2\nonumber\\
 &+2M\mathcal Q r+(Ea^2-aL)^2-a^2\mathcal Q.
 \label{eq:R-polynomial}
\end{align}
We write $R^+(r)\equiv\max\{R(r),0\}$.
The constants and the radial identity hold through regular horizons.
The regular combinations of \eqref{eq:tdot}-\eqref{eq:phidot} are
given in Lemma~\ref{lem:horizon-continuation}. At the axis,
$\mathcal Q$ is defined by the smooth Carter tensor, and
\eqref{eq:Theta} is understood in its $L=0$ limit.
\end{proposition}

\begin{proof}
On the cotangent bundle let $p_t=-E$, $p_\phi=L$ and
\begin{equation}
 H\equiv\frac12g^{\mu\nu}p_\mu p_\nu.
\end{equation}
Substitution of Lemma~\ref{lem:inverse} and
$P=E(r^2+a^2)-aL$ gives the additive identity
\begin{equation}\label{eq:additive-separation}
 2\Sigma H
 =\underbrace{\Delta p_r^2-\frac{P^2}{\Delta}}_{F_r(r,p_r)}
 +\underbrace{p_\theta^2+
 \frac{(L-aE\sin^2\theta)^2}{\sin^2\theta}}_{F_\theta(\theta,p_\theta)}.
\end{equation}
Since $2H=-1$, Hamilton's equations show that
\begin{equation}\label{eq:Q-definition}
 \mathcal Q\equiv p_\theta^2+
 \frac{(L-aE\sin^2\theta)^2}{\sin^2\theta}
 +a^2\cos^2\theta
\end{equation}
is constant.  Hamilton's equation for the angular momentum is
\begin{equation}\label{eq:pdottheta}
 \dot p_\theta
 =-\frac{\partial_\theta F_\theta}{2\Sigma}
 +\frac{H}{\Sigma}\,\partial_\theta\Sigma,
\end{equation}
because $2\Sigma H=F_r+F_\theta$.  Since
$\dot\theta=p_\theta/\Sigma$ and
$\partial_\theta\Sigma=\partial_\theta(a^2\cos^2\theta)$, one obtains
\begin{align}
 \frac{d\mathcal Q}{ds}
 &=\frac{p_\theta}{\Sigma}\,\partial_\theta F_\theta
 +2p_\theta\left[-\frac{\partial_\theta F_\theta}{2\Sigma}
 +\frac{H}{\Sigma}\partial_\theta(a^2\cos^2\theta)\right]
 +\frac{p_\theta}{\Sigma}\partial_\theta(a^2\cos^2\theta)\nonumber\\
 &=\frac{p_\theta}{\Sigma}(2H+1)
 \partial_\theta(a^2\cos^2\theta)=0.
\end{align}
Every term in \eqref{eq:Q-definition} is nonnegative.  Since
$p_r=(\Sigma/\Delta)\dot r$ and $p_\theta=\Sigma\dot\theta$, substituting $F_\theta=\mathcal Q-a^2\cos^2\theta$ in \eqref{eq:additive-separation} gives \eqref{eq:R} and \eqref{eq:Theta}.  Inverting the orbit block using \eqref{eq:E-L} and Lemma~\ref{lem:inverse} gives \eqref{eq:tdot}-\eqref{eq:phidot}.  Expanding $P^2-\Delta(r^2+\mathcal Q)$ gives \eqref{eq:R-polynomial}.

At an axis crossing, $\Phi=0$ gives $L=0$, and the quotient in
\eqref{eq:Q-definition} becomes $a^2E^2\sin^2\theta$.
The global fields and tensor give global constants $E,L,\mathcal Q$
by \eqref{eq:app-carter-normalization}.  The radial identity contains no horizon denominator and hence
extends by continuity. The time and angular formulas hold on the
stated Boyer-Lindquist domains.
\end{proof}

\begin{remark}[Carter-constant convention]\label{rem:carter-convention}
Many references use
\begin{equation}
 Q_{\mathrm C}\equiv \mathcal Q-(L-aE)^2.
\end{equation}
Thus the classical equatorial ring-impact condition $Q_{\mathrm C}=0$
is exactly $\mathcal Q=(L-aE)^2$ in the nonnegative normalization
\eqref{eq:Q-definition}.
\end{remark}

\begin{remark}[Elementary consequences]\label{rem:separation-consequences}
The motion is confined to $R\ge0$ and $\Theta\ge0$.  On a compact $r$ interval,
\begin{equation}
 \abs{\Sigma\dot r}=\sqrt{R}\le C,
 \qquad
 \abs{\Sigma\dot\theta}=\sqrt\Theta\le\sqrt{\mathcal Q}.
\end{equation}
If $L\ne0$, the last term in \eqref{eq:Theta} is asymptotic to
$L^2/\sin^2\theta$ at the rotation axis and prevents an approach to it.
Finally, $R$ is never the zero polynomial because its cubic coefficient
is $2M>0$.
\end{remark}

\subsection{Continuation away from the ring}

\begin{lemma}[Escape of the tangent lift]\label{lem:escape}
Let $(N,h)$ be smooth and let $\gamma:[0,b)\to N$, $b<\infty$, be a maximal affinely parametrized geodesic.  Its tangent lift leaves every compact subset of $TN$.
\end{lemma}

\begin{proof}
Suppose, to the contrary, that a sequence of tangent points
$z(s_j)$, $s_j\uparrow b$, converges in $TN$. Smoothness of the
geodesic spray gives a uniform local existence time $\varepsilon>0$
near the limit. For sufficiently late $z(s_j)$, uniqueness then
extends the geodesic beyond $b$, contrary to maximality.
\end{proof}

\begin{lemma}[No finite-proper-time escape through infinity]\label{lem:radial-bounded}
If a unit-speed timelike Kerr geodesic is defined on $[0,\Lambda)$ with $\Lambda<\infty$, then $|r|$ remains bounded.
\end{lemma}

\begin{proof}
The constants in \eqref{eq:R-polynomial} are fixed, so $R^+(r)\le C(1+r^4)$.  For $|r|\ge1$, one has $\Sigma^2\ge r^4$, and \eqref{eq:R} gives
\begin{equation}
 |\dot r|^2=\frac{R(r)}{\Sigma^2}\le 2C.
\end{equation}
It follows that $r$ has bounded speed outside a fixed compact
interval. It therefore reaches neither $+\infty$ nor $-\infty$ in
the finite parameter time $\Lambda$.
\end{proof}

\begin{lemma}[Finitely many horizon contacts]\label{lem:finite-horizons}
Let $\gamma:[0,\Lambda)\to\KerrSpace$ be a unit-speed timelike geodesic,
$\Lambda<\infty$, with $|r|\le B$.  Then $\Delta(r(s))=0$ for only
finitely many $s\in[0,\Lambda)$.
\end{lemma}

\begin{proof}
Let $Z_R$ be the finite set of zeros of the nonzero polynomial $R$ in
$[-B,B]$, let $\mathscr H_r$ be the finite set of horizon radii, and put
$V\equiv\sqrt{\max_{[-B,B]}R^+}$.  On a compact parameter interval introduce
the rescaled parameter $d\tau/ds=\Sigma^{-1}$ and set
$v\equiv dr/d\tau$.  The radial
Hamilton equations give the smooth phase system
\begin{equation}\label{eq:radial-phase-system}
 \frac{dr}{d\tau}=v,\qquad
 \frac{dv}{d\tau}=\frac12R'(r),\qquad v^2=R(r).
\end{equation}
Here $v=\Delta p_r$, and Hamilton's equation, with $H=-1/2$ and
$\partial_r\Sigma=2r$, gives
\begin{align}
 \frac{dv}{d\tau}
 &=\frac{\Delta'}{2\Delta}(v^2-P^2)+PP'-r\Delta\nonumber\\
 &=-\frac{\Delta'}2(r^2+\mathcal Q)+PP'-r\Delta
 =\frac12\frac{d}{dr}\bigl(P^2-\Delta(r^2+\mathcal Q)\bigr).
\end{align}
This calculation extends by continuity through regular horizons,
also when $v=0$, since it involves no division by $v$. The radial
sign can therefore change only at $Z_R$. The polynomial phase system
has analytic solutions.  Accumulation of contacts with a horizon $r=h$ at an
interior parameter value would imply $r\equiv h$, impossible for a
timelike curve in a null hypersurface.  Every contact is isolated.

If $V=0$, the continuous image of $r$ lies in the finite set $Z_R$, so
$r$ is constant and there are no horizon events.  Assume from now on
that $V>0$.  On every monotone radial subarc,
\begin{equation}\label{eq:weighted-horizon-time}
 ds=\frac{\Sigma}{\sqrt R}|dr|
 \ge\frac{r^2}{V}|dr|
 =\frac1V\left|d\!\left(\frac{r^3}{3}\right)\right|.
\end{equation}
If two consecutive events have distinct horizon values
$h_1,h_2\in\mathscr H_r$, their weighted radial variation is at least
$|h_2^3-h_1^3|/3$.  If they have the same value $h$, a nonconstant
excursion contains a distinct turning value
$z\in Z_R\setminus\{h\}$ and has weighted variation at least
$2|z^3-h^3|/3$.  These positive numbers form a finite set. Thus
\eqref{eq:weighted-horizon-time} gives a uniform positive lower
bound for the affine time of a complete horizon-to-horizon excursion.
There can consequently be only finitely many contacts before
$\Lambda$.
The proof uses only the finiteness of $\mathscr H_r$, and also applies
to the extremal case $\mathscr H_r=\{M\}$.
\end{proof}

\begin{lemma}[Horizon continuation]\label{lem:horizon-continuation}
In the subextremal range \eqref{eq:parameters}, let $r_h$ be a root
of $\Delta$.  Suppose a unit timelike geodesic has $r(s)\to r_h$
monotonically as $s\uparrow\Lambda<\infty$, with $\Sigma\ge\rho>0$
on the tail.
\begin{enumerate}[label=(\alph*)]
 \item If $P(r_h)\ne0$, the tangent lift is bounded in one of the two Eddington-Finkelstein charts and the geodesic extends through $s=\Lambda$.
 \item If $P(r_h)=0$, the side $\Delta>0$ is forbidden by $R<0$.  An approach from the side $\Delta<0$ reaches the bifurcation sphere with bounded tangent in a regular double-null chart and extends through it.
\end{enumerate}
\end{lemma}

\begin{proof}
Write $\Sigma\dot r=\varepsilon\sqrt R$ on the monotone branch.  If $P_h\equiv P(r_h)\ne0$, then $R(r_h)=P_h^2$.  In the ingoing chart of \eqref{eq:ef-coordinates},
\begin{align}
 \dot v
 &=\frac{(r^2+a^2)(P+\varepsilon\sqrt R)}{\Delta\Sigma}
 +\frac{a(L-aE\sin^2\theta)}{\Sigma}.
\end{align}
When $\varepsilon P_h<0$, the identity
\begin{equation}
 (P+\varepsilon\sqrt R)(P-\varepsilon\sqrt R)
 =P^2-R=\Delta(r^2+\mathcal Q)
\end{equation}
rewrites the first term as
\begin{equation}
 \frac{(r^2+a^2)(r^2+\mathcal Q)}{(P-\varepsilon\sqrt R)\Sigma},
\end{equation}
since $P-\varepsilon\sqrt R\to2P_h\ne0$.  The same cancellation
bounds the radial part of
$\dot{\bar\phi}=\dot\phi+a\dot r/\Delta$.  Its polar term is
bounded: $L\ne0$ excludes approach to an axis by $\Theta\ge0$, while
for $L=0$ the singular quotient vanishes.  Thus the angular velocity
is bounded in ordinary or polar-cap coordinates.  If
$\varepsilon P_h>0$, use the outgoing chart with the opposite signs.
Remark~\ref{rem:separation-consequences} bounds $\dot r,\dot\theta$.
The regular time and spherical coordinates therefore have limits
on the finite parameter interval. Their limit lies in the full horizon
chart, where the bounded tangent lift is precompact. By
Lemma~\ref{lem:escape}, the geodesic extends through the horizon.

Suppose next that $P_h=0$. The horizon roots are simple in the
parameter range \eqref{eq:parameters}. Taylor expansion gives
\begin{equation}\label{eq:critical-horizon-expansion}
 R(r)=-\Delta'(r_h)(r_h^2+\mathcal Q)(r-r_h)
 +O((r-r_h)^2).
\end{equation}
The coefficient $r_h^2+\mathcal Q$ is positive.  Thus $R<0$ on the adjacent $\Delta>0$ side.  On the $\Delta<0$ side, $R$ has a positive simple zero, so $\dot r\to0$ and the affine arrival time is finite.  Moreover $P/\Delta\to P'(r_h)/\Delta'(r_h)$, and \eqref{eq:tdot}-\eqref{eq:phidot} give finite limits for $\dot t$ and $\dot\phi$ in the Boyer-Lindquist variables as the bifurcation sphere is approached.

To estimate the tangent in regular coordinates, put
\begin{equation}
 \Omega_h\equiv\frac{a}{r_h^2+a^2},
 \qquad
 \kappa_h\equiv\frac{\Delta'(r_h)}{2(r_h^2+a^2)},
\end{equation}
write $r_*=(2\kappa_h)^{-1}\log|r-r_h|+O(1)$, and use
\begin{equation}
 U\equiv\pm e^{-\kappa_h(t-r_*)},
 \qquad
 V\equiv\pm e^{\kappa_h(t+r_*)},
 \qquad
 \varphi_h\equiv\phi-\Omega_h t+\chi_h(r),
\end{equation}
where
\begin{equation}
 \chi_h'(r)=\frac{a-\Omega_h(r^2+a^2)}{\Delta(r)}
\end{equation}
is analytic at $r_h$.  Equation~\eqref{eq:critical-horizon-expansion}
gives
\begin{equation}
 U,V=O(|r-r_h|^{1/2}),\qquad
 \dot r=O(|r-r_h|^{1/2}),\qquad
 \dot r_*=O(|r-r_h|^{-1/2}).
\end{equation}
Thus
\begin{equation}
 \dot U=-\kappa_hU(\dot t-\dot r_*),\qquad
 \dot V=\kappa_hV(\dot t+\dot r_*)
\end{equation}
are bounded.  Likewise $\dot\varphi_h=\dot\phi-\Omega_h\dot t+\chi_h'\dot r$
is bounded away from the axis.  At an axis the argument in
\textup{(a)} gives $L=0$ and bounded polar-cap velocity; the angular
equation bounds $\dot\theta$.  The base coordinates thus converge to the bifurcation sphere
$U=V=0$, and the tangent lift is precompact in a regular double-null
chart. Lemma~\ref{lem:escape} again gives continuation.
\end{proof}

\begin{proposition}[Continuation away from $\Sigma=0$]\label{prop:away-ring-continuation}
Let $\gamma:[0,\Lambda)\to\KerrSpace$ be a future-inextendible unit-speed
timelike geodesic with $\Lambda<\infty$.  It cannot satisfy
$\Sigma\ge\rho>0$ on a terminal interval.
\end{proposition}

\begin{proof}
Suppose $\Sigma\ge\rho>0$ on a terminal interval.
Lemmas~\ref{lem:radial-bounded} and~\ref{lem:finite-horizons} bound
$r$ and place the tail, after its last horizon event, in one Carter
block.  Since $R^+$ is bounded on the radial range,
$|\dot r|=\sqrt{R(r)}/\Sigma\le C_\rho$, so $r\to r_\infty$.

If $\Delta(r_\infty)\ne0$, take
$|\Delta|\ge\delta_\Delta>0$ on the tail.  Equations
\eqref{eq:Theta} and~\eqref{eq:phidot} give
\begin{align}
 |\dot\theta|&\le\frac{\sqrt{\mathcal Q}}\rho,\nonumber\\
 |\sin\theta\,\dot\phi|
 &\le\frac1\rho\left(
  \frac{|aP|}{\delta_\Delta}
  +\frac{|L-aE\sin^2\theta|}{\sin\theta}\right)
 \le\frac1\rho\left(
  \frac{|aP|}{\delta_\Delta}+\sqrt{\mathcal Q}\right).
 \label{eq:sphere-speed-bound}
\end{align}
The $L=0$ limit covers the axis.  The round-sphere speed and $\dot t$ are therefore bounded. Hence
the time and spherical coordinates converge on the finite interval.  Lemma~\ref{lem:block-properness} and a chart
at the angular limit make the tangent lift precompact.
Lemma~\ref{lem:escape} contradicts inextendibility.

If $\Delta(r_\infty)=0$, then either $P(r_\infty)\ne0$ and
$R(r_\infty)>0$, or \eqref{eq:critical-horizon-expansion} makes
$r_\infty$ a simple zero of $R$.  In either case $R$ has no other
zero on the allowed side sufficiently near the horizon, so the tail
is radially monotone.  The constant alternative would put a timelike
curve in a null horizon.  Lemma~\ref{lem:horizon-continuation} then extends the geodesic
across the horizon, again contradicting inextendibility.
\end{proof}

\subsection{Localization to the ring}

We will use the following form of Barb\u{a}lat's lemma. Related
forms and quantitative estimates can be found in
\cite{FarkasWegner2016}.

\begin{lemma}[Uniformly Lipschitz $L^2$ functions]\label{lem:lipschitz-l2}
If $f:[0,\infty)\to\mathbb R$ is uniformly Lipschitz and
$f\in L^2([0,\infty))$, then $f(\tau)\to0$ as $\tau\to\infty$.
\end{lemma}

\begin{proof}
Let $K$ be a Lipschitz constant.  If the conclusion failed, there would
be $\varepsilon>0$ and a sequence $\tau_j\to\infty$ with
$|f(\tau_j)|\ge\varepsilon$.  For $K>0$, the estimate
$|f(\tau)|\ge\varepsilon/2$ holds whenever
$|\tau-\tau_j|\le\varepsilon/(2K)$.  Pass to a subsequence for which
these intervals are disjoint.  Their contributions make
$\int_0^\infty|f|^2d\tau$ infinite.  If $K=0$, the assertion is
immediate.
\end{proof}

\begin{theorem}[Localization to the ring]\label{thm:ring-localization}
Let $\gamma:[0,\Lambda)\to\KerrSpace$ be a future-inextendible unit-speed
timelike geodesic with $\Lambda<\infty$.  Then
\begin{equation}\label{eq:ring-localization}
 \Sigma(s)\longrightarrow0,
 \qquad
 r(s)\longrightarrow0,
 \qquad
 \theta(s)\longrightarrow\frac\pi2.
\end{equation}
\end{theorem}

\begin{proof}
Lemma~\ref{lem:radial-bounded} gives $|r|\le B$.  Put $x\equiv\cos\theta$
and introduce
\begin{equation}\label{eq:rescaled-time}
 \tau(s)\equiv\int_0^s\frac{d\sigma}{\Sigma(\gamma(\sigma))}.
\end{equation}
The separated equations give uniform bounds
\begin{equation}\label{eq:rescaled-lipschitz}
 \left|\frac{dr}{d\tau}\right|=\sqrt{R(r)}\le C_B,
 \qquad
 \left|\frac{dx}{d\tau}\right|
 =\sin\theta\sqrt{\Theta(\theta)}\le\sqrt{\mathcal Q}.
\end{equation}
If the limiting value of $\tau$ is infinite, then
\begin{equation}
 \int_0^\infty\bigl(r(\tau)^2+a^2x(\tau)^2\bigr)d\tau
 =\Lambda<\infty,
\end{equation}
and Lemma~\ref{lem:lipschitz-l2} gives $r,x\to0$.
If the limiting value is finite, the Lipschitz estimates give limits
$r_\infty,x_\infty$.  A positive value of
$r_\infty^2+a^2x_\infty^2$ would contradict
Proposition~\ref{prop:away-ring-continuation}.  Thus again
$r,x\to0$, proving \eqref{eq:ring-localization}.
\end{proof}

\section{Global \texorpdfstring{$C^0$}{C0}-inextendibility}
\label{sec:global-inextendibility}

We combine ring localization with the local causal obstruction to
prove the main theorems. We then give the corresponding continuation
arguments for the overextremal and extremal realizations.

\subsection{Proofs of the main theorems}

\begin{proof}[Proof of Theorem~\ref{thm:main-terminal-lengths}]
The curve is eventually contained in one ring collar.  On the base
spacetime, Theorem~\ref{thm:fixed-endpoint-infinite-length} and the
local observation following its proof give the assertion in all three
parameter regimes.  The assertion on a connected cover follows by
applying Lemma~\ref{lem:covering-causal-length} to the projection of
the curve.
Reversing the time orientation gives the past-directed assertion.
The assertion for arbitrary positive continuous conformal factors follows
from Theorem~\ref{thm:conformal-class-ring}. The endpoint assertion
for locally Lipschitz causal approaches follows from
Corollary~\ref{cor:causal-ring-approaches}, applied on the terminal
ring collar.
\end{proof}

\begin{proof}[Proof of Theorem~\ref{thm:main-c0}]
We begin with the subextremal case.  Let $\pi:\widehat{\KerrSpace}\to\KerrSpace$ be a connected
covering of a member of $\CarterClass$, with the pulled-back metric.
Suppose that $\widehat{\KerrSpace}$ has a proper $C^0$ extension.
Theorem~\ref{thm:boundary-geodesic} provides, after time reversal if
needed, a future-inextendible unit-speed timelike geodesic
$\widehat\gamma:[0,\Lambda)\to\widehat{\KerrSpace}$,
$\Lambda<\infty$, converging to an extension boundary point.
The projection $\gamma\equiv\pi\circ\widehat\gamma$ is a
unit-speed timelike geodesic.  It is also future-inextendible in the
base, since any extension would lift from the prescribed initial point
and, by uniqueness of path lifting, extend $\widehat\gamma$.
Theorem~\ref{thm:ring-localization} gives $\Sigma\circ\gamma\to0$.
By Theorem~\ref{thm:main-terminal-lengths}, arbitrarily late
segments of $\widehat\gamma$ have infinite causal homotopy length,
in contradiction to Theorem~\ref{thm:causal-length-criterion}.
The same projection argument proves localization on every cover.
The overextremal and extremal cases follow from
Corollaries~\ref{cor:overextremal-inextendibility}
and~\ref{cor:extremal-inextendibility}, respectively.
\end{proof}

\subsection{Overextremal Kerr}

We next consider the Kerr spacetime without horizons.  For $M>0$ and $|a|>M$, put
\begin{equation}\label{eq:overextremal-kerr-manifold}
 \mathcal K^{\mathrm{over}}_{M,a}\equiv
 \bigl(\mathbb R_t\times\mathbb R_r\times\mathbb S^2\bigr)
 \setminus\{r=0,\ \mu=0\},
\end{equation}
with the Kerr metric \eqref{eq:kerr-metric}, interpreted in smooth
polar charts.  Both regular disc branches are present.  Since
$\Delta=(r-M)^2+a^2-M^2>0$, this is an analytic spacetime with the
time orientation determined by
$Z\equiv(r^2+a^2)T+a\Phi$, for which
$g(Z,Z)=-\Delta\Sigma<0$.

\begin{corollary}[Overextremal Kerr]
\label{cor:overextremal-inextendibility}
For $M>0$ and $|a|>M$, the spacetime
$\mathcal K^{\mathrm{over}}_{M,a}$ and all its connected smooth
coverings are $C^0$-inextendible.  Every future-inextendible unit-speed
timelike geodesic of finite proper time in any of these spacetimes
satisfies $\Sigma\to0$.
\end{corollary}

\begin{proof}
The separation identities depend only on the Kerr metric and the
Carter tensor.  Thus Proposition~\ref{prop:separation} applies.  For
a unit timelike geodesic which is inextendible in finite proper time,
Lemma~\ref{lem:radial-bounded} bounds $r$.
If $\Sigma\ge\rho>0$ on a tail, the uniform inequality
$\Delta\ge a^2-M^2>0$ bounds the time, radial, and round-sphere
velocities as in \eqref{eq:sphere-speed-bound}.  The base point therefore converges in
\eqref{eq:overextremal-kerr-manifold}, and the tangent lift is
precompact in a regular chart.  This contradicts
Lemma~\ref{lem:escape}.  The rescaled-time argument of
Theorem~\ref{thm:ring-localization} therefore gives $\Sigma\to0$.

The ring construction applies because $M>0$, $a\ne0$, and
$\Delta(0)=a^2>0$.  The boundary-geodesic theorem and
Theorem~\ref{thm:causal-length-criterion} exclude extensions in
either time orientation.  For a connected cover, path lifting preserves localization, while
Lemma~\ref{lem:covering-causal-length} transfers the infinite
homotopy lengths.  The same conclusion follows.
\end{proof}

\subsection{Extremal Kerr}
\label{subsec:extremal-kerr}

At the degenerate horizon we use ingoing and outgoing coordinates
in place of the bifurcation charts.  Geodesic continuation is then
simpler than in the bifurcate case.  Fix
$M>0$ and $a^2=M^2$, so that $\Delta=(r-M)^2$, and choose
$0<\eta<M/2$.  Put
\begin{equation}\label{eq:extremal-primitives}
 r_*(r)\equiv r+2M\log|r-M|-\frac{2M^2}{r-M},
 \qquad \phi_*(r)\equiv-\frac a{r-M}.
\end{equation}
Then $r_*'=(r^2+a^2)/\Delta$ and $\phi_*'=a/\Delta$ on either
side of $r=M$.

For $\epsilon\in\{-1,1\}$, introduce
\begin{equation}\label{eq:extremal-ef-maps}
 v_\epsilon\equiv t+\epsilon r_*(r),\qquad
 \psi_\epsilon\equiv\phi+\epsilon\phi_*(r).
\end{equation}
The second formula denotes axial rotation on the whole sphere;
no lift of the angular coordinate is required.  On
$\mathcal X_\epsilon\equiv\mathbb R_{v_\epsilon}
\times(M-\eta,M+\eta)\times\mathbb S^2$ the metric is
\begin{align}\label{eq:extremal-ef-metric}
 g={}&-\left(1-\frac{2Mr}{\Sigma}\right)dv_\epsilon^2
 +2\epsilon\,dv_\epsilon\,dr
 -\frac{4Mar\sin^2\theta}{\Sigma}\,dv_\epsilon\,d\psi_\epsilon
 -2\epsilon a\sin^2\theta\,dr\,d\psi_\epsilon\nonumber\\
 &+\Sigma\,d\theta^2
 +\left(r^2+a^2+\frac{2Ma^2r\sin^2\theta}{\Sigma}\right)
       \sin^2\theta\,d\psi_\epsilon^2.
\end{align}
This analytic metric is Lorentzian throughout the collar.  Its
determinant is $-\Sigma^2\sin^2\theta$ off the axis and
$-\Sigma^2/\mu^2$ in polar caps, so it remains nondegenerate at the
double horizon and its poles.

We construct the spacetime as follows.  For each $n\in\mathbb Z$,
take the full block
\begin{equation}\label{eq:extremal-block-chain}
 B_n\equiv
 \begin{cases}
 \mathbb R_t\times(M,\infty)\times\mathbb S^2,&n\text{ even},\\
 (\mathbb R_t\times(-\infty,M)\times\mathbb S^2)
       \setminus\{r=0,\mu=0\},&n\text{ odd}.
 \end{cases}
\end{equation}
Let $X_n$ be a copy of $\mathcal X_{\epsilon_n}$ with
$\epsilon_n\equiv(-1)^n$.  We glue the two sides of $X_n$ to the full radial
collars of $B_n$ and $B_{n+1}$ using
\eqref{eq:extremal-ef-maps}.  No further identifications are made.
Write $\mathcal K^{\mathrm{ext}}_{M,a}$ for this quotient.  Every
block has one ingoing and one outgoing complete horizon collar.
Both regular branches of the disc occur in every odd block; there
is no bifurcation sphere at the degenerate horizon.

\begin{lemma}[The extremal realization]\label{lem:extremal-realization}
The quotient $\mathcal K^{\mathrm{ext}}_{M,a}$ is a connected,
Hausdorff, second-countable, time-oriented analytic Kerr spacetime.
The functions $r,\mu$, the Killing fields $T,\Phi$, and the Carter
tensor are globally defined.  The horizon generators are complete
and every finite value of an ingoing or outgoing horizon coordinate
belongs to the spacetime.
\end{lemma}

\begin{proof}
The gluing maps are open, and hence the quotient map is open.
Each model embeds in the quotient: a point away from the horizon
has at most three representatives, whereas a horizon point has one.
We show that the equivalence relation is closed.  For a pair consisting
of a block point and a crossing point which approaches $r=M$, the
block point has no limit in the block.  All other limits lie in the
transition graph.  Adjacent crossings satisfy
\begin{equation}\label{eq:extremal-ef-transition}
 v_+-v_-=2r_*(r),\qquad
 \psi_+-\psi_-=2\phi_*(r)\pmod{2\pi}.
\end{equation}
Since $|r_*|\to\infty$, equivalent points cannot have finite horizon
limits in both crossings.  Nonadjacent models have no equivalent
points.  The relation is therefore closed, and the open product
quotient map gives a closed quotient diagonal.  It follows that the quotient is
Hausdorff.  The countable connected chain of analytic models gives
a second-countable atlas, in which the descended metric is analytic.

On each block $Z\equiv(r^2+a^2)T+a\Phi$ is timelike.  In $X_n$ put
$W_n\equiv Z-\epsilon_n\partial_r$, with $\partial_r$ taken at fixed
crossing coordinates.  The identities
$g(Z,\partial_r)=\epsilon_n\Sigma$ and $g(\partial_r,\partial_r)=0$
give
\begin{equation}\label{eq:extremal-time-orientation}
 g(W_n,W_n)=-\Sigma(\Delta+2)<0,\qquad
 g(W_n,Z)=-\Sigma(\Delta+1)<0.
\end{equation}
The resulting time orientations therefore agree on all overlaps.
The fields $T=\partial_{v_\epsilon}$ and
$\Phi=\partial_{\psi_\epsilon}$ agree on the overlaps as well.  Substitution into
\eqref{eq:app-ky-form} cancels all extra $dr$ terms and gives
\begin{equation}\label{eq:extremal-ky-form}
 f\equiv a\cos\theta\,dr\wedge
       (dv_\epsilon-a\sin^2\theta\,d\psi_\epsilon)
   +r\sin\theta\,d\theta\wedge
       ((r^2+a^2)d\psi_\epsilon-a\,dv_\epsilon).
\end{equation}
This form is analytic at the horizon and poles; its Killing-Yano
identity extends from the dense block set.  Squaring it gives the
global Carter tensor with the stated normalization.

On $r=M$, the generator $K\equiv T+(a/(2M^2))\Phi$ has the complete flow
$(v,\omega)\mapsto(v+s,\mathscr R_{as/(2M^2)}\omega)$.
The function $g(K,K)$ has a double zero in $r-M$, hence the Killing
identity gives $\nabla_KK=-\tfrac12\nabla g(K,K)=0$ on the horizon.
Its complete orbits are therefore affinely parametrized null geodesics.
The full $v$-range includes every finite horizon coordinate.
\end{proof}

\begin{lemma}[Critical extremal geodesics]\label{lem:critical-extremal-geodesics}
Let $M>0$, $a^2=M^2$, and let a unit timelike Kerr geodesic have
$P(M)=0$.  It cannot approach $r=M$ in finite proper time.
\end{lemma}

\begin{proof}
The critical relation gives $aL=2M^2E$.  We thus obtain the
factorization
\begin{equation}\label{eq:critical-extremal-factorization}
 P(r)=E(r-M)(r+M),\qquad
 R(r)=(r-M)^2\bigl(E^2(r+M)^2-r^2-\mathcal Q\bigr).
\end{equation}
For $|r-M|<M/2$, one has $\Sigma\ge r^2\ge M^2/4$.
The factor in parentheses is bounded on this interval.  Along
the geodesic, the radial equation therefore gives
\begin{equation}\label{eq:critical-extremal-gronwall}
 |\dot r|\le C|r-M|.
\end{equation}
If $r(s)\to M$ as $s\uparrow\Lambda<\infty$, integrating this
inequality backwards from the continuous value $r(\Lambda)=M$
and applying Gronwall's inequality forces $r\equiv M$ on a
terminal interval.  A timelike curve cannot lie in the null
hypersurface $r=M$.  This contradiction proves the assertion.
\end{proof}

\begin{corollary}[Extremal Kerr]\label{cor:extremal-inextendibility}
For $M>0$ and $a^2=M^2$, the spacetime
$\mathcal K^{\mathrm{ext}}_{M,a}$ and every connected smooth
covering spacetime are $C^0$-inextendible.  Every inextendible
unit-speed timelike geodesic with a finite proper-time endpoint
in any of these spacetimes satisfies $\Sigma\to0$ at that endpoint.
\end{corollary}

\begin{proof}
Proposition~\ref{prop:separation} follows from the global Carter
tensor.  On a finite proper-time interval,
Lemmas~\ref{lem:radial-bounded} and~\ref{lem:finite-horizons}, with
$\mathscr H_r=\{M\}$ in the latter, bound $r$ and the number of
horizon events, respectively.
Suppose $\Sigma\ge\rho>0$ on a terminal interval.  Then $r$ is
Lipschitz, has a limit $r_\infty$, and the tail after the last horizon
event lies in one full block.  If $r_\infty\ne M$, the separated
equations bound the time and round-sphere velocities, giving a
convergent base point and precompact tangent lift in that block.
Lemma~\ref{lem:escape} extends the geodesic.

If $r_\infty=M$, Lemma~\ref{lem:critical-extremal-geodesics} gives
$P(M)\ne0$.  Hence $R(M)=P(M)^2>0$ and the radial velocity has fixed
sign on a tail.  The rationalization used in
Lemma~\ref{lem:horizon-continuation}\textup{(a)} depends only on
\begin{equation}
 (P+\epsilon\sqrt R)(P-\epsilon\sqrt R)
 =\Delta(r^2+\mathcal Q),\qquad P(M)\ne0,
\end{equation}
and does not require the zero of $\Delta$ to be simple.  It bounds
the time and angular velocities in the appropriate ingoing or outgoing
chart.  The adjacent full collar contains the limiting finite
coordinates, and its metric \eqref{eq:extremal-ef-metric} is
nondegenerate.  The tangent lift is precompact there, so the geodesic
again extends.
We have therefore excluded a positive lower bound for $\Sigma$
near a finite inextendible endpoint.  The rescaled-time argument in
Theorem~\ref{thm:ring-localization} now gives $\Sigma\to0$.

The ring construction applies since $\Delta(0)=a^2>0$.  The
boundary-geodesic and causal-length criteria exclude a proper
$C^0$ extension in either time orientation.  Path lifting gives localization on connected covers.
Lemma~\ref{lem:covering-causal-length} then transfers the homotopy
lengths, and the same proof gives inextendibility.
\end{proof}

\section{Perturbations of the ring geometry}
\label{sec:ring-perturbations}

We now estimate the perturbations that preserve the causal
construction near the ring. All bounds in this section are expressed
in the coordinate norm defined below.

The Kerr coframe yields transverse timelike connectors with
stationary-axial displacement
$O(\sqrt\Sigma)$. Theorem~\ref{thm:conformal-class-ring} gives the two independent
orbit directions and the four nonnegative durations that cancel this
displacement.
\subsection{Perturbations controlled by the transverse radius}

Put $y\equiv a\cos\theta$ and $R\equiv\sqrt\Sigma$ on a full
punctured ring collar
$\mathbb R_t\times\{0<r^2+y^2<\delta_0^2\}\times\mathbb S^1_\phi$,
where $\delta_0<|a|$ and $\Delta>0$ on $|r|<\delta_0$.
We use the coordinate Euclidean metric
\begin{equation}\label{eq:ring-coordinate-norm}
 e\equiv dt^2+dr^2+dy^2+d\phi^2,
 \qquad |V|_{\mathrm E}^2\equiv e(V,V),
\end{equation}
where $d\phi$ is the invariant one-form on the circle.  For a symmetric
bilinear form $h$, let
$\|h\|_{\mathrm E}\equiv\sup_{|V|_{\mathrm E}=1}|h(V,V)|$.
All asymptotic bounds in this subsection are taken in these
coordinates.  Put
\begin{equation}\label{eq:ring-perturbation-constants}
 A\equiv\sqrt{a^2+1},\quad B\equiv|a|+2,\quad K\equiv2a^2+B^2,
 \quad L\equiv1+\frac{MA}{a^2},\quad
 \eta_{M,a}\equiv\frac1{8KL}.
\end{equation}

For the unrescaled metric, the ring obstruction persists under a
small $O(\sqrt\Sigma)$ perturbation. Here one negative-sheet axial
arc of width $A_0$ has exact Kerr length
\begin{equation}\label{eq:intro-axial-length}
 A_0\sqrt{\frac{2Ma^2}{\rho}-a^2-\rho^2}\longrightarrow+\infty
 \qquad(\rho\downarrow0).
\end{equation}
An earlier timelike segment absorbs the connectors' displacement.
A modified lift requires no derivatives of $h$
(Theorem~\ref{thm:stationary-ring-perturbation}); an averaged axial
criterion treats conformal factors for this larger class
(Theorem~\ref{thm:conformal-ring-obstruction}).

\begin{theorem}[Persistence of the local ring obstruction]
\label{thm:stationary-ring-perturbation}
Let $M>0$ and $a\ne0$, and let $\widehat g\equiv g_{M,a}+h$ be a continuous
time-oriented Lorentzian metric of signature $(-,+,+,+)$ on a full
punctured ring collar $\mathcal N_{\delta_0}$.  Suppose that $h$ is
invariant under stationary and axial translations and that
\begin{equation}\label{eq:ring-weighted-bound}
 \limsup_{\Sigma\to0}
 \frac{\|h\|_{\mathrm E}}{\sqrt\Sigma}<\eta_{M,a}.
\end{equation}
Every future-directed $C^1$ timelike curve
$\gamma:[0,\Lambda)\to\mathcal N_{\delta_0}$ satisfying
$\Sigma(\gamma(t))\to0$ obeys
\begin{equation}\label{eq:perturbed-infinite-length}
 \forall s<\Lambda\ \exists t_0\in(s,\Lambda)\quad
 \mathcal L^{\widehat g}_\gamma(s,t)=+\infty
 \quad(t_0<t<\Lambda).
\end{equation}
Given $0<\delta<\delta_0$, the threshold $t_0$ may be chosen so
that all varying portions of these homotopies lie in $\mathcal N_\delta$.
Consequently, no such curve has an endpoint in a $C^0$ extension of
$(\mathcal N_{\delta_0},\widehat g)$.
\end{theorem}

\begin{proof}
Shrink the collar so that
$a^2/2\le\Delta\le2a^2$ and $a^2-y^2\ge3a^2/4$.
Write $h_{ij}$ for the components in the frame dual to
$(dt,dr,dy,d\phi)$, and put
\begin{equation}\label{eq:ring-perturbation-majorant}
 b\equiv\left(\Sigma^2+\sum_{i,j=1}^4h_{ij}^2\right)^{1/2}.
\end{equation}
The function $b$ is positive, continuous, stationary and axisymmetric.
The Frobenius norm lies between the operator norm and twice that norm, so
\begin{equation}\label{eq:ring-perturbation-majorant-bounds}
 |h(U,V)|\le b|U|_{\mathrm E}|V|_{\mathrm E},\qquad
 \limsup_{R\to0}\frac bR
 \le2\limsup_{R\to0}\frac{\|h\|_{\mathrm E}}R
 <\frac1{4KL}.
\end{equation}
Define the continuous positive function
\begin{equation}
 c\equiv2+\frac{Kb}{\Sigma}.
\end{equation}
For $\varepsilon\in\{-1,1\}$, define the stationary-axial field
$Z_c$ by
\begin{equation}\label{eq:perturbed-orbit-lift}
 \alpha(Z_c)=\varepsilon c\frac{\Sigma}{\sqrt\Delta},
 \qquad
 \beta(Z_c)=-\varepsilon\operatorname{sgn}(a)(c-1)\Sigma.
\end{equation}
Solving the two equations gives
\begin{align}
 dt(Z_c)&=\varepsilon\left[|a|\sin^2\theta
  +c\left(\frac{r^2+a^2}{\sqrt\Delta}
                   -|a|\sin^2\theta\right)\right],\nonumber\\
 d\phi(Z_c)&=\varepsilon\left[\operatorname{sgn}(a)
  +c\left(\frac a{\sqrt\Delta}-\operatorname{sgn}(a)\right)\right].
 \label{eq:perturbed-orbit-components}
\end{align}
Let $D(r,y)$ denote the vector whose components are the two
expressions in parentheses.
It is smooth at $(0,0)$, vanishes there, and satisfies
\begin{equation}
 dD_{(0,0)}(u,v)
 =\left(\frac{M}{|a|}u,
         \frac{M\operatorname{sgn}(a)}{a^2}u\right).
\end{equation}
Its derivative has operator norm $MA/a^2<L$, so $|D|\le LR$
on a sufficiently small disc.  Restrict the collar further so that
$b/R<1/(4KL)$ and $2LR<1/4$.  Since
$|(|a|\sin^2\theta,\operatorname{sgn}(a))|\le A$,
\begin{equation}\label{eq:perturbed-variable-c-limit}
 |Z_c|_{\mathrm E}
 \le A+cLR=A+2LR+KL\frac bR<A+\frac12<B.
\end{equation}
The exact coframe identity yields
\begin{equation}\label{eq:perturbed-orbit-norm}
 g_{M,a}(Z_c,Z_c)
 =\Sigma[-c^2+\sin^2\theta(c-1)^2]
 \le-(2c-1)\Sigma.
\end{equation}
Therefore
$\widehat g(Z_c,Z_c)\le-3\Sigma-(2K-B^2)b<0$.
Choose $\varepsilon$ so that $Z_c$ is future directed on this connected
collar.

For a nonzero transverse velocity $W=(\dot r,\dot y)$, write
\begin{equation}
 k(W)^2\equiv\frac{\dot r^2}{\Delta}
            +\frac{\dot y^2}{a^2-y^2},
 \qquad V\equiv W+k(W)Z_c.
\end{equation}
The transverse and orbit coordinates are orthogonal for both $e$ and
$g_{M,a}$.  Since $|W|_{\mathrm E}^2\le2a^2k(W)^2$,
we have $|V|_{\mathrm E}^2\le Kk(W)^2$.  Consequently
\begin{align}
 \widehat g(V,V)
 &\le-2(c-1)\Sigma k(W)^2+Kb\,k(W)^2
   =-(2\Sigma+Kb)k(W)^2<0,\nonumber\\
 \widehat g(V,Z_c)
 &\le-(2c-1)\Sigma k(W)+B\sqrt K\,b\,k(W)
   \le-(3\Sigma+Kb)k(W)<0.
 \label{eq:perturbation-lift-estimate}
\end{align}
Here $B\sqrt K\le K$.  Thus $V$ is future timelike and its orbit
velocity satisfies
\begin{equation}\label{eq:perturbed-orbit-bound}
 |\dot t|+|\dot\phi|
 \le\sqrt2 Bk(W)\le\frac{2B}{|a|}|W|_{\mathbb R^2}.
\end{equation}
Thus the orbit velocity is bounded even when $c$ diverges.  For a
prescribed piecewise smooth transverse path, we obtain the orbit
coordinates by integrating $k(W)Z_c$.  The path lies in a compact
annulus, where continuity gives continuous dependence on its
truncation.  This construction requires only continuity of $h$.

At $(r,y)=(-\rho,0)$ put
$Y_-\equiv-\varepsilon\operatorname{sgn}(a)\Phi$ and
$b_\rho\equiv b(-\rho,0)$.  Then $b_\rho=O(\rho)$ and
\begin{equation}\label{eq:perturbed-axial-bound}
 \widehat g(Y_-,Y_-)
 \le a^2+\rho^2+b_\rho-\frac{2Ma^2}{\rho}
 \le-\frac{Ma^2}{\rho}
\end{equation}
for all sufficiently small $\rho$.  The exact pairing is
\begin{align}
 g_{M,a}(Y_-,Z_c)
 &=-c|a|\sqrt{a^2+2M\rho+\rho^2}
                +(c-1)(a^2+\rho^2)\nonumber\\
 &\le-a^2+(c-1)\rho^2
   =-a^2+\rho^2+Kb_\rho.
 \label{eq:perturbed-axial-orientation}
\end{align}
Since $|h(Y_-,Z_c)|\le Bb_\rho$, the perturbed pairing is negative
for small $\rho$.  Thus $Y_-$ is future timelike, and an arc of
angular width $A_0$ has length at least $A_0\sqrt{Ma^2/\rho}$.

Stationary and axial translations preserve $\widehat g$ and $Z_c$,
and hence the chosen time orientation.  We apply the lift just
constructed to the transverse paths in
Lemma~\ref{lem:radial-attachment}.  Estimate~\eqref{eq:perturbed-orbit-bound} gives the same
uniform displacement bound $C_1(R+A_0)$.  On a fixed compact timelike
segment of $\gamma$, the cutoff variation
\eqref{eq:timelike-endpoint-variation} is future timelike for all
sufficiently small displacements.  Choose $A_0>0$ and then a
terminal radius threshold $R_0>0$ so small that the displacements lie in this
neighborhood.  Concatenation~\eqref{eq:fixed-endpoint-radial-concatenation}
now gives a fixed-endpoint causal homotopy whose final curve contains
an arc of length at least $A_0\sqrt{Ma^2/\rho}$.  Keeping the endpoints and
$A_0$ fixed and letting $\rho\downarrow0$ proves
\eqref{eq:perturbed-infinite-length}.  All collar restrictions may be
made inside $\mathcal N_\delta$, and the earlier portion of $\gamma$
remains fixed.  Theorem~\ref{thm:causal-length-criterion} excludes an
extension endpoint.
\end{proof}

In particular, the theorem applies when
$\|h\|_{\mathrm E}=o(\sqrt\Sigma)$, including the bounds
$C\Sigma^p$ with $p>1/2$ and arbitrary finite $C$, and
$C\sqrt\Sigma/|\log\Sigma|$.  The signature
assumption is essential: $h=-2\Sigma\Delta^{-1}dr^2$ satisfies
the bound but produces a second negative direction.  No field
equations enter the proof.  The theorem excludes extension endpoints
at the ring.

The perturbation preserves the leading axial length and its first
correction.  Stationarity and axial invariance give
\begin{equation}\label{eq:perturbed-axial-asymptotic}
 \sqrt\rho\,L_{\widehat g}(\text{axial arc of width }A_0)
 =A_0\sqrt{2Ma^2-a^2\rho-\rho^3
                    -\rho h(\Phi,\Phi)|_{(-\rho,0)}}.
\end{equation}
Here $h(\Phi,\Phi)|_{(-\rho,0)}=O(\rho)$, so Taylor's formula gives
\begin{equation}
 \frac{\sqrt\rho\,L_{\widehat g}(\text{axial arc of width }A_0)}
      {A_0\sqrt{2M}\,|a|}
 =1-\frac{\rho}{4M}+O(\rho^2).
\end{equation}
Under the stronger hypothesis $\|h\|_{\mathrm E}=o(\sqrt\Sigma)$,
we have $h(\Phi,\Phi)|_{(-\rho,0)}=o(\rho)$.  The second correction
then agrees with that of Kerr as well:
\begin{equation}\label{eq:perturbed-axial-two-corrections}
 \frac{\sqrt\rho\,L_{\widehat g}(\text{axial arc of width }A_0)}
      {A_0\sqrt{2M}\,|a|}
 =1-\frac{\rho}{4M}-\frac{\rho^2}{32M^2}+o(\rho^2)
 \qquad(\rho\downarrow0).
\end{equation}

\subsection{Sharp bounds for conformal homotopy lengths}

We next require the ring obstruction to persist under every positive
continuous conformal rescaling. For this conclusion, we impose a stronger
perturbation bound. We first prove the sharp leading bound and then treat
the critical coefficient and its radial refinement.

At the leading coefficient, the pointwise bound
$\|h\|_{\mathrm E}\le\Sigma/(1+a^2)$ still suffices, whereas equality
of the limiting ratio alone does not. For the refined bound
$\|h\|_{\mathrm E}\le(\Sigma+\kappa\Sigma^2)/(1+a^2)$, the critical
coefficient is $\kappa_*\equiv M^2/(2a^4)$. At $\kappa=\kappa_*$,
the obstruction persists for $a^2\ge2M^2$; for $a^2<2M^2$ the
rank-one example has finite length supremum in every fixed-endpoint
causal homotopy class after a suitable conformal rescaling
(Proposition~\ref{prop:critical-conformal-perturbation} and
Corollary~\ref{cor:critical-equality}). These examples show that the bound for infinite homotopy lengths
is sharp; they give no conclusion about geometric extendibility.
An exact radial bound, obtained by minimizing the helical norm on each
transverse circle, contains these coefficient tests as expansions and
requires strict control only along radii tending to zero
(Proposition~\ref{prop:exact-radial-envelope}).

\begin{theorem}[A sharp perturbation bound for conformal homotopy lengths]
\label{thm:sharp-conformal-perturbation}
Let $M>0$, $a\ne0$, and put $C_a\equiv1+a^2$.
\begin{enumerate}[label=\textup{(\roman*)}]
 \item Let $\widehat g\equiv g_{M,a}+h$ be a continuous time-oriented
 Lorentzian metric on a full punctured ring collar.  Suppose that
 $h$ is invariant under stationary and axial translations and
 \begin{equation}\label{eq:sharp-conformal-perturbation-bound}
  \limsup_{\Sigma\to0}\frac{\|h\|_{\mathrm E}}{\Sigma}
  <\frac1{C_a}.
 \end{equation}
 Then Theorem~\ref{thm:conformal-class-ring} holds with $g_{M,a}$
 replaced by $\widehat g$: every positive continuous conformal
 factor, including factors on connected covers which do not descend,
 gives infinite fixed-endpoint causal-homotopy lengths along every
 timelike ring approach, and no such approach has a $C^0$ extension
 endpoint.
 \item For every $c>1$, there are a smooth stationary and axisymmetric
 perturbation $h_c$ and a smooth positive conformal factor $\Omega_c$
 on a sufficiently small full ring collar such that
 \begin{equation}\label{eq:sharp-conformal-counterexample-size}
  \frac{\|h_c\|_{\mathrm E}}{\Sigma}=\frac{c}{C_a},
 \end{equation}
 but every fixed-endpoint causal-homotopy class has finite length
 supremum for $\Omega_c^2(g_{M,a}+h_c)$.
\end{enumerate}
The coefficient $1/C_a$ is therefore optimal in the strict bound
which ensures infinite lengths for every positive conformal factor.
At the critical coefficient, pointwise bounds and equality of the
limiting ratio behave differently; see
Proposition~\ref{prop:critical-conformal-perturbation}.  Part
\textup{(ii)} concerns finiteness of these lengths and does not
determine $C^0$-extendibility of the examples.
\end{theorem}

\begin{proof}
For \textup{(i)}, choose $c_0<1/C_a$ such that
$\|h\|_{\mathrm E}\le c_0\Sigma$ on a smaller collar, and put
$\eta\equiv1-c_0C_a>0$, $\beta\equiv\eta/8$, and $A\equiv|a|$.
Since the perturbation satisfies \eqref{eq:ring-weighted-bound},
Theorem~\ref{thm:stationary-ring-perturbation} gives the future
field $Z_c$ and bounded transverse lifts.
Choose $\varepsilon$ as in that construction.  For small $\ell>0$,
define
\begin{equation}\label{eq:sharp-paired-perturbed-fields}
 d_1\equiv\frac{\beta\ell^2}{A},\qquad d_2\equiv2d_1,
 \qquad
 \rho\equiv\min\left\{1,\frac{M}{4A}\right\}d_1,
 \qquad
 K_i\equiv\varepsilon\bigl((A-d_i)T+\operatorname{sgn}(a)\Phi\bigr).
\end{equation}
At $X_+\equiv(0,\ell)$, one has $|K_i|_{\mathrm E}^2\le C_a$ and
\begin{equation}\label{eq:sharp-perturbed-positive-pair}
 \widehat g(K_i,K_i)
 \le-\ell^2+2Ad_i-d_i^2+c_0C_a\ell^2
 \le-\frac\eta2\ell^2<0.
\end{equation}
At $X_-\equiv(-\rho,0)$, the inequalities
$\rho\le d_1\le d_i$ and $\rho\le Md_1/(4A)$ give
\begin{equation}\label{eq:sharp-perturbed-negative-pair}
 \widehat g(-K_i,-K_i)
 \le-d_i^2-\frac{2Md_i^2}{\rho}+2Ad_i+\rho^2+c_0C_a\rho^2
 \le-6Ad_i+c_0C_a\rho^2<0
\end{equation}
for sufficiently small $\ell$.

We now determine the time orientations.  The identities $g(T,T)=-1$ and
$g(T,\Phi)=0$ hold at $X_+$. Formula~\eqref{eq:perturbed-orbit-components}
gives
\begin{equation}
 g(\varepsilon T,Z_c)
 =-A-\frac{(c-1)\ell^2}{A}\le-A.
\end{equation}
The perturbation is $O(\ell^2)$, so $\varepsilon T$ is future
timelike. Since $\widehat g(K_i,\varepsilon T)=-A+O(\ell^2)<0$,
so are $K_i$. At $X_-$, the field
$Y_-\equiv-\varepsilon\operatorname{sgn}(a)\Phi$ is future timelike
by \eqref{eq:perturbed-axial-orientation}, and
\begin{equation}
 g(-K_i,Y_-)=a^2+\rho^2-\frac{2MAd_i}{\rho}
 \le-7a^2+\rho^2.
\end{equation}
The perturbation is $O(\rho^2)$, so $-K_i$ are future timelike.
Since $K_1,K_2$ are independent, any orbit displacement can be
cancelled by following these fields and their negatives for
nonnegative durations.  Combining this with the bounded transverse
lifts and the closed-loop construction in
Theorem~\ref{thm:conformal-class-ring} proves \textup{(i)},
including the assertion on covers.

For \textup{(ii)}, let
\begin{equation}\label{eq:sharp-rank-one-perturbation}
 U\equiv |a|T+\operatorname{sgn}(a)\Phi,\qquad
 \zeta\equiv |a|\,dt+\operatorname{sgn}(a)\,d\phi,
 \qquad
 h_c\equiv\frac{c\Sigma}{C_a^2}\,\zeta^2.
\end{equation}
Since $|\zeta|_{\mathrm E}^2=C_a$, the rank-one form has
$\|h_c\|_{\mathrm E}=c\Sigma/C_a$ and $h_c(U,U)=c\Sigma$.
The exact Kerr norm is
\begin{equation}\label{eq:sharp-helical-norm}
 g(U,U)=r^2-y^2-\frac{r^2y^2}{a^2}
              +\frac{2Mr\,y^4}{a^2\Sigma}.
\end{equation}
Writing $R\equiv\sqrt\Sigma$, we obtain
\begin{equation}\label{eq:sharp-helical-positive}
 \frac{(g+h_c)(U,U)}{\Sigma}
 \ge c-1-\frac{R^2}{4a^2}-\frac{2MR}{a^2}>0
\end{equation}
on a sufficiently small collar.

Let $G_c$ denote the $(t,\phi)$ matrix of $g+h_c$.
The Kerr orbit matrix has determinant
$-\Delta(1-y^2/a^2)$ and entries $O(R^{-1})$, while $h_c=O(R^2)$.
Consequently,
\begin{equation}
 \det G_c=-a^2+O(R)<0.
\end{equation}
The transverse block is unchanged and positive definite.
Thus $g+h_c$ is Lorentzian.  The closed one-form
\begin{equation}
 \tau\equiv dt-a\,d\phi
\end{equation}
annihilates $U$, and the inverse of the two-dimensional orbit matrix
gives
\begin{equation}\label{eq:sharp-temporal-covector}
 (g+h_c)^{-1}(\tau,\tau)
 =\frac{(g+h_c)(U,U)}{\det G_c}<0.
\end{equation}
Choose the time orientation for which $\tau(V)>0$ on nonzero
future causal vectors, and define
\begin{equation}\label{eq:sharp-conformal-counterexample}
 \Omega_c^2\equiv-(g+h_c)^{-1}(\tau,\tau)>0,
 \qquad g_c\equiv\Omega_c^2(g+h_c).
\end{equation}
These coefficients are smooth on the punctured collar, and
$g_c^{-1}(\tau,\tau)=-1$.
By the reverse Cauchy-Schwarz inequality,
\begin{equation}\label{eq:sharp-conformal-length-bound}
 \sqrt{-g_c(V,V)}\le\tau(V),\qquad
 L_{g_c}(\sigma)\le\int_\sigma\tau
\end{equation}
for future causal vectors and curves.  A lift of the angular
coordinate along a fixed-endpoint homotopy identifies $\int\tau$
with the fixed endpoint difference of $t-a\phi$.  This proves
\textup{(ii)}.  The bounded transverse lifts also give timelike
curves approaching the ring in these examples.
\end{proof}

\begin{proposition}[A sharp radial bound]\label{prop:exact-radial-envelope}
Put $C_a\equiv1+a^2$, $b\equiv M/a^2$, and
\begin{equation}\label{eq:critical-helical-circle}
 F_R(x)\equiv\left(2-\frac{R^2(1-x^2)}{a^2}\right)x^2
                 +2bRx(1-x^2)^2,
 \qquad -1\le x\le1.
\end{equation}
For sufficiently small $R>0$, define
\begin{equation}\label{eq:exact-radial-envelope}
 m(R)\equiv\min_{[-1,1]}F_R,
 \qquad E(R)\equiv\frac{R^2}{C_a}\bigl(1-m(R)\bigr).
\end{equation}
\begin{enumerate}[label=\textup{(\roman*)}]
\item Suppose $\widehat g\equiv g_{M,a}+h$ is continuous, time oriented,
Lorentzian, stationary and axisymmetric on a full punctured ring
collar, with $\|h\|_{\mathrm E}=O(\Sigma)$.
If radii $R_j\downarrow0$ satisfy
\begin{equation}\label{eq:exact-envelope-sequence}
 \max_{\Sigma=R_j^2}\|h\|_{\mathrm E}<E(R_j),
\end{equation}
then the conclusions of Theorem~\ref{thm:conformal-class-ring} hold
for $\widehat g$, including arbitrary positive continuous conformal
factors on connected covers.
\item Let $q$ be smooth for small $R>0$, with $q(R)=O(R^2)$ and
$q(R)>E(R)$. The perturbation
\begin{equation}\label{eq:exact-envelope-rank-one}
 h_q\equiv\frac{q(\sqrt\Sigma)}{C_a}\,\zeta^2,
 \qquad \zeta\equiv |a|\,dt+\operatorname{sgn}(a)\,d\phi,
\end{equation}
has $\|h_q\|_{\mathrm E}=q(\sqrt\Sigma)$ and admits a smooth
positive conformal factor for which every fixed-endpoint
causal-homotopy class has finite length supremum.
\end{enumerate}
Both assertions use strict comparison with $E$.  The proposition
does not treat $q\equiv E$.
\end{proposition}

\begin{proof}
Retain $U\equiv |a|T+\operatorname{sgn}(a)\Phi$.
On $\Sigma=R^2$, put $x\equiv r/R$. Formula
\eqref{eq:sharp-helical-norm} reads $g(U,U)=R^2(F_R(x)-1)$.
For small $R$, $F_R''=4+O(R)>0$ uniformly on $[-1,1]$.
The implicit function theorem gives a unique analytic minimizer
$x_*(R)=-bR/2+O(R^3)$.  The identity $F_{-R}(-x)=F_R(x)$ shows
that the minimizer is odd in $R$, so $m$ and $E$ are analytic and
even near zero.  Moreover, $F_R(0)=0$ implies $m(R)\le0$ and
$E(R)>0$.

For \textup{(i)}, at
$X_j\equiv(R_jx_*(R_j),R_j\sqrt{1-x_*(R_j)^2})$,
\begin{equation}\label{eq:exact-envelope-negative}
 \widehat g(U,U)|_{X_j}
 \le C_a\left(\max_{\Sigma=R_j^2}\|h\|_{\mathrm E}-E(R_j)\right)<0.
\end{equation}
Since $h=O(\Sigma)$, the bounded transverse lifts of
Theorem~\ref{thm:stationary-ring-perturbation} apply.  At $X_j$,
the corresponding future field satisfies
$Z_c=\varepsilon U+O(R_j^2)$, and the orbit coefficients of Kerr
are bounded.  We also have
\[
 \widehat g(T,T)=-1-\frac{M^2}{a^2}+O(R_j^2),
 \qquad \widehat g(U,T)=-|a|+O(R_j^2).
\]
Hence $\varepsilon T$, and then $\varepsilon U$, are future
timelike at $X_j$. Choose $d_1>0$ sufficiently small, put
$d_2\equiv2d_1$, and set
$K_i\equiv\varepsilon(U-d_iT)$.
Both $K_i$ remain future timelike at $X_j$.
At $X_-\equiv(-\rho,0)$, with
$\rho\equiv\min\{1,M/(4|a|)\}d_1$, their negatives are future
timelike by \eqref{eq:sharp-perturbed-negative-pair} and its
orientation calculation: their norms are at most
$-6|a|d_i+O(\rho^2)$, and their pairings with $Y_-$ are at most
$-7a^2+O(\rho^2)$. By reducing $d_1$, we obtain the four orbit directions in any
prescribed collar.  The two positive directions are independent.
The bounded transverse lifts and the closed-loop construction of
Theorem~\ref{thm:conformal-class-ring} now give the assertion.
Only convergence of the radii to zero is required; the strict
margins in \eqref{eq:exact-envelope-sequence} need not have a
uniform positive lower bound.

For \textup{(ii)}, the rank-one norm gives the asserted size, and
\begin{equation}\label{eq:exact-envelope-positive}
 \min_{\Sigma=R^2}(g+h_q)(U,U)=C_a\bigl(q(R)-E(R)\bigr)>0.
\end{equation}
Since $h_q=O(R^2)$, the orbit determinant is $-a^2+O(R)<0$;
the unchanged transverse block is positive. Thus $g+h_q$ is
Lorentzian on a smaller collar, and $\tau\equiv dt-a\,d\phi$
is timelike by \eqref{eq:sharp-temporal-covector}. Choose the
orientation for which $\tau$ is positive on future causal vectors.
Set $\Omega_q^2\equiv-(g+h_q)^{-1}(\tau,\tau)>0$.
The estimate \eqref{eq:sharp-conformal-length-bound}, with $q$
in place of $c$, proves \textup{(ii)}.
\end{proof}

\Needspace{12\baselineskip}
\begin{proposition}[The critical perturbation coefficient]
\label{prop:critical-conformal-perturbation}
Put $C_a\equiv1+a^2$ and $\kappa_*\equiv M^2/(2a^4)$.
\begin{enumerate}[label=\textup{(\roman*)}]
 \item Let $\widehat g\equiv g_{M,a}+h$ be a continuous time-oriented
 Lorentzian metric on a full punctured ring collar, with $h$
 stationary and axisymmetric.  If, for some $\kappa<\kappa_*$,
 \begin{equation}\label{eq:critical-perturbation-upper}
  \|h\|_{\mathrm E}\le\frac{\Sigma+\kappa\Sigma^2}{C_a}
 \end{equation}
 near the ring, then the conclusions of
 Theorem~\ref{thm:conformal-class-ring} hold with $g_{M,a}$ replaced
 by $\widehat g$, including arbitrary positive continuous conformal
 factors on connected covers.
 \item For each $\kappa>\kappa_*$, the smooth rank-one perturbation
 \begin{equation}\label{eq:critical-rank-one-perturbation}
  h_\kappa\equiv\frac{\Sigma+\kappa\Sigma^2}{C_a^2}\,\zeta^2,
  \qquad
  \zeta\equiv|a|\,dt+\operatorname{sgn}(a)\,d\phi,
 \end{equation}
 on a sufficiently small full ring collar admits a smooth positive
 conformal factor for which every fixed-endpoint causal-homotopy
 class has finite length supremum.  Moreover,
 \begin{equation}\label{eq:critical-perturbation-limit}
  \frac{\|h_\kappa\|_{\mathrm E}}{\Sigma}
  =\frac{1+\kappa\Sigma}{C_a}\longrightarrow\frac1{C_a}.
 \end{equation}
\end{enumerate}
Thus the pointwise bound $\|h\|_{\mathrm E}\le\Sigma/C_a$ preserves
the conformal ring obstruction, whereas the asymptotic condition
$\limsup\|h\|_{\mathrm E}/\Sigma\le1/C_a$ does not imply infinite
causal-homotopy lengths for all conformal factors.
This distinction concerns causal-homotopy lengths and does not
determine geometric $C^0$-extendibility of the examples in
\textup{(ii)}.
\end{proposition}

\begin{proof}
Write $U\equiv |a|T+\operatorname{sgn}(a)\Phi$,
$R\equiv\sqrt\Sigma$, and $b\equiv M/a^2$.
For $x\ge0$ and $R<|a|$, $F_R(x)\ge0$.
For $x<0$, completing the square gives
\begin{equation}\label{eq:critical-helical-lower}
 F_R(x)\ge\left(2-\frac{R^2}{a^2}\right)x^2+2bRx
 \ge-\frac{b^2R^2}{2-R^2/a^2}.
\end{equation}
Evaluation at $x\equiv-bR/2$ supplies the matching upper bound:
\begin{equation}\label{eq:critical-helical-minimum}
 m(R)=-\frac{b^2}{2}R^2+O(R^4),\qquad
 E(R)=\frac{R^2+\kappa_*R^4+O(R^6)}{C_a}.
\end{equation}
Consequently $(R^2+\kappa R^4)/C_a$ lies below $E(R)$ for
$\kappa<\kappa_*$ and above it for $\kappa>\kappa_*$.
Proposition~\ref{prop:exact-radial-envelope} proves both assertions.

A negative direction in \textup{(i)} can be exhibited explicitly.
Take
\begin{equation}\label{eq:critical-shifted-target}
 X_R\equiv\left(-\frac b2R^2,
                  R\sqrt{1-\frac{b^2R^2}{4}}\right).
\end{equation}
The norm bound gives
\begin{equation}\label{eq:critical-target-negative}
 \widehat g(U,U)|_{X_R}
 \le(\kappa-\kappa_*)R^4+O(R^6)<0.
\end{equation}
For \textup{(ii)}, the rank-one identity gives
\eqref{eq:critical-perturbation-limit} and
\begin{equation}\label{eq:critical-example-positive}
 \min_{\Sigma=R^2}
 \frac{(g+h_\kappa)(U,U)}{\Sigma^2}
 =\kappa-\kappa_*+O(R^2)>0.
\end{equation}
The positive conformal factor is
\begin{equation}\label{eq:critical-normalized-conformal-factor}
 \Omega_\kappa^2\equiv-(g+h_\kappa)^{-1}(\tau,\tau)>0,
\end{equation}
and $L_{\Omega_\kappa^2(g+h_\kappa)}(\sigma)\le\int_\sigma\tau$
on every future causal curve, as before.
\end{proof}

For an explicit example at the limiting coefficient, take
$\kappa=M^2/a^4$.
In this case \eqref{eq:critical-helical-lower} gives the direct bound
\begin{equation}\label{eq:critical-example-explicit}
 (g+h_\kappa)(U,U)
 \ge\frac{M^2\Sigma^2}{a^4}
       \frac{1-\Sigma/a^2}{2-\Sigma/a^2}>0
 \qquad(0<\Sigma<a^2).
\end{equation}
Hence a correction of order $\Sigma^2$ to the critical rank-one
perturbation changes the conclusion about infinite lengths.

\Needspace{8\baselineskip}
\begin{corollary}[The critical equality]\label{cor:critical-equality}
At $\kappa=\kappa_*\equiv M^2/(2a^4)$, the conclusions of
Proposition~\ref{prop:critical-conformal-perturbation}\textup{(i)} hold
if $a^2\ge 2M^2$, and those of
Proposition~\ref{prop:critical-conformal-perturbation}\textup{(ii)}
hold if $a^2<2M^2$.
\end{corollary}

\begin{proof}
Write $\alpha\equiv a^{-2}$ and $b\equiv M/a^2$, and retain
$F_R$ from \eqref{eq:critical-helical-circle}.
By Proposition~\ref{prop:exact-radial-envelope}, the unique
minimizer is an odd analytic function.  Substituting in
$\partial_xF_R=0$ gives
\begin{equation}\label{eq:critical-minimizer-expansion}
 x_{\min}(R)=-\frac b2R+
       \frac{b(3b^2-\alpha)}4R^3+O(R^5),
 \qquad
 \min F_R=-\frac{b^2}2R^2+
       \frac{b^2(2b^2-\alpha)}4R^4+O(R^6).
\end{equation}
In particular,
\begin{equation}\label{eq:exact-envelope-expansion}
 E(R)=\frac1{C_a}\left[R^2+\frac{b^2}{2}R^4
              -\frac{b^2(2b^2-\alpha)}4R^6+O(R^8)\right].
\end{equation}
If $a^2<2M^2$, the coefficient of $R^4$ is positive. Thus
$g+h_{\kappa_*}$ makes $U$ spacelike throughout a sufficiently small
punctured collar. The timelike one-form and positive conformal factor
in \eqref{eq:critical-normalized-conformal-factor} give the assertion.

If $a^2>2M^2$, evaluation at $x=-bR/2$ gives
\begin{equation}\label{eq:critical-linear-target-exact}
 F_R(-bR/2)+\frac{b^2}2R^2
 =\frac{b^2(2b^2-\alpha)}4R^4+
       \frac{b^4(\alpha-b^2)}{16}R^6<0
\end{equation}
for all sufficiently small $R>0$.
At $a^2=2M^2$, put $z\equiv bR$ and
$x\equiv-z/2+z^3/4$. Direct substitution yields
\begin{equation}\label{eq:critical-degenerate-target-exact}
 F_R(x)+\frac{z^2}{2}
 =-\frac{z^6}{16}+\frac{5z^8}{32}-\frac{z^{10}}{32}
       +\frac{z^{12}}{64}-\frac{3z^{14}}{256}
       +\frac{z^{16}}{512}<0
\end{equation}
for all sufficiently small $z>0$.
In both cases $m(R)+\kappa_*R^2<0$, hence
$(R^2+\kappa_*R^4)/C_a<E(R)$ for all sufficiently small $R$.
Proposition~\ref{prop:exact-radial-envelope}\textup{(i)} applies.
\end{proof}

The distinction at $a^2=2M^2$ belongs to the sharp length
criterion in the specified coordinate norm.  It yields no conclusion
about a change in geometric $C^0$-extendibility at that parameter value.

\subsection{An averaged criterion for larger perturbations}
\label{subsec:conformal-ring}

We next give an averaged axial condition for conformal factors
of the perturbed metrics in
Theorem~\ref{thm:stationary-ring-perturbation}.  For exact Kerr,
this condition also gives quantitative estimates, although
Theorem~\ref{thm:conformal-class-ring} applies to every positive
factor.

For a stationary and axisymmetric metric $g_0$ on the ring collar, put
\begin{equation}\label{eq:conformal-axial-weight}
 b_0(r,y)\equiv\sqrt{\max\{-g_0(\Phi,\Phi)(r,y),0\}},\qquad
 m_{I,A}(r,y)\equiv
 \min_{\tau\in I,\,\varphi\in\mathbb S^1}
 \int_0^A\Omega(\tau,r,y,\varphi+u)\,du.
\end{equation}
The minimum exists when $I$ is a nonempty compact interval.

\begin{theorem}[Conformal persistence of the ring obstruction]
\label{thm:conformal-ring-obstruction}
Let $g_0$ be the Kerr metric, with $M>0$ and $a\ne0$, or a metric
satisfying Theorem~\ref{thm:stationary-ring-perturbation}, on a full
punctured ring collar $\mathcal N_{\delta_0}$.  Let
$\Omega:\mathcal N_{\delta_0}\to(0,\infty)$ be continuous and put
$g_\Omega\equiv\Omega^2g_0$, with the same time orientation.
Suppose that, for every nonempty compact interval $I\subset\mathbb R$
and every $0<A<\pi$,
\begin{equation}\label{eq:conformal-axial-divergence}
 \limsup_{(r,y)\to(0,0)} b_0(r,y)m_{I,A}(r,y)=+\infty.
\end{equation}
Every future-directed $C^1$ timelike curve
$\gamma:[0,\Lambda)\to\mathcal N_{\delta_0}$ with
$\Sigma(\gamma(t))\to0$ satisfies
\begin{equation}\label{eq:conformal-fixed-endpoint-lengths}
 \forall s\in[0,\Lambda)\ \exists t_0\in(s,\Lambda)\quad
 \mathcal L^{g_\Omega}_\gamma(s,t)=+\infty
 \qquad(t_0<t<\Lambda).
\end{equation}
The varying portions of these homotopies can be chosen to lie in any
prescribed smaller ring collar.  Consequently, no such curve has an
endpoint in a $C^0$ extension of
$(\mathcal N_{\delta_0},g_\Omega)$.
The same conclusions hold on every connected covering with the
pulled-back metric and function $\Sigma$.
No stationarity or axial invariance of $\Omega$ is assumed.
\end{theorem}

\begin{proof}
The metrics $g_0$ and $g_\Omega$ have the same future cones.
Let $x$ have transverse radius $R$, and choose an axial-timelike
target $X_*\equiv(r_*,y_*)$ with $0<|X_*|<R$.
A radial segment followed by a circular arc joins the transverse
point of $x$ to $X_*$ with Euclidean length at most $(1+\pi)R$.
Use the future lifts from Lemma~\ref{lem:radial-attachment} or
Theorem~\ref{thm:stationary-ring-perturbation} along this path and
its reverse, inserting an axial arc of width $A_0$ at $X_*$.
Opening first the path and then the axial arc gives a causal
homotopy as in the preceding construction.  Its orbit displacement
is bounded by $C(R+A_0)$ independently of $X_*$, and
$0<\Sigma\le R^2$ along it.  We construct a separate homotopy
for each target; no continuous choice of targets or axial
orientation is required.

Fix $s<\Lambda$ and a smaller collar.  Choose the earlier compact
timelike segment for endpoint correction, and then $A_0$ and the
late radius $R$ so that $C(R+A_0)$ is admissible in
\eqref{eq:fixed-endpoint-radial-concatenation}.
For every sufficiently late fixed $t$, this produces curves
$\sigma_{X_*}$ causally homotopic to $\gamma|_{[s,t]}$.
The uniform bound on orbit displacements places the stationary
coordinates of their axial arcs $J_{X_*}$ in a compact interval
$I$ independent of $X_*$.  This requires no bound on the
stationary coordinate along the entire terminal tail.
Since $g_0$ is stationary and axisymmetric,
\begin{equation}\label{eq:conformal-detour-length-bound}
 L_{g_\Omega}(\sigma_{X_*})
 \ge L_{g_\Omega}(J_{X_*})
 \ge b_0(X_*)m_{I,A_0}(X_*).
\end{equation}
Reversing the axial orientation changes only the initial angle in
the minimum. The right-hand side is unbounded by
\eqref{eq:conformal-axial-divergence}.  The endpoint assertion
follows from Theorem~\ref{thm:causal-length-criterion}, and
Lemma~\ref{lem:covering-causal-length} gives the assertions on
connected covers.
\end{proof}

On the negative equator, $b_0(-\rho,0)\ge\sqrt{Ma^2/\rho}$ for
sufficiently small $\rho$.  The following simpler condition thus
implies \eqref{eq:conformal-axial-divergence}:
\begin{equation}\label{eq:conformal-equatorial-test}
 \limsup_{\rho\downarrow0}
 \frac{m_{I,A}(-\rho,0)}{\sqrt\rho}=+\infty
 \qquad\text{for every }I\text{ and }0<A<\pi.
\end{equation}
In particular, it suffices that, for every compact interval $I$,
there are $c_I>0$ and $q_I<1/4$ with
$\Omega\ge c_I\Sigma^{q_I}$ near the ring on the negative equator
over $I$.  The proof requires neither derivative estimates nor an upper
bound for $\Omega$, and the lower bound need only hold over each
compact interval in the stationary coordinate.

For example,
\begin{equation}\label{eq:conformal-angular-example}
 \Omega(\tau,r,y,\varphi)\equiv1-\cos\varphi+\sqrt\Sigma
\end{equation}
has angular minimum $\rho$ on the negative equator, but
$m_{I,A}(-\rho,0)=A-2\sin(A/2)+A\rho$, so
\eqref{eq:conformal-equatorial-test} holds.
The averaged condition can also be tested along other transverse
directions.
For $y=b\rho$ and $s_y\equiv1-y^2/a^2$,
\begin{equation}\label{eq:conformal-ray-axial-norm}
 -g_{M,a}(\Phi,\Phi)(-\rho,b\rho)
 =\frac{2Ma^2s_y^2}{(1+b^2)\rho}-s_y(a^2+\rho^2),
\end{equation}
and the same leading term holds under the perturbation bound above.
Thus $\sqrt\rho\,b_0(-\rho,b\rho)\to|a|\sqrt{2M/(1+b^2)}$.
For $q\ge1/4$,
\begin{equation}\label{eq:conformal-directional-example}
 \Omega_q\equiv\frac{y^2}{\Sigma}+\Sigma^q
\end{equation}
fails the equatorial test, whereas on $y=\rho$,
$b_0m_{I,A}\sim A|a|\sqrt M/(2\sqrt\rho)\to+\infty$.
Replacing $\Sigma^q$ by $e^{-1/\Sigma}$ has the same effect.

\begin{proposition}[Exact axial maximum]
\label{prop:exact-axial-maximum}
For exact Kerr, let $0<R<\min\{M/3,|a|\}$ and
$\mathcal C_R\equiv\{(r,y):r^2+y^2=R^2\}$.
With $b_0\equiv\sqrt{\max\{-g_{M,a}(\Phi,\Phi),0\}}$,
\begin{equation}\label{eq:exact-axial-circle-maximum}
 \max_{\mathcal C_R}b_0
 =\sqrt{\frac{2Ma^2}{R}-a^2-R^2},
 \qquad
 \sqrt R\max_{\mathcal C_R}b_0\longrightarrow |a|\sqrt{2M}.
\end{equation}
The maximizer is unique and equals $(-R,0)$.
If $\Omega=w(\sqrt\Sigma)>0$, the largest $g_\Omega$-length
of a timelike axial arc of angular width $A>0$ on $\mathcal C_R$
is $Aw(R)$ times the maximum in
\eqref{eq:exact-axial-circle-maximum}.
\end{proposition}

\begin{proof}
Put $\nu\equiv R^2/a^2$, $C\equiv2Ma^2/R$, and $k\equiv-r/R$.
On $\mathcal C_R$, one has $-1\le k\le1$ and
$z\equiv1-y^2/a^2=1-\nu+\nu k^2>0$.
The exact axial coefficient gives
\begin{equation}
 F(k)\equiv-g_{M,a}(\Phi,\Phi)
 =Ckz^2-z(a^2+R^2k^2),\qquad F(1)=C-a^2-R^2>0.
\end{equation}
Indeed, $0<\nu<1$ and $C>6a^2>6R^2$.
For $k\le0$, one has $F(k)<0$.
For $0\le k\le1$, direct factorization yields
\begin{align}
 F(1)-F(k)
 &=(1-k)\bigl\{C[1+k\nu(1+k)(1+z)]
            -R^2(1+k)(2+\nu k^2)\bigr\}\nonumber\\
 &\ge(1-k)(C-6R^2).
 \label{eq:axial-circle-strict-comparison}
\end{align}
The last factor is positive.  Equality therefore forces $k=1$,
and hence $(r,y)=(-R,0)$.  Taking square roots gives the unique
maximum; multiplication by $\sqrt R$ gives the stated limit.
An axial arc has constant $R$ and length $Aw(R)b_0(r,y)$.
\end{proof}

For $\Omega\equiv\Sigma^q$, the exact largest fixed-width axial length is
\begin{equation}\label{eq:conformal-radial-sharp-test}
 A R^{2q-1/2}\sqrt{2Ma^2-a^2R-R^3}.
\end{equation}
It diverges for $q<1/4$, tends to $A|a|\sqrt{2M}$ for $q=1/4$,
and tends to zero for $q>1/4$.
These alternatives concern a single fixed-width axial arc.
The paired orbit construction in
Theorem~\ref{thm:conformal-class-ring} excludes a continuous
ring endpoint for every $q\in\mathbb R$.

The power bound also applies to conformal factors on covers
which do not descend to the base.  For a connected cover
$\pi:\widehat{\mathcal N}\to\mathcal N_{\delta_0}$, suppose
$\widehat\Omega>0$ is continuous and
$\widehat\Omega\ge c(\Sigma\circ\pi)^q$ on the lifted negative
equator near the ring, where $c>0$ and $q<1/4$.
The lifted axial arcs then have
$\widehat\Omega^2\pi^*g_0$-length at least
$cA_0\sqrt{Ma^2}\rho^{2q-1/2}$, proving the same ring obstruction.

The conformal Kerr factor
$\Omega^2\equiv(1+\ell^2/\Sigma)^4$, $\ell>0$, from
\cite{BambiModestoRachwal2017} satisfies $\Omega\ge1$.
The exclusion of a ring endpoint is consistent with infinite
proper time along its geodesics approaching the ring. The conclusion
concerns the ring of this conformal metric.

\section{Terminal geodesics, curvature, and Jacobi fields}\label{sec:global-proof}

Throughout this section, let $M>0$ and $a\ne0$, and consider any
realization in Theorem~\ref{thm:main-c0}, including its connected
covers.  Ring localization permits a common analysis of terminal
geodesics in a single collar for all three parameter regimes.  We
first determine their constants of motion and asymptotic behavior,
and then study curvature and polar Jacobi fields along them.  The
results of this section do not enter the proof of
$C^0$-inextendibility.

The analysis of terminal timelike geodesics has further geometric
consequences. The classical equatorial restriction
\cite{Carter,ChruscielMaliborskiYunes} confines incomplete unit timelike
data to a five-dimensional submanifold of the seven-dimensional phase
space; an open dense set of full measure generates complete geodesics
(Corollary~\ref{cor:generic-timelike-completeness}). Complete geodesics
with fixed Killing momenta can nevertheless encounter curvature of
order $|p_\theta|^{-6}$. Their closest approach has radial scale
$|p_\theta|^{4/3}$ and precedes the radial turn by Mino time of order
$|p_\theta|^{2/3}$ (Proposition~\ref{prop:complete-near-miss} and
Corollary~\ref{cor:closest-approach}). Higher terminal expansions and
the exact polar Jacobi reduction in Mino's parameter \cite{Mino2003}
give a conserved quadratic quantity, a transfer matrix, and the zero
set. The leading radial exponent $2/5$ and tidal coefficient $6/25$
are due to \cite{ChruscielMaliborskiYunes}.

\subsection{Equatorial rigidity and terminal expansions}

\begin{theorem}[Terminal constants, sheet, and equatorial rigidity]\label{thm:terminal-geodesics}
Let $\gamma$ be an inextendible unit-speed timelike geodesic with
a finite proper-time endpoint in a spacetime of
Theorem~\ref{thm:main-c0}.  Put
\begin{equation}
 D\equiv L-aE.
\end{equation}
Then
\begin{equation}\label{eq:terminal-constants}
 \mathcal Q=D^2>0.
\end{equation}
Along the entire geodesic one has $r>0$ and $\mu=0$.
In a terminal block, $\theta=\pi/2$ and the geodesic approaches the
deleted ring from the positive sheet.
\end{theorem}

\begin{proof}
The inequalities $R(r(s))\ge0$ and $\Theta(\theta(s))\ge0$, together with the ring localization in Theorem~\ref{thm:main-c0}, give
\begin{equation}
 0\le\Theta(\pi/2)=\mathcal Q-D^2,
 \qquad
 0\le R(0)=a^2D^2-a^2\mathcal Q.
\end{equation}
The two inequalities imply $\mathcal Q=D^2$.  If $D=0$, then
$\mathcal Q=0$ and $L=aE$.  Since the defining identity
\eqref{eq:Q-definition} is a sum of nonnegative terms, it follows that
\begin{equation}
 p_\theta=0,
 \qquad
 \cos\theta=0,
 \qquad
 L-aE\sin^2\theta=0.
\end{equation}
The geodesic is therefore equatorial, and the removal of the ring implies $r(s)\ne0$.  The radial polynomial becomes
\begin{equation}
 R(r)=r^2\bigl((E^2-1)r^2+2Mr-a^2\bigr)<0
\end{equation}
for all sufficiently small $r\ne0$.  Since $r(s)\to0$, this contradicts
$R(r(s))\ge0$.  Therefore $D^2=\mathcal Q>0$.

Now
\begin{equation}\label{eq:R-near-ring}
 R(0)=0,
 \qquad
 R'(0)=2M\mathcal Q=2MD^2>0.
\end{equation}
Choose $\delta>0$ so that $R<0$ on $(-\delta,0)$ and $R>0$ on
$(0,\delta)$.  A tail contained in $|r|<\delta$ has $r\ge0$.
At a regular contact with $r=0$, the phase system satisfies
$(dr/d\tau)^2=R(r)$ and
\begin{equation}
 \frac{d^2r}{d\tau^2}=\frac12R'(r)
\end{equation}
including at turning points; see \eqref{eq:radial-phase-system}.
Thus $r=0$ would be a strict minimum with acceleration $MD^2$.
The geodesic could not turn inward again in the collar, since $R$
has no other zero there.  Hence $r>0$ eventually.

To prove that the geodesic is equatorial, we use
\eqref{eq:R-near-ring}:
\begin{equation}
 R(r)=2MD^2r+O(r^2).
\end{equation}
On the last inward branch,
\begin{equation}
 \tau_*-\tau=\int_0^{r(\tau)}\frac{d\rho}{\sqrt{R(\rho)}}<\infty,
\end{equation}
so the rescaled-parameter endpoint $\tau_*$ is finite.  Set $x\equiv\cos\theta$.  Substituting $\mathcal Q=D^2$ into \eqref{eq:Theta} and multiplying by $1-x^2$ yields the exact identity
\begin{equation}\label{eq:x-equation}
 \left(\frac{dx}{d\tau}\right)^2
 =(1-x^2)\Theta
 =x^2\left[a^2(1-E^2)x^2-\bigl(L^2+a^2(1-E^2)\bigr)\right].
\end{equation}
Put
\begin{equation}
 B\equiv L^2+a^2(1-E^2).
\end{equation}
Since $|x|\le1$, equation~\eqref{eq:x-equation} implies
\begin{equation}
 \left|\frac{dx}{d\tau}\right|\le C|x|,
 \qquad C^2\equiv |a^2(1-E^2)|+|B|.
\end{equation}
The continuous extension $x(\tau_*)=0$ therefore satisfies
\begin{equation}
 |x(\tau)|\le C\int_\tau^{\tau_*}|x(\sigma)|\,d\sigma.
\end{equation}
The backward Gronwall inequality gives $x=0$ on the terminal
interval.  Thus $\theta=\pi/2$ there for every value of $B$.

In every regular Kerr or covering chart, the reflection
$\mu\mapsto-\mu$ is a local isometry.  By geodesic uniqueness,
the set of parameters with $\mu=d\mu/ds=0$ is open; it is also
closed and contains the tail.  It is therefore the whole interval.
Then $\Sigma=r^2>0$, so continuity preserves the positive sign of $r$.
\end{proof}

Chru\'sciel, Maliborski, and Yunes obtained the leading toroidal-radius
law and the polar tidal coefficient $6/25$
\cite[Section~V.B, equations~(V.22)-(V.23) and (V.28)]{ChruscielMaliborskiYunes}.
Their relation $\widehat r=r^2/(2|a|)+O(r^4)$ gives the coefficient
and exponent in \eqref{eq:sharp-leading-laws}.  The following theorem
establishes this law for every terminal branch in the realizations
and covers under consideration.  It also gives two further terms in
the radial expansion and the corresponding orbit-coordinate expansions.

\begin{theorem}[Sharp terminal equatorial asymptotics]\label{thm:sharp-terminal-asymptotics}
Let $\gamma:[0,\Lambda)\to\KerrSpace$ be a future-inextendible unit-speed timelike
geodesic with $\Lambda<\infty$.  On the terminal equatorial interval
put
\begin{equation}
 u\equiv\Lambda-s,\qquad D\equiv L-aE,\qquad d\equiv|D|,\qquad
 B\equiv L^2+a^2(1-E^2),
\end{equation}
and define
\begin{equation}\label{eq:terminal-radial-scale}
 \varrho(u)\equiv\left(\frac52\sqrt{2M}\,d\,u\right)^{2/5}.
\end{equation}
Then $d>0$ and
\begin{equation}\label{eq:sharp-r-asymptotics}
 r(s)=\varrho-\frac{B}{14Md^2}\varrho^2
 +\left(\frac{1}{9d^2}-\frac{B^2}{147M^2d^4}\right)\varrho^3
 +O(\varrho^4).
\end{equation}
There are finite constants $t_{\rm end}$ and $\phi_{\rm end}$, with a
continuous lift of $\phi$ chosen on the terminal interval, such that,
on setting
\begin{align}
 c_t&\equiv-\frac{2MD}{a},
 &b_t&\equiv E-\frac{4M^2D}{a^3},
 &e_t&\equiv\frac{2ME}{a^2}+\frac{2MD}{a^3}-\frac{8M^3D}{a^5},
 \label{eq:terminal-t-constants}\\
 c_\phi&\equiv-\frac{2MD}{a^2},
 &b_\phi&\equiv\frac{L}{a^2}-\frac{4M^2D}{a^4},
 &e_\phi&\equiv\frac{2M(L+D)}{a^4}-\frac{8M^3D}{a^6},
 \label{eq:terminal-phi-constants}
\end{align}
one has
\begin{align}
 t_{\rm end}-t(s)
 &=
 \frac{1}{\sqrt{2M}\,d}
 \Bigl[
  \frac{2c_t}{3}\varrho^{3/2}
  +\left(\frac{2b_t}{5}
  +\frac{Bc_t}{35Md^2}\right)\varrho^{5/2}\nonumber\\
 &\quad+\left\{\frac{2e_t}{7}
  +c_t\left(\frac{B^2}{294M^2d^4}-\frac{2}{63d^2}\right)
  \right\}\varrho^{7/2}
  +O(\varrho^{9/2})
 \Bigr],\label{eq:sharp-t-asymptotics}\\
 \phi_{\rm end}-\phi(s)
 &=
 \frac{1}{\sqrt{2M}\,d}
 \Bigl[
  \frac{2c_\phi}{3}\varrho^{3/2}
  +\left(\frac{2b_\phi}{5}
  +\frac{Bc_\phi}{35Md^2}\right)\varrho^{5/2}\nonumber\\
 &\quad+\left\{\frac{2e_\phi}{7}
  +c_\phi\left(\frac{B^2}{294M^2d^4}-\frac{2}{63d^2}\right)
  \right\}\varrho^{7/2}
  +O(\varrho^{9/2})
 \Bigr].\label{eq:sharp-phi-asymptotics}
\end{align}
The corotating combination satisfies the sharper cancellation
\begin{equation}\label{eq:terminal-corotating-asymptotics}
 (t_{\rm end}-t(s))-a(\phi_{\rm end}-\phi(s))
 =-\frac{D}{a}\,u-\frac{10MD}{7a^3}\,u\varrho
 +O(u\varrho^2).
\end{equation}
In particular,
\begin{equation}\label{eq:sharp-leading-laws}
 r(s)=\left(\frac52\sqrt{2M}\,d\right)^{2/5}
 u^{2/5}\bigl(1+O(u^{2/5})\bigr).
\end{equation}
\end{theorem}

\begin{proof}
By Theorem~\ref{thm:terminal-geodesics},
$\theta=\pi/2$ and $\mathcal Q=D^2=d^2>0$ on a terminal interval.
On this interval, $\Sigma=r^2$ and $P=Er^2-aD$, and the radial
polynomial becomes
\begin{align}
 R(r)
 &=(Er^2-aD)^2-(r^2-2Mr+a^2)(r^2+d^2)\nonumber\\
 &=2Md^2r-Br^2+2Mr^3+(E^2-1)r^4.
 \label{eq:terminal-R-exact}
\end{align}
It is positive for small $r>0$, and $r\to0^+$ forces the inward branch
\begin{equation}\label{eq:terminal-rdot}
 \dot r=-\frac{\sqrt{R(r)}}{r^2}.
\end{equation}
Separation of variables therefore gives
\begin{equation}\label{eq:terminal-time-integral}
 u=\int_0^{r(s)}\frac{y^2}{\sqrt{R(y)}}\,dy.
\end{equation}
Factoring the leading term in \eqref{eq:terminal-R-exact} gives
\begin{equation}
 R(y)=2Md^2y
 \left(1-\frac{B}{2Md^2}y+\frac{1}{d^2}y^2+O(y^3)\right),
\end{equation}
and hence
\begin{equation}\label{eq:terminal-integrand-expansion}
 \frac{y^2}{\sqrt{R(y)}}
 =\frac{y^{3/2}}{\sqrt{2M}\,d}
 \left[1+\frac{B}{4Md^2}y
 +\left(\frac{3B^2}{32M^2d^4}-\frac1{2d^2}\right)y^2
 +O(y^3)\right].
\end{equation}
After integration,
\begin{equation}\label{eq:u-r-expansion}
 u=\frac{2}{5\sqrt{2M}\,d}r^{5/2}
 \left[1+\frac{5B}{28Md^2}r
 +\left(\frac{5B^2}{96M^2d^4}-\frac{5}{18d^2}\right)r^2
 +O(r^3)\right].
\end{equation}
In the notation of Lemma~\ref{lem:simple-zero-terminal-law}, the
coefficients are
\begin{align}
 H(r)&=r^2, &q&=2, &p&=\frac52,\nonumber\\
 \kappa_0&=2Md^2,
 &\kappa_1&=-\frac{B}{2Md^2},
 &\kappa_2&=\frac1{d^2},\nonumber\\
 \eta_1&=\frac{B}{4Md^2},
 &\eta_2&=\frac{3B^2}{32M^2d^4}-\frac1{2d^2},\nonumber\\
 k_1&=\frac{B}{14Md^2},
 &k_2&=\frac1{9d^2}-\frac{B^2}{147M^2d^4}.
 \label{eq:kerr-simple-zero-coefficients}
\end{align}
Taking the $2/5$ power gives
\begin{equation}
 \varrho=r+k_1r^2+\left(2k_1^2-k_2\right)r^3+O(r^4),
\end{equation}
and inversion yields \eqref{eq:sharp-r-asymptotics}.

The orbit equations simplify exactly to
\begin{align}
 \dot t
 &=\frac{1}{r^2}
 \left[\frac{(r^2+a^2)(Er^2-aD)}{\Delta}+aD\right]
 =E\frac{r^2+a^2}{\Delta}
 -\frac{2MaD}{\Delta r},
 \label{eq:terminal-tdot-exact}\\
 \dot\phi
 &=\frac{1}{r^2}
 \left[\frac{a(Er^2-aD)}{\Delta}+D\right]
 =\frac{L}{\Delta}-\frac{2MD}{\Delta r}.
 \label{eq:terminal-phidot-exact}
\end{align}
The expansion
\begin{equation}
 \frac1\Delta=\frac1{a^2}
 \left[1+\frac{2M}{a^2}r
 +\left(\frac{4M^2}{a^4}-\frac1{a^2}\right)r^2
 +O(r^3)\right],
\end{equation}
therefore gives
\begin{equation}\label{eq:terminal-orbit-expansions}
 \dot t=\frac{c_t}{r}+b_t+e_tr+O(r^2),\qquad
 \dot\phi=\frac{c_\phi}{r}+b_\phi+e_\phi r+O(r^2),
\end{equation}
with the stated coefficients.  Both increments are integrable.
Applying \eqref{eq:simple-z-expansion} to $t,\phi$ with $\ell=1$,
$n=3/2$, and the coefficients in
\eqref{eq:kerr-simple-zero-coefficients} yields
\eqref{eq:sharp-t-asymptotics}-\eqref{eq:sharp-phi-asymptotics}.
The coefficient of the third term simplifies as follows:
\begin{align}
 \frac{b_2}{n+2}-b_1k_1
 +z_0\left(k_2+\frac{n-1}{2}k_1^2\right)
 =\frac{2z_2}{7}
 +z_0\left(\frac{B^2}{294M^2d^4}-\frac{2}{63d^2}\right),
 \label{eq:kerr-orbit-third-coefficient}
\end{align}
where $(z_0,z_1,z_2)$ is $(c_t,b_t,e_t)$ or
$(c_\phi,b_\phi,e_\phi)$.  All terms containing $z_1$ cancel.

The corotating combination simplifies exactly:
\begin{equation}
 \dot t-a\dot\phi
 =\frac{P}{\Delta}
 =\frac{Er^2-aD}{r^2-2Mr+a^2}
 =-\frac Da-\frac{2MD}{a^3}r+O(r^2).
\end{equation}
Writing $\varrho(v)\equiv((5/2)\sqrt{2M}\,d\,v)^{2/5}$, the radial
expansion gives
\begin{align}
 \int_s^\Lambda r(\sigma)\,d\sigma
 &=\int_0^u\left(\varrho(v)+O(\varrho(v)^2)\right)\,dv\nonumber\\
 &=\frac57u\varrho(u)+O(u\varrho(u)^2).
\end{align}
Integrating this identity gives
\eqref{eq:terminal-corotating-asymptotics}.
\end{proof}

\begin{proposition}[A family of incomplete timelike geodesics]\label{prop:explicit-incomplete}
Fix a full Kerr block containing both regular disc branches, and constants $E,L\in\mathbb R$
with
\begin{equation}
 D\equiv L-aE\ne0,
 \qquad \mathcal Q=D^2,
\end{equation}
and arbitrary $t_0,\phi_0$.  For every sufficiently small $r_0>0$
there is a unit-speed timelike geodesic in that block with initial
position
\begin{equation}
 (t,r,\theta,\phi)=(t_0,r_0,\pi/2,\phi_0),
\end{equation}
and initial momenta satisfying
\begin{equation}
 p_t=-E,\qquad p_\phi=L,\qquad p_\theta=0,
 \qquad \dot r(0)<0.
\end{equation}
It approaches the deleted positive ring in finite proper time.  Hence
every spacetime in Theorem~\ref{thm:main-c0} is timelike geodesically incomplete.
\end{proposition}

\begin{proof}
In the subextremal case, the required block is supplied by
conditions (C2), (C4); such a block is explicit in the other
constructions.  Reflection and uniqueness for the Hamiltonian flow
preserve $\theta=\pi/2$, $p_\theta=0$.  With $\mathcal Q=D^2$, the
separated radial equation gives
\begin{equation}\label{eq:explicit-R}
 R(r)=2MD^2r-
 \bigl(L^2+a^2(1-E^2)\bigr)r^2
 +2Mr^3+(E^2-1)r^4
 =2MD^2r+O(r^2).
\end{equation}
Choose $r_0$ with $R,\Delta>0$ on $(0,r_0]$.
The prescribed momenta and $\dot r=-\sqrt R/r^2$ satisfy $H=-1/2$
and determine a unit timelike geodesic.  Since $R>0$, it has no
radial turn before the ring.  Its remaining proper time is
\begin{equation}
 \int_0^{r_0}\frac{r^2}{\sqrt{R(r)}}\,dr<\infty,
\end{equation}
because the integrand is
$(\sqrt{2M}|D|)^{-1}r^{3/2}(1+O(r))$.
Equations~\eqref{eq:terminal-tdot-exact}-\eqref{eq:terminal-phidot-exact}
give $dt/dr,d\phi/dr=O(r^{1/2})$, so both orbit coordinates have
finite limits.  The geodesic reaches $\Sigma=0$ and cannot continue
in the spacetime.  The same holds for every lift to a connected cover.
\end{proof}

\begin{remark}
By Proposition~\ref{prop:explicit-incomplete}, an inextendibility
criterion requiring timelike geodesic completeness does not apply
to these spacetimes.
\end{remark}

The classical equatorial restriction
\cite{Carter,ChruscielMaliborskiYunes}, together with the smooth
geodesic flow on the unit timelike bundle \cite{ONeill1983}, gives the
following global description of the exceptional initial data.

\begin{corollary}[Generic timelike completeness]
\label{cor:generic-timelike-completeness}
Let $(N,g)$ be any spacetime in Theorem~\ref{thm:main-c0}, including
any of its connected covers, and let $\mathcal U^+N$ denote its
seven-dimensional bundle of future unit timelike vectors.  The set
\begin{equation}\label{eq:exceptional-equatorial-data}
 \mathcal E\equiv
 \{(p,v)\in\mathcal U^+N:r(p)>0,\ \mu(p)=0,\ d\mu_p(v)=0\}
\end{equation}
is a closed invariant embedded submanifold of codimension two.
On the open dense set $\mathcal U^+N\setminus\mathcal E$, the
geodesic flow is defined for every proper time in $\mathbb R$.
The sets $\mathcal I_+$ and $\mathcal I_-$ of initial data giving,
respectively, a future-incomplete or past-incomplete geodesic are
nonempty relatively open subsets of $\mathcal E$.  Thus
$\mathcal I\equiv\mathcal I_+\cup\mathcal I_-$ is an embedded
submanifold of codimension two.  It has Hausdorff dimension five
for every smooth auxiliary Riemannian metric on $\mathcal U^+N$,
and measure zero for every smooth positive density.
\end{corollary}

\begin{proof}
On $\{\mu=0\}$, including horizons by continuity,
$g(\nabla\mu,\nabla\mu)=1/r^2>0$.  If $d\mu(v)=0$, then
$\nabla\mu\in v^\perp$ is a vertical tangent vector on the unit
timelike fiber, and the vertical derivative of $d\mu(v)$ in that
direction is $1/r^2$.  Thus the two defining equations are
independent.  Their zero set excludes $r=0$, because $\Sigma>0$;
its $r>0$ part is therefore both open and closed in it.  Hence
$\mathcal E$ is a closed codimension-two submanifold.

Reflection and uniqueness preserve the equatorial conditions, and
an equatorial geodesic cannot cross $r=0$.  Thus $\mathcal E$ and
its complement are invariant.  By Theorem~\ref{thm:terminal-geodesics},
the flow on the complement is complete.  This complement is open,
dense, and of full measure for every smooth density.

Choose an inward datum from Proposition~\ref{prop:explicit-incomplete},
made future-directed by simultaneously reversing $E,L$ if needed;
this leaves $R$ and the inward branch unchanged.  In a compact
neighborhood of these constants, $|D|\ge d>0$ and uniformly
\begin{equation}
 R(r)\ge Md^2r\qquad(0<r<r_*).
\end{equation}
The inward data in $\mathcal E$ are locally parametrized by
$(t,r,\phi,E,L)$.  The first integrals and
$\dot r=-\sqrt R/r^2$ determine the velocity, while the inverse map
is given by the position coordinates and Killing momenta.  Restricting
to open parameter intervals gives a nonempty relatively open future-incomplete family,
since no radial turn is possible and
\begin{equation}
 \int_0^{r_0}\frac{r^2}{\sqrt{R(r)}}\,dr
 \le\frac{2r_0^{5/2}}{5\sqrt M\,d}<\infty.
\end{equation}
The outward branch gives a nonempty relatively open past-incomplete
family of future unit timelike data.

Every future-incomplete datum enters such a family at a finite
time $T$.  The time-$T$ flow is a local diffeomorphism, so
$\mathcal I_+$ is relatively open in $\mathcal E$; reversing time
gives the same conclusion for $\mathcal I_-$.  Their nonempty union
is a five-dimensional submanifold and has the asserted Hausdorff
dimension and measure.  The argument applies equally to covers,
using local charts and the global functions $r,\mu$.
\end{proof}

Generic completeness is compatible with $C^0$-inextendibility.
An extension would produce an exceptional incomplete geodesic,
and the ring obstruction excludes its endpoint.  The dimension in
the corollary describes the exceptional initial data; it does not
measure the strength of the singularity.

\subsection{\texorpdfstring{$C^{1,1}$}{C1,1} metrics have locally bounded curvature}

\begin{lemma}[Essential curvature bound]\label{lem:c11-curvature}
Let $\widetilde h\in C^{1,1}_{\mathrm{loc}}$ be a nondegenerate
Lorentzian metric.  On compact subsets of every smooth coordinate
chart, the components of its Riemann and Weyl tensors belong to
$L^\infty$.  Every scalar formed by finite sums, tensor products, and
algebraic contractions of finitely many curvature or Weyl tensors with
$\widetilde h$, $\widetilde h^{-1}$, and a local volume form is locally
essentially bounded.  In particular, orientation-independent
polynomials obtained by squaring a pseudoscalar have this property.
\end{lemma}

\begin{proof}
On a compact coordinate set, $\widetilde h$, its inverse, and its
first derivatives are bounded.  The second weak derivatives are
bounded almost everywhere, so the Christoffel symbols are locally
Lipschitz and
\begin{equation}
 R^\alpha{}_{\beta\mu\nu}
 =\partial_\mu\Gamma^\alpha_{\nu\beta}
 -\partial_\nu\Gamma^\alpha_{\mu\beta}
 +\Gamma^\alpha_{\mu\lambda}\Gamma^\lambda_{\nu\beta}
 -\Gamma^\alpha_{\nu\lambda}\Gamma^\lambda_{\mu\beta}
\end{equation}
has $L^\infty$ components.  Taking traces gives the same conclusion for
Weyl curvature.  Nondegeneracy bounds the inverse metric and the
coefficients of a local volume form.  Finite tensor products and
contractions therefore preserve $L^\infty$.
\end{proof}

\begin{lemma}[Unbounded smooth invariant excludes a $C^{1,1}$ boundary]\label{lem:invariant-obstruction}
Let $(N,h)$ be smoothly and isometrically embedded as an open subset
of a $C^{1,1}_{\mathrm{loc}}$ Lorentzian manifold
$(\widetilde N,\widetilde h)$.  Suppose $q\in\partial N$ is approached
by points $p_j\in N$, and a smooth algebraic curvature scalar
$\mathcal I$ on $N$ satisfies
$\abs{\mathcal I(p_j)}\to\infty$.  If $\mathcal I$ is one of the
scalars covered by Lemma~\ref{lem:c11-curvature}, then no such
extension exists.
\end{lemma}

\begin{proof}
In a relatively compact extension chart $W\ni q$,
Lemma~\ref{lem:c11-curvature} gives an essential bound for $\mathcal I$.
On $W\cap N$, weak and classical curvature agree almost everywhere.  For arbitrary $A>0$,
a point with $|\mathcal I|>2A$ has an open neighborhood
$O\Subset W\cap N$ on which $|\mathcal I|>A$.  Since $O$ has
positive coordinate measure, $\esssup_W|\mathcal I|\ge A$, a contradiction.
\end{proof}

\begin{theorem}[Finite-proper-time curvature criterion]\label{thm:general-curvature-criterion}
Let $(N,h)$ be a connected smooth time-oriented spacetime.  Suppose
there are finitely many algebraic curvature scalars
$\mathcal I_1,\ldots,\mathcal I_m$ covered by
Lemma~\ref{lem:c11-curvature}, with the following property:
for either choice of time orientation, every future-inextendible
unit-speed timelike geodesic
\begin{equation}
 \gamma:[0,\Lambda)\longrightarrow N,\qquad \Lambda<\infty,
\end{equation}
satisfies
\begin{equation}\label{eq:general-curvature-divergence}
 \limsup_{s\uparrow\Lambda}\,
 \max_{1\le j\le m}|\mathcal I_j(\gamma(s))|=\infty.
\end{equation}
Then $(N,h)$ has no proper $C^{1,1}_{\mathrm{loc}}$ extension.
\end{theorem}

\begin{proof}
If such an extension existed, Theorem~\ref{thm:boundary-geodesic},
with the time orientation reversed if necessary, would give a geodesic
$\gamma$ as in the statement and a boundary limit $q$.  By
\eqref{eq:general-curvature-divergence}, one of the invariants is
unbounded along a subsequence.  This contradicts
Lemma~\ref{lem:invariant-obstruction} at $q$.
\end{proof}

\subsection{Curvature and polar Jacobi fields}

\begin{proposition}[The two Weyl invariants]\label{prop:weyl-invariants}
On Kerr, choose a local orientation, put $y\equiv a\cos\theta$, and define
\begin{equation}
 \mathcal W_1\equiv C_{\alpha\beta\gamma\delta}C^{\alpha\beta\gamma\delta},
 \qquad
 \mathcal W_2\equiv C_{\alpha\beta\gamma\delta}{}^\star C^{\alpha\beta\gamma\delta}.
\end{equation}
Since Kerr is Ricci-flat, $\mathcal W_1$ is also the Kretschmann
scalar $R_{\alpha\beta\gamma\delta}R^{\alpha\beta\gamma\delta}$.
Then
\begin{align}
 \mathcal W_1
 &=\frac{48M^2
 (r^6-15r^4y^2+15r^2y^4-y^6)}
 {(r^2+y^2)^6},\label{eq:W1-real}\\
 \mathcal W_2
 &=\varepsilon\,
 \frac{96M^2ry(3r^4-10r^2y^2+3y^4)}
 {(r^2+y^2)^6},\qquad \varepsilon\in\{-1,1\},
 \label{eq:W2-real}
\end{align}
where $\varepsilon$ changes sign when the orientation is reversed.
For the orientation corresponding to $\varepsilon=-1$,
\begin{equation}\label{eq:complex-weyl}
 \mathcal W_1+i\mathcal W_2
 =\frac{48M^2}{(r+ia\cos\theta)^6}.
\end{equation}
For the opposite orientation, take the complex conjugate.
Consequently
\begin{equation}\label{eq:weyl-magnitude}
 \sqrt{\mathcal W_1^2+\mathcal W_2^2}
 =\frac{48M^2}{\Sigma^3}.
\end{equation}
Although $\mathcal W_2$ changes sign between the two local
orientations, $\mathcal J\equiv\mathcal W_1^2+\mathcal W_2^2$ is a smooth
global scalar.  It diverges along every approach to $\Sigma=0$,
irrespective of the direction of approach.
\end{proposition}

\begin{proof}
In a principal Newman-Penrose tetrad, the only nonzero Weyl scalar is
$\Psi_2=-M(r-iy)^{-3}$.  With the local orientation for which
$\varepsilon=-1$ in \eqref{eq:W2-real}, the type-D contraction in our
Hodge convention is
\begin{equation}
 \mathcal W_1-i\mathcal W_2=48\Psi_2^2
 =\frac{48M^2}{(r-iy)^6}.
\end{equation}
Taking the complex conjugate gives \eqref{eq:complex-weyl}.
Reversing the orientation interchanges these two identities and hence
changes only the sign of $\mathcal W_2$.
A direct metric calculation of the same contractions appears in
\cite[Section~III.C, equations~(III.18)-(III.21)]{ChruscielMaliborskiYunes}.

The corresponding real expansion is
\begin{align}
 (r-iy)^6={}&
 r^6-15r^4y^2+15r^2y^4-y^6\nonumber\\
 &{}-2iry(3r^4-10r^2y^2+3y^4).
\end{align}
Division by $(r^2+y^2)^6$ proves
\eqref{eq:W1-real}-\eqref{eq:W2-real}.  If
\begin{align}
 A_6&\equiv r^6-15r^4y^2+15r^2y^4-y^6,\\
 B_6&\equiv2ry(3r^4-10r^2y^2+3y^4),
\end{align}
then the squared modulus of the preceding binomial identity gives
\begin{equation}
 A_6^2+B_6^2=(r^2+y^2)^6.
\end{equation}
It follows that
\begin{equation}
 \mathcal W_1^2+\mathcal W_2^2
 =\frac{(48M^2)^2}{(r^2+y^2)^6}.
\end{equation}
Since
\begin{equation}
 |r+ia\cos\theta|^2=r^2+a^2\cos^2\theta=\Sigma,
\end{equation}
this is \eqref{eq:weyl-magnitude}.
\end{proof}

The curvature need not be uniformly bounded over a family of
complete geodesics.  We construct such a family by perturbing an
incomplete geodesic while keeping both Killing momenta fixed.  This
complements \cite[Section~V]{ChruscielMaliborskiYunes} and requires
no assumption that the geodesics originate at infinity.

\begin{proposition}[Curvature along a family of complete geodesics]
\label{prop:complete-near-miss}
In any spacetime of Theorem~\ref{thm:main-c0}, choose a future unit
timelike inward equatorial geodesic sufficiently close to the positive
ring, with constants $E,L$ and $D\equiv L-aE\ne0$.  Put
\begin{align}
 B&\equiv L^2+a^2(1-E^2),\qquad
 R_0(r)\equiv 2MD^2r-Br^2+2Mr^3+(E^2-1)r^4,\nonumber\\
 \tau_*&\equiv\int_0^{r_0}\frac{dr}{\sqrt{R_0(r)}},\qquad
 \Lambda_0\equiv\int_0^{r_0}\frac{r^2\,dr}{\sqrt{R_0(r)}}.
 \label{eq:near-miss-background}
\end{align}
Here $r_0$ is its initial radius, and $\Lambda_0$ is its remaining
proper time.  Let $S_B$ solve
\begin{equation}
 S_B''+BS_B=0,\qquad S_B(0)=0,\qquad S_B'(0)=1.
\end{equation}
The initial radius can be chosen so small that
$b\equiv S_B(\tau_*)>0$.
At the same initial point, keep $E,L$ fixed, set $p_\theta=\varepsilon$,
and choose the inward unit timelike radial branch.  For every
sufficiently small $\varepsilon\ne0$, the resulting geodesic is
complete in both proper-time directions.  It has a radial turning
event $q_\varepsilon$ at proper time
\begin{equation}
 s_\varepsilon=\Lambda_0+O(\varepsilon^2),
\end{equation}
where
\begin{align}
 r(q_\varepsilon)
 &=\frac{a^2}{2MD^2}\varepsilon^2+O(\varepsilon^4),\nonumber\\
 \mu(q_\varepsilon)&=-b\varepsilon+O(\varepsilon^3),\nonumber\\
 \Sigma(q_\varepsilon)&=a^2b^2\varepsilon^2
   \bigl(1+O(\varepsilon^2)\bigr),\label{eq:near-miss-scales}\\
 \sqrt{\mathcal W_1^2+\mathcal W_2^2}(q_\varepsilon)
 &=\frac{48M^2}{a^6b^6}|\varepsilon|^{-6}
   \bigl(1+O(\varepsilon^2)\bigr).
 \label{eq:near-miss-curvature}
\end{align}
The point in these estimates is the radial turning event.  The
local minimum of $\Sigma$ during this encounter is described in the next corollary.
\end{proposition}

\begin{proof}
Choose $r_0$ in a positive ring collar where $\Delta>0$ and
$R_0(r)>0$ for $0<r\le r_0$.  At the initial equatorial point,
the Carter identity gives
\begin{equation}
 \mathcal Q_\varepsilon=D^2+\varepsilon^2,
 \qquad R_\varepsilon(r)=R_0(r)-\Delta(r)\varepsilon^2.
 \label{eq:near-miss-radial-polynomial}
\end{equation}
For small $\varepsilon$, $R_\varepsilon(r_0)>0$ gives smooth
future unit timelike data on the inward branch.  For $\varepsilon\ne0$,
the data lie outside $\mathcal E$.  By
Corollary~\ref{cor:generic-timelike-completeness}, the corresponding
geodesics are complete in both directions, including on covers.

In Mino's parameter $d\tau=ds/\Sigma$, with $\tau=0$ initially,
we write $v\equiv dr/d\tau$.  The radial phase system is
\begin{equation}
 r'=v,\qquad v'=\tfrac12R_\varepsilon'(r),\qquad
 (r(0),v(0))=(r_0,-\sqrt{R_\varepsilon(r_0)}).
 \label{eq:near-miss-radial-phase}
\end{equation}
It preserves $v^2=R_\varepsilon(r)$.  At $\varepsilon=0$,
$(r,v)(\tau_*)=(0,0)$ and $v'(\tau_*)=MD^2>0$.
The polynomial phase system admits a smooth continuation through
this point, although the corresponding spacetime geodesic has no
continuation there.

The radial solution depends smoothly on $h\equiv\varepsilon^2$.
The implicit function theorem applied to $v(\tau,h)=0$ gives a
unique nearby turning time
\begin{equation}
 \tau_\varepsilon=\tau_*+O(\varepsilon^2).
\end{equation}
At this event the preserved first integral gives
$R_0(r)-h\Delta(r)=0$.  Since $R_0'(0)=2MD^2>0$, a second
application of the implicit function theorem gives
\begin{equation}
 r(\tau_\varepsilon)
 =\frac{\Delta(0)}{R_0'(0)}h+O(h^2)
 =\frac{a^2}{2MD^2}\varepsilon^2+O(\varepsilon^4)>0.
\end{equation}
The unperturbed velocity is negative before $\tau_*$, and the
acceleration is uniformly positive near this time.  By smooth
dependence, the perturbed event is therefore the first radial turn,
with $r>0$ up to it.
Since $\Delta>0$ and $|\mu|<1$, the smooth orbit equations
reconstruct the geodesic throughout the interval.

For $\mu\equiv\cos\theta$, the angular phase equation is
\begin{equation}
 \mu''=-(B+\varepsilon^2)\mu
       +2a^2(1-E^2)\mu^3,
 \qquad \mu(0)=0,\qquad\mu'(0)=-\varepsilon.
 \label{eq:near-miss-angular-phase}
\end{equation}
The angular phase system is smooth, including at turning points.
Its solution is odd in $\varepsilon$, with first variation $-S_B$.
Thus, uniformly on a fixed interval containing $[0,\tau_*]$, 
\begin{equation}
 \mu_\varepsilon(\tau)=-\varepsilon S_B(\tau)
                         +O(\varepsilon^3).
\end{equation}
As $r_0\downarrow0$, $\tau_*\downarrow0$, while
$S_B(\tau)=\tau+O(\tau^3)$.  Thus $r_0$ can be chosen so that
$b>0$.  Evaluation at $\tau_\varepsilon$ proves the first three
estimates in \eqref{eq:near-miss-scales}.

Finally,
\begin{equation}
 s_\varepsilon
 =\int_0^{\tau_\varepsilon}
   \bigl(r_\varepsilon(\tau)^2
         +a^2\mu_\varepsilon(\tau)^2\bigr)\,d\tau.
\end{equation}
On this compact interval,
$r_\varepsilon-\bar r=O(\varepsilon^2)$,
$\mu_\varepsilon=O(\varepsilon)$, and
$\tau_\varepsilon-\tau_*=O(\varepsilon^2)$; hence
$s_\varepsilon=\Lambda_0+O(\varepsilon^2)$.
The Weyl identity \eqref{eq:weyl-magnitude} proves
\eqref{eq:near-miss-curvature}.  The construction lifts to covers.
\end{proof}

\begin{corollary}[Closest approach of the complete geodesics]
\label{cor:closest-approach}
For the family in Proposition~\ref{prop:complete-near-miss}, choose
$r_0$ so that $\beta\equiv S_B'(\tau_*)>0$, and put
\begin{equation}
 \eta\equiv|\varepsilon|^{2/3},\qquad
 \kappa\equiv MD^2,\qquad
 \lambda\equiv\left(\frac{2a^2b\beta}{\kappa^2}\right)^{1/3}.
\end{equation}
On a fixed sufficiently small Mino-time interval about $\tau_*$,
$\Sigma$ has a unique minimum, at $p_\varepsilon$, and
\begin{align}
 \tau(p_\varepsilon)&=\tau_*-\lambda\eta+O(\eta^2),
 &r(p_\varepsilon)&=\frac{\kappa\lambda^2}{2}\eta^2+O(\eta^3),\nonumber\\
 \Sigma(p_\varepsilon)
 &=a^2b^2\varepsilon^2
   \left[1-\frac{3\beta\lambda}{2b}\eta+O(\eta^2)\right].
 \label{eq:closest-approach-scales}
\end{align}
Consequently the maximum of $\sqrt{\mathcal W_1^2+\mathcal W_2^2}$
during this encounter equals
\begin{equation}\label{eq:closest-approach-curvature}
 \frac{48M^2}{a^6b^6}|\varepsilon|^{-6}
 \left[1+\frac{9\beta\lambda}{2b}\eta+O(\eta^2)\right].
\end{equation}
The proper time from $p_\varepsilon$ to the radial turning event
$q_\varepsilon$ is $a^2b^2\lambda\eta^4+O(\eta^5)$.
If $L=0$ and $E^2>1$, these are the unique global minimum and
maximum along the complete geodesic, for all sufficiently small
$\varepsilon\ne0$.
\end{corollary}

\begin{proof}
Write $\xi\equiv\tau-\tau_*$ and $h\equiv\varepsilon^2$.
The phase equations in the preceding proof give smooth functions
$r(\xi,h)$ and $u(\xi,h)\equiv\mu_\varepsilon(\tau)/\varepsilon$,
including at $h=0$.  There
$r(\xi,0)=\kappa\xi^2/2+O(\xi^4)$ and
$u(\xi,0)=-S_B(\tau_*+\xi)$.
For $F(\xi,h)\equiv r(\xi,h)^2+a^2h u(\xi,h)^2$,
\begin{equation}
 F_\xi=\kappa^2\xi^3+2a^2b\beta h
       +O(\xi^5+h|\xi|+h^2).
\end{equation}
It follows that every critical point in a sufficiently small fixed
interval has $\xi<0$ and $|\xi|$ comparable to $h^{1/3}$.
The quotient $\eta^{-3}F_\xi(\eta z,\eta^3)$ extends smoothly to
$\eta=0$, where it equals $\kappa^2z^3+2a^2b\beta$.
Its unique zero $z=-\lambda$ is simple.  The implicit function
theorem, together with the preceding bound on all critical points,
gives a unique minimum with $z=-\lambda+O(\eta)$.
Uniformly for bounded $z$,
\begin{equation}
 F(\eta z,\eta^3)=a^2b^2\eta^3+
 \eta^4\left(\frac{\kappa^2z^4}{4}+2a^2b\beta z\right)
 +O(\eta^5).
\end{equation}
Evaluation at the minimum proves
\eqref{eq:closest-approach-scales}; equation~\eqref{eq:weyl-magnitude}
gives \eqref{eq:closest-approach-curvature}.
Between this point and the radial turn,
$\Sigma=a^2b^2\eta^3+O(\eta^4)$ uniformly, and their Mino-time
separation is $\lambda\eta+O(\eta^2)$.
Integrating $ds=\Sigma\,d\tau$ gives the stated proper-time gap.

For $L=0$, $E^2>1$, and $h<a^2(E^2-1)$, the polynomial
\begin{equation}
 R_h(r)=(E^2-1)r^4+2Mr^3+
 [a^2(E^2-1)-h]r^2+2M(a^2E^2+h)r-a^2h
\end{equation}
is strictly increasing for $r>0$ and has exactly one positive root.
The complete radial branches run to infinity in both directions:
a finite limit above the root would give
$|dr/ds|\ge\sqrt{R_h(r)}/(r^2+a^2)>0$, while the simple root
is reached in finite proper time and is a strict radial minimum.
Outside the fixed encounter, $r$ is uniformly bounded away from zero
as $h\downarrow0$.  Hence so is $\Sigma$, proving the global assertion.
\end{proof}

\begin{corollary}[Sharp curvature blow-up]\label{cor:sharp-curvature}
Under the hypotheses and notation of
Theorem~\ref{thm:sharp-terminal-asymptotics},
\begin{align}
 \mathcal W_2(s)&=0,\nonumber\\
 \mathcal W_1(s)
 &=48M^2\varrho^{-6}
 \left[1+\frac{3B}{7Md^2}\varrho
 +\left(\frac{29B^2}{196M^2d^4}-\frac{2}{3d^2}\right)
 \varrho^2+O(\varrho^3)\right].
 \label{eq:sharp-W1-asymptotics}
\end{align}
Equivalently,
\begin{equation}\label{eq:sharp-W1-leading}
 \mathcal W_1(s)=48M^2
 \left(\frac52\sqrt{2M}\,d\right)^{-12/5}
 u^{-12/5}\bigl(1+O(u^{2/5})\bigr).
\end{equation}
\end{corollary}

\begin{proof}
On the terminal equator, $y=0$, and
\eqref{eq:W1-real}-\eqref{eq:W2-real} give
$\mathcal W_2=0$ and $\mathcal W_1=48M^2r^{-6}$.  By
\eqref{eq:sharp-r-asymptotics},
\begin{align}
 r^{-6}
 &=\varrho^{-6}
 \left[1+6k_1\varrho+(21k_1^2-6k_2)\varrho^2
 +O(\varrho^3)\right]\nonumber\\
 &=\varrho^{-6}
 \left[1+\frac{3B}{7Md^2}\varrho
 +\left(\frac{29B^2}{196M^2d^4}-\frac{2}{3d^2}\right)
 \varrho^2+O(\varrho^3)\right],
\end{align}
using \eqref{eq:kerr-simple-zero-coefficients}.
Substitution of \eqref{eq:terminal-radial-scale} gives
\eqref{eq:sharp-W1-leading}.
\end{proof}

\begin{proposition}[Parallel-propagated tidal curvature]
\label{prop:pp-tidal}
Under the hypotheses of Theorem~\ref{thm:sharp-terminal-asymptotics},
let $U\equiv\dot\gamma$ and $u\equiv\Lambda-s$.  There is a parallel unit
spacelike field $N$ along the terminal interval, orthogonal to $U$, for
which
\begin{align}
 R(U,N,U,N)
 &=\frac{M}{r^3}+\frac{3MD^2}{r^5}\nonumber\\
 &=\frac6{25}u^{-2}
 \left[1+\frac{5B}{14Md^2}\varrho
 +\left(\frac{65B^2}{588M^2d^4}-\frac{2}{9d^2}\right)
 \varrho^2+O(\varrho^3)\right].
 \label{eq:pp-tidal-exact}
\end{align}
Fixing $s_0$ in that interval and putting
\begin{equation}
 \mathfrak I_1(s)\equiv\int_{s_0}^sR(U,N,U,N)(\sigma)\,d\sigma,
 \qquad
 \mathfrak I_2(s)\equiv\int_{s_0}^s \mathfrak I_1(\sigma)\,d\sigma,
\end{equation}
one has
\begin{align}
 \mathfrak I_1(s)&=\frac6{25}u^{-1}
 \left[1+\frac{25B}{42Md^2}\varrho
 +\left(\frac{325B^2}{588M^2d^4}-\frac{10}{9d^2}\right)
 \varrho^2+O(\varrho^{5/2})\right],
 \label{eq:pp-first-integral}\\
 \mathfrak I_2(s)&=\frac6{25}\log\frac1u+C_2+O(\varrho),
 \label{eq:pp-second-integral}
\end{align}
where $C_2\in\mathbb R$ depends on the fixed geodesic and on $s_0$.
\end{proposition}

\begin{proof}
On a terminal interval with $r,\Delta>0$, let $E_a$ be the dual
orthonormal frame of Lemma~\ref{lem:coframe}, and write $U=U^aE_a$.
The Killing-Yano form \eqref{eq:app-ky-form} has totally antisymmetric
covariant derivative by
\eqref{eq:app-connection-forms}-\eqref{eq:app-ky-crossing}.
The parallel-vector construction of
\cite{KamranMarck1986} therefore applies.  For
\begin{equation}
 V^\alpha\equiv f^\alpha{}_{\beta}U^\beta,
\end{equation}
the geodesic equation gives
\begin{equation}
 \nabla_UV^\alpha
 =U^\mu U^\beta\nabla_\mu f^\alpha{}_{\beta}=0.
\end{equation}
On the terminal equator,
\begin{equation}
 f=r e^2\wedge e^3,
 \qquad U^2=0,
 \qquad U^3=\frac{L-aE}{r}=\frac Dr.
\end{equation}
The last identity follows from $\xi^3$ preceding
\eqref{eq:app-carter-normalization}.  Hence $V=DE_2$, and
\begin{equation}
 N\equiv\frac{V}{|D|}=\operatorname{sgn}(D)E_2
\end{equation}
is parallel, unit spacelike, and orthogonal to $U$.

We now compute the frame component in the curvature convention
\eqref{eq:curvature-convention}.  At $\theta=\pi/2$, put
$A\equiv r^2+a^2$.  The dual frame is
\begin{equation}
 E_0=\frac{A\partial_t+a\partial_\phi}{r\sqrt\Delta},
 \qquad E_1=\frac{\sqrt\Delta}{r}\partial_r,
 \qquad E_2=\frac1r\partial_\theta,
 \qquad E_3=\frac{a\partial_t+\partial_\phi}{r}.
\end{equation}
Direct differentiation of the metric gives
\begin{align}
 R_{t\theta t\theta}
 &=\frac M{r^3}(r^2-2Mr+3a^2),
 &R_{r\theta r\theta}&=-\frac{Mr}{\Delta},\nonumber\\
 R_{t\theta\phi\theta}
 &=-\frac{Ma}{r^3}(3r^2-2Mr+3a^2),
 &R_{t\theta r\theta}&=R_{r\theta\phi\theta}=0,\nonumber\\
 R_{\phi\theta\phi\theta}
 &=\frac M{r^3}
 (2r^4+5a^2r^2+3a^4-2Ma^2r).
 \label{eq:pp-coordinate-components}
\end{align}
To determine the sign of the radial-polar component, we
differentiate before restricting to $\theta=\pi/2$.  The relevant
Christoffel symbols are
\begin{align}
 \Gamma^r_{\theta\theta}&=-\frac{r\Delta}{\Sigma},
 &\Gamma^r_{r\theta}
 &=\frac{\partial_\theta\Sigma}{2\Sigma}
 =-\frac{a^2\sin\theta\cos\theta}{\Sigma},\nonumber\\
 \Gamma^r_{rr}&=\frac r\Sigma-\frac{r-M}{\Delta},
 &\Gamma^\theta_{r\theta}&=\frac r\Sigma,
 &\Gamma^\theta_{\theta\theta}
 &=\frac{\partial_\theta\Sigma}{2\Sigma}.
 \label{eq:pp-representative-christoffels}
\end{align}
Since $\Gamma^\theta_{\theta\theta}=\Gamma^r_{r\theta}$, the mixed
quadratic terms cancel in convention~\eqref{eq:curvature-convention}:
\begin{equation}
 R^r{}_{\theta r\theta}
 =\partial_r\Gamma^r_{\theta\theta}
 -\partial_\theta\Gamma^r_{r\theta}
 +\Gamma^r_{\theta\theta}
 \left(\Gamma^r_{rr}-\Gamma^\theta_{r\theta}\right).
 \label{eq:pp-representative-riemann}
\end{equation}
At the equator,
\begin{align}
 \Gamma^r_{\theta\theta}&=-\frac\Delta r,
 &\partial_r\Gamma^r_{\theta\theta}
 &=-\frac{\Delta'}r+\frac\Delta{r^2},\nonumber\\
 \partial_\theta\Gamma^r_{r\theta}&=\frac{a^2}{r^2},
 &\Gamma^r_{rr}-\Gamma^\theta_{r\theta}
 &=-\frac{r-M}{\Delta}.
\end{align}
Using $\Delta'=2(r-M)$ and $\Delta-a^2=r^2-2Mr$ in
\eqref{eq:pp-representative-riemann} gives
\begin{equation}
 R^r{}_{\theta r\theta}=-\frac Mr,
 \qquad
 R_{r\theta r\theta}
 =g_{rr}R^r{}_{\theta r\theta}=-\frac{Mr}{\Delta},
\end{equation}
The same calculation gives the orbit components in
\eqref{eq:pp-coordinate-components}.

The contractions with the orthonormal frame simplify to
\begin{align}
 A^2R_{t\theta t\theta}+2aA R_{t\theta\phi\theta}
 +a^2R_{\phi\theta\phi\theta}&=Mr\Delta,\nonumber\\
 a^2R_{t\theta t\theta}+2aR_{t\theta\phi\theta}
 +R_{\phi\theta\phi\theta}&=2Mr,\nonumber\\
 aA R_{t\theta t\theta}+(A+a^2)R_{t\theta\phi\theta}
 +aR_{\phi\theta\phi\theta}&=0.
 \label{eq:pp-orbit-contractions}
\end{align}
On putting $m\equiv M/r^3$, we obtain
\begin{equation}\label{eq:pp-frame-components}
 R_{0202}=m,\qquad R_{1212}=-m,\qquad R_{3232}=2m,
\end{equation}
and $R_{a2b2}=0$ for distinct $a,b\in\{0,1,3\}$.  Since
$U^2=0$ and
\begin{equation}
 -(U^0)^2+(U^1)^2+(U^3)^2=-1,
\end{equation}
equation~\eqref{eq:pp-frame-components} gives
\begin{align}
 R(U,N,U,N)
 &=m\left((U^0)^2-(U^1)^2+2(U^3)^2\right)\nonumber\\
 &=m\left(1+3(U^3)^2\right)
 =\frac{M}{r^3}+\frac{3MD^2}{r^5}.
\end{align}
For its expansion, we use
\begin{equation}
 \varrho^5=\frac{25}{2}MD^2u^2,
 \qquad
 r=\varrho\left(1-k_1\varrho+k_2\varrho^2
 +O(\varrho^3)\right),
\end{equation}
where $k_1,k_2$ are given in
\eqref{eq:kerr-simple-zero-coefficients}.  It follows that
\begin{align}
 r^{-5}
 &=\varrho^{-5}\left[1+5k_1\varrho
 +(15k_1^2-5k_2)\varrho^2+O(\varrho^3)\right],\nonumber\\
 r^{-3}
 &=\varrho^{-3}\left[1+O(\varrho)\right].
 \label{eq:pp-negative-powers}
\end{align}
Factoring $3Md^2\varrho^{-5}=6u^{-2}/25$ from the exact tidal
scalar, the $Mr^{-3}$ term contributes $\varrho^2/(3d^2)$ to the
second relative coefficient.  Hence
\begin{align}
 5k_1&=\frac{5B}{14Md^2},\nonumber\\
 15k_1^2-5k_2+\frac1{3d^2}
 &=\frac{65B^2}{588M^2d^4}-\frac{2}{9d^2}.
 \label{eq:pp-tidal-coefficients}
\end{align}
Substituting these identities yields
\eqref{eq:pp-tidal-exact}.

To integrate the expansion, put
\begin{equation}
 A_1\equiv\frac{5B}{14Md^2},\qquad
 A_2\equiv\frac{65B^2}{588M^2d^4}-\frac{2}{9d^2},\qquad
 \varrho(v)\equiv\left(\frac52\sqrt{2M}\,d\,v\right)^{2/5}.
\end{equation}
With $u_0\equiv\Lambda-s_0$, direct integration gives
\begin{align}
 \int_u^{u_0}v^{-2}\,dv&=u^{-1}+O(1),\nonumber\\
 \int_u^{u_0}v^{-2}\varrho(v)\,dv
 &=\frac53u^{-1}\varrho(u)+O(1),\nonumber\\
 \int_u^{u_0}v^{-2}\varrho(v)^2\,dv
 &=5u^{-1}\varrho(u)^2+O(1).
 \label{eq:pp-elementary-integrals}
\end{align}
Since the remainder $O(v^{-4/5})$ is integrable, its relative
contribution after factoring $u^{-1}$ is $O(\varrho^{5/2})$.
This proves \eqref{eq:pp-first-integral}, with coefficients
$(5/3)A_1$ and $5A_2$.  Finally,
\begin{equation}
 \mathfrak I_1(s)-\frac6{25}u^{-1}=O(u^{-3/5})
\end{equation}
is integrable at zero, giving a finite limit $C_2$ and the error
$O(u^{2/5})=O(\varrho)$ in \eqref{eq:pp-second-integral}.
\end{proof}

\begin{proposition}[Exact reduction of the polar Jacobi equation]
\label{prop:polar-jacobi}
Under the hypotheses of Proposition~\ref{prop:pp-tidal}, set
\begin{equation}
 B\equiv L^2+a^2(1-E^2),\qquad
 k\equiv\sqrt{2M}\,d,\qquad
 c\equiv\left(\frac52k\right)^{2/5},
 \qquad w(s)\equiv\tau_*-\tau(s).
\end{equation}
Let $C_B,S_B$ be the solutions of $z''+Bz=0$ normalized by
\begin{equation}
 C_B(0)=1,\quad C_B'(0)=0,\qquad
 S_B(0)=0,\quad S_B'(0)=1.
\end{equation}
Thus they are the cosine and normalized sine when $B>0$, the
functions $1,w$ when $B=0$, and the hyperbolic cosine and normalized
hyperbolic sine when $B<0$.  The polar line spanned by $N$ is invariant
under the Jacobi equation, and $J=jN$ is a Jacobi field precisely when
\begin{equation}\label{eq:polar-jacobi-equation}
 \frac{d^2j}{du^2}+K(u)j=0,
 \qquad K(u)\equiv\frac{M}{r^3}+\frac{3Md^2}{r^5}.
\end{equation}
An exact fundamental system is
\begin{equation}\label{eq:polar-exact-solutions}
 j_c(s)\equiv r(s)C_B(w(s)),\qquad
 j_s(s)\equiv r(s)S_B(w(s)).
\end{equation}
Equivalently, the rescaled position and momentum
\begin{equation}\label{eq:polar-position-momentum}
 v\equiv\frac{j}{r},\qquad p\equiv r\dot j-\dot rj
\end{equation}
satisfy $dv/d\tau=p$ and $dp/d\tau=-Bv$.  Consequently
\begin{equation}\label{eq:polar-conserved-energy}
 \mathcal E\equiv p^2+Bv^2
\end{equation}
is constant.  With $h\equiv\tau(s)-\tau(s_0)$, their exact evolution is
\begin{equation}\label{eq:polar-transfer}
 \begin{pmatrix}v(s)\\p(s)\end{pmatrix}
 =\begin{pmatrix}C_B(h)&S_B(h)\\-B S_B(h)&C_B(h)\end{pmatrix}
 \begin{pmatrix}v(s_0)\\p(s_0)\end{pmatrix}.
\end{equation}
Both this matrix and the change of variables
$(v,p)\mapsto(j,\dot j)=(rv,\dot r v+p/r)$ have determinant one.

The normalized solutions
\begin{equation}\label{eq:polar-normalization}
 j_{2/5}\equiv c^{-1}j_c,\qquad
 j_{3/5}\equiv\frac{k}{2c^{3/2}}j_s
\end{equation}
satisfy
\begin{align}
 j_{2/5}(u)
 &=u^{2/5}\left[1-\frac{15B}{14Md^2}\varrho
 +\left(\frac{20B^2}{147M^2d^4}+\frac1{9d^2}\right)
 \varrho^2+O(\varrho^3)\right],
 \label{eq:polar-jacobi-first}\\
 j_{3/5}(u)
 &=u^{3/5}\left[1-\frac{5B}{14Md^2}\varrho
 +\left(\frac{B^2}{196M^2d^4}+\frac1{15d^2}\right)
 \varrho^2+O(\varrho^3)\right].
 \label{eq:polar-jacobi-second}
\end{align}
After division by the respective leading powers, both expansions
are convergent power series in $\varrho$ near zero.  With primes
denoting $u$-derivatives,
\begin{equation}\label{eq:polar-wronskian}
 j_{2/5}j_{3/5}'-j_{2/5}'j_{3/5}=\frac15.
\end{equation}
In particular, every polar Jacobi field tends to zero at the ring.
\end{proposition}

\begin{proof}
Equatorial reflection fixes $U$ and reverses $N$; its negative
eigenspace is the polar line.  Curvature equivariance in
convention~\eqref{eq:curvature-convention} gives
\begin{equation}\label{eq:polar-curvature-operator}
 R(U,N)U=-K N.
\end{equation}
Since $N$ is parallel, the Jacobi equation
$\nabla_U^2J=R(U,J)U$ is equivalent to
\eqref{eq:polar-jacobi-equation}; replacing $s$ by $u=\Lambda-s$
does not change the second derivative.

The exact radial equation gives
\begin{equation}\label{eq:polar-radial-acceleration}
 \dot r^2=\frac{2Md^2}{r^3}-\frac{B}{r^2}
       +\frac{2M}{r}+E^2-1,
 \qquad
 \ddot r=-\frac{3Md^2}{r^4}+\frac{B}{r^3}-\frac{M}{r^2}.
\end{equation}
The terminal branch has $\dot r\ne0$, so this differentiation is
legitimate.  Since $\dot w=-r^{-2}$, every twice differentiable $z$
satisfies
\begin{equation}\label{eq:polar-exact-reduction}
 \frac{d^2}{ds^2}\bigl(rz(w)\bigr)
       +Krz(w)
 =\frac{z''(w)+Bz(w)}{r^3}.
\end{equation}
Differentiating \eqref{eq:polar-position-momentum} gives
\[
 \dot v=\frac{p}{r^2},\qquad
 \dot p=r\ddot j-\ddot rj=-\frac{Bv}{r^2}.
\]
Since $d\tau/ds=r^{-2}$, these are the oscillator equations.
They give energy conservation, \eqref{eq:polar-transfer}, and
\eqref{eq:polar-exact-solutions}.  The determinants equal one by
$C_B^2+B S_B^2=1$.  Since $C_BS_B'-C_B'S_B=1$ and
$dw/du=r^{-2}$, the unnormalized $u$-Wronskian is one;
normalization multiplies it by $k/(2c^{5/2})=1/5$.

The next two terms follow by Taylor expansion.  The radial
integral gives
\begin{align}
 w&=\int_0^r\frac{dy}{\sqrt{R(y)}}\nonumber\\
 &=\frac{2\sqrt r}{k}\left[1+\frac{B}{12Md^2}r
 +\left(\frac{3B^2}{160M^2d^4}-\frac1{10d^2}\right)r^2
 +O(r^3)\right].\label{eq:polar-mino-expansion}
\end{align}
Using $C_B(w)=1-Bw^2/2+B^2w^4/24+O(w^6)$ and
$S_B(w)=w-Bw^3/6+B^2w^5/120+O(w^7)$, one obtains
\begin{align}
 rC_B(w)&=r\left[1-\frac{B}{Md^2}r+O(r^3)\right],\nonumber\\
 rS_B(w)&=\frac{2r^{3/2}}k\left[1-\frac{B}{4Md^2}r
 -\left(\frac{B^2}{32M^2d^4}+\frac1{10d^2}\right)r^2
 +O(r^3)\right].\label{eq:polar-r-series}
\end{align}
Substituting~\eqref{eq:sharp-r-asymptotics} and
$\varrho=cu^{2/5}$ proves
\eqref{eq:polar-jacobi-first}-\eqref{eq:polar-jacobi-second}.
Finally, $R(r)/(k^2r)$ is analytic and equals one at zero.
Consequently $w/\sqrt r$ and $u/r^{5/2}$ are analytic in $r$ near
zero.  The analytic inverse theorem makes $r$ analytic in $\varrho$.
The parity of $C_B$ and $S_B$ then proves the asserted convergence of
the normalized series.
\end{proof}

\begin{corollary}[Zeros of polar Jacobi fields]\label{cor:polar-zeros}
Fix $s_0$ in the terminal equatorial interval.  The scalar solution
with $j(s_0)=0$ and $\dot j(s_0)=1$ is
\begin{equation}\label{eq:polar-initial-point-solution}
 j(s)=r(s_0)r(s)S_B\bigl(\tau(s)-\tau(s_0)\bigr).
\end{equation}
If $B\le0$, it has no zero for $s_0<s<\Lambda$.  If $B>0$, its
zeros in this interval are exactly the instants satisfying
\begin{equation}\label{eq:polar-zero-times}
 \tau(s)-\tau(s_0)=\frac{n\pi}{\sqrt B},
 \qquad n\in\mathbb N,\quad n\ge1.
\end{equation}
There are only finitely many such instants.  Its leading terminal
order is $u^{2/5}$ unless
$S_B(\tau_*-\tau(s_0))=0$, in which case it is $u^{3/5}$.
\end{corollary}

\begin{proof}
We apply~\eqref{eq:polar-exact-reduction} with the increasing
parameter $\tau(s)-\tau(s_0)$.  The initial values follow from
$\dot\tau=r^{-2}$.  Since $r>0$, the zeros are exactly those of
$S_B$.  The endpoint $\tau_*$ is finite, so only finitely many
positive integer multiples in~\eqref{eq:polar-zero-times} occur.
If the endpoint value of $S_B$ is nonzero, multiplication by $r$
gives order $u^{2/5}$.  At a zero, $S_B'$ is nonzero, and
$\tau_*-\tau(s)\sim2\sqrt r/k$ gives order $r^{3/2}$, equivalently
$u^{3/5}$.
\end{proof}

\begin{remark}
Each interior zero in Corollary~\ref{cor:polar-zeros} gives a conjugate
point.  The polar sector alone determines neither the full Jacobi
volume nor Tipler strength.
\end{remark}

\begin{remark}
The coefficient $6/25$ agrees with
\cite[Section~V.B, equations~(V.27)-(V.29)]{ChruscielMaliborskiYunes};
Proposition~\ref{prop:polar-jacobi} refines that leading calculation.
The iterated scalar integral $\mathfrak I_2$ is distinct from a
Jacobi field.  For a locally Lipschitz metric, curvature is
distributional, so divergence of a curvature component alone does
not exclude such an extension. The $C^0$ proof in
Section~\ref{sec:two-sheet-fan} is independent of these considerations.
\end{remark}

\begin{corollary}[Independent curvature obstruction]\label{thm:main-c11}
Every spacetime in Theorem~\ref{thm:main-c0} is $C^{1,1}_{\mathrm{loc}}$-inextendible.
The terminal classification and curvature divergence give a
direct proof, independent of the causal homotopy construction.
\end{corollary}

\begin{proof}
In either time orientation, ring localization gives
$\Sigma(\gamma(s))\to0$ along every finite-proper-time terminal
unit timelike geodesic, also on covers.  Hence
\begin{equation}
 \mathcal J(\gamma(s))
 =\bigl(\mathcal W_1^2+\mathcal W_2^2\bigr)(\gamma(s))
 =\frac{(48M^2)^2}{\Sigma(\gamma(s))^6}
 \longrightarrow+\infty.
\end{equation}
The globally defined scalar $\mathcal J$ therefore satisfies the
hypotheses of Theorem~\ref{thm:general-curvature-criterion}.
\end{proof}

Consequently these spacetimes are $C^k$-inextendible for $k\ge2$
and analytically inextendible, since $C^2$ metrics are locally $C^{1,1}$.

\section{Axial geometry and boundary symmetries}
\label{sec:symmetry}

In this section we study the stationary and axial directions near the
ring.  The constructions apply for $M>0$, $a\ne0$ in all three
models and their covers.  On a cover, the curves and Killing fields
are lifted locally; an embedded lift of an entire base block is not
needed.

\subsection{Exact disc-crossing bridges}

We first construct timelike curves through the regular disc and
estimate their stationary and axial displacements.  The two
combinations of displacements have different orders of vanishing.

\begin{lemma}[Timelike bridge through a regular disc]
\label{lem:exact-ring-bridge}
There is $\rho_0>0$, depending only on $M$, $a$, and $c$,
with the following property.  For every $0<\rho<\rho_0$, the normal
semicircle
\begin{equation}\label{eq:fan-normal-semicircle}
 r_\rho(\psi)\equiv\rho\cos\frac\psi2,
 \qquad
 a\cos\theta_\rho(\psi)\equiv\rho\sin\frac\psi2,
 \qquad 0\le\psi\le2\pi,
\end{equation}
has a strict future-timelike stationary-axial lift from the positive
equator to the negative equator.  Reversing the transverse path and
using the same lift rule gives a strict future-timelike return through
the same regular disc.

Let $\Delta t_{\mathrm{br}}(\rho)$ and
$\Delta\phi_{\mathrm{br}}(\rho)$ be the total increments of the outgoing
and return bridges, computed with one continuous lift of $\phi$, and set
\begin{equation}\label{eq:fan-b-e-def}
 b_\rho\equiv\varepsilon\operatorname{sgn}(a)
          \Delta\phi_{\mathrm{br}}(\rho),
 \qquad
 e_\rho\equiv\varepsilon\bigl(
          \Delta t_{\mathrm{br}}(\rho)
          -a\Delta\phi_{\mathrm{br}}(\rho)\bigr).
\end{equation}
Then $b_\rho>0$, $e_\rho>0$, and
\begin{equation}\label{eq:fan-bridge-orders}
 b_\rho=O(\rho),
 \qquad e_\rho=O(\rho^3),
 \qquad
 \varepsilon\Delta t_{\mathrm{br}}
 =|a|b_\rho+e_\rho.
\end{equation}
The constants in the two remainder estimates are uniform for
$0<\rho<\rho_0$.
\end{lemma}

\begin{proof}
The semicircle has $\Sigma=\rho^2$.  Its transverse tangent
$W_\rho\equiv d(r_\rho,\theta_\rho)/d\psi$ satisfies
\begin{equation}\label{eq:fan-Q-F}
 Q_\rho\equiv\sqrt{\mathfrak q(W_\rho,W_\rho)}
 =\frac{\rho^2}{2}F_\rho(\psi),
\end{equation}
where
\begin{equation}\label{eq:fan-F}
 F_\rho(\psi)^2
 \equiv\frac{\sin^2(\psi/2)}{\Delta(r_\rho(\psi))}
 +\frac{\cos^2(\psi/2)}
 {a^2\sin^2\theta_\rho(\psi)}.
\end{equation}
Choose $\rho_0<|a|/2$ sufficiently small.  The two denominators
are then bounded above and below by positive constants, uniformly in
the indicated range.  Hence
\begin{equation}\label{eq:fan-F-bounds}
 0<C_1\le F_\rho(\psi)\le C_2<\infty.
\end{equation}
The lift $V\equiv W_\rho+\dot tT+\dot\phi\Phi$ is prescribed by
\begin{equation}\label{eq:fan-lift-rule}
 \beta(V)=0,\qquad
 \alpha(V)=\varepsilon c\sqrt{\frac\Sigma\Delta}\,Q_\rho.
\end{equation}
Solving for the orbit components and substituting into the coframe
decomposition gives
\begin{equation}\label{eq:fan-orbit-velocities}
 \dot t=\frac{r^2+a^2}{\Sigma}\alpha(V),\qquad
 \dot\phi=\frac a\Sigma\alpha(V),
\end{equation}
\begin{equation}\label{eq:fan-bridge-timelike}
 g(V,V)=-(c^2-1)Q_\rho^2<0,\qquad
 \varepsilon\alpha(V)>0.
\end{equation}
The reversed transverse path has the same speed $Q_\rho$, so its
lift is again future timelike.  Along the outward traversal we have
\begin{equation}\label{eq:fan-per-psi}
 \frac{d\phi}{d\psi}
 =\varepsilon\frac{ca\rho F_\rho(\psi)}
 {2\sqrt{\Delta(r_\rho(\psi))}},\qquad
 \frac{d(t-a\phi)}{d\psi}
 =\varepsilon\frac{c\rho^3\cos^2(\psi/2)F_\rho(\psi)}
 {2\sqrt{\Delta(r_\rho(\psi))}}.
\end{equation}
The second equality uses $(dt-a\,d\phi)(V)=r^2\alpha(V)/\Sigma$.
Since the two traversals have equal increments, we obtain
\begin{align}
 b_\rho
 &=c|a|\rho\int_0^{2\pi}
 \frac{F_\rho(\psi)}{\sqrt{\Delta(r_\rho(\psi))}}\,d\psi,
 \label{eq:fan-b-integral}\\
 e_\rho
 &=c\rho^3\int_0^{2\pi}
 \frac{\cos^2(\psi/2)F_\rho(\psi)}
 {\sqrt{\Delta(r_\rho(\psi))}}\,d\psi.
 \label{eq:fan-e-integral}
\end{align}
The integrals multiplying $\rho$ and $\rho^3$ are positive and
uniformly bounded.  The final identity follows from the definition of
$e_\rho$.
\end{proof}

\subsection{Timelike detours through both sheets}

\begin{lemma}[Two-sheet timelike detours]
\label{lem:timelike-lollipops}
There is $\rho_1>0$ such that, for every $A_0\in(0,\pi)$,
$0<\rho<\rho_1$, and positive-equatorial point $x$ with $r(x)=\rho$,
there is a continuous family of future-causal paths
$L_{\rho,\lambda}^x$, $0\le\lambda\le1$, from $x$ to
$\mathscr S_{h_{\rho,\lambda}}(x)$, with the following properties.
The path $L_{\rho,0}^x$ is constant, every nonconstant piece is
strictly future timelike, and
\begin{equation}\label{eq:fan-h-uniform}
 h_{\rho,0}=0,
 \qquad
 \sup_{0\le\lambda\le1}|h_{\rho,\lambda}|
 \le C(A_0+\rho).
\end{equation}
Here the angular component is given a continuous real lift and
$|(\xi,\eta)|\equiv|\xi|+|\eta|$.
The last path contains a negative-equatorial axial arc of width $A_0$
and length
\begin{equation}\label{eq:fan-negative-arc-length}
 L_-(\rho)
 =A_0\sqrt{\frac{2Ma^2}{\rho}-a^2-\rho^2}
 \longrightarrow+\infty.
\end{equation}
The constant $C$ depends only on $M$, $a$, and $c$.
\end{lemma}

\begin{proof}
We concatenate the outgoing bridge, an orbit segment of the future
timelike field $Y_-$ with parameter length $A_0$, and a translated
copy of the return bridge.  The endpoint displacement is
\begin{equation}\label{eq:fan-final-lollipop-endpoint}
 h_{\rho,1}
 =\bigl(\varepsilon(|a|b_\rho+e_\rho),
        \varepsilon\operatorname{sgn}(a)(b_\rho-A_0)\bigr).
\end{equation}
The norm of the axial field gives
\eqref{eq:fan-negative-arc-length}.

For $0\le\lambda\le1/2$, follow the semicircle to
$\psi=4\pi\lambda$ and retrace it with the future lift.  The
displacement is twice a partial integral in \eqref{eq:fan-per-psi},
hence uniformly $O(\rho)$.  For $1/2\le\lambda\le1$, keep both
bridges and insert the axial segment of length $(2\lambda-1)A_0$,
translating the return bridge.  The two families agree at $1/2$.  By using fixed parameter
intervals and allowing pauses, we obtain a jointly continuous family.
The lift and the translations preserve the future direction, and the
additional displacement is at most $A_0$.  This proves
\eqref{eq:fan-h-uniform}.
\end{proof}

\subsection{The orbit Gram form}

\begin{proposition}[Helical norm and its determinant]\label{prop:helical-form}
For $K_\Omega\equiv T+\Omega\Phi$,
\begin{equation}\label{eq:helical-form}
 g(K_\Omega,K_\Omega)
 =-1+\Omega^2(r^2+a^2)\sin^2\theta
 +\frac{2Mr}{\Sigma}\bigl(1-a\Omega\sin^2\theta\bigr)^2.
\end{equation}
Where $g_{\phi\phi}>0$, the quadratic in $\Omega$ has the minimum
\begin{equation}\label{eq:helical-general-minimum}
 \min_\Omega g(K_\Omega,K_\Omega)
 =\frac{g_{tt}g_{\phi\phi}-g_{t\phi}^2}{g_{\phi\phi}}
 =-\frac{\Delta\sin^2\theta}{g_{\phi\phi}}.
\end{equation}
\end{proposition}

\begin{proof}
The orbit part of \eqref{eq:kerr-metric} can be written as
\begin{equation}
 -dt^2+(r^2+a^2)\sin^2\theta\,d\phi^2
 +\frac{2Mr}{\Sigma}(dt-a\sin^2\theta\,d\phi)^2,
\end{equation}
which gives \eqref{eq:helical-form}.  The minimum of
$g_{tt}+2\Omega g_{t\phi}+\Omega^2g_{\phi\phi}$ is
$g_{tt}-g_{t\phi}^2/g_{\phi\phi}$; using
\eqref{eq:determinants} yields \eqref{eq:helical-general-minimum}.
\end{proof}

\subsection{Negative-sheet arcs and the critical angular scale}

\begin{definition}[Axial width and product wedges]\label{def:axial-width}
Fix a full two-sheet ring collar $B_\alpha$ in one of the base
models, with its Kerr coordinate map $\Psi_\alpha$, and let
$U\subset\KerrSpace$ be open.  For $\rho>0$, let
$A_{U,\alpha}(\rho)$ be the supremum of the lengths $|J|<2\pi$ of
intervals $J\subset\mathbb R$ for which, after choosing a continuous
angular lift, there is a $t_0$ such that the embedded axial arc
\begin{equation}
 \phi\longmapsto
 \Psi_\alpha(t_0,-\rho,\pi/2,\phi),\qquad \phi\in J,
\end{equation}
is contained in $U\cap B_\alpha$.  Set $A_{U,\alpha}(\rho)\equiv0$ if there
is no such arc.  A negative-ring product wedge in $B_\alpha$ is an open
set containing, for all sufficiently small $\rho>0$, such arcs with one
fixed interval $J$ of positive length, together with a product
neighborhood in $(t,r,\theta)$.
For a connected cover, use lifts of these embedded base arcs through
arbitrary chosen initial points over them, and take the supremum over
the lifts contained in $U$.  The wedge condition uses the corresponding
lifted product neighborhoods.  Such lifts remain embedded, and the
local isometry preserves their length; no global angular coordinate
on the covering spacetime is required.
\end{definition}

\begin{theorem}[Negative-sheet axial obstruction]\label{thm:negative-sheet}
On the equatorial negative sheet, $r=-\rho<0$,
\begin{equation}\label{eq:negative-axial}
 g(\Phi,\Phi)=\rho^2+a^2-\frac{2Ma^2}{\rho}.
\end{equation}
Hence $\Phi$ is timelike for all sufficiently small $\rho>0$.  An axial arc of angular width $A$ at fixed $(t,r,\theta)$ has Lorentzian length
\begin{equation}\label{eq:axial-length}
 L=A\sqrt{\frac{2Ma^2}{\rho}-a^2-\rho^2}.
\end{equation}
Consequently:
\begin{enumerate}[label=(\alph*)]
 \item every negative-ring product wedge has infinite future timelike
 diameter;
 \item if $W$ is a coordinate neighborhood of finite timelike diameter
 in a $C^0$ extension and
 \begin{equation}
  U\equiv\iota^{-1}\bigl(W\cap\iota(\KerrSpace)\bigr),
 \end{equation}
 then
 \begin{equation}\label{eq:critical-width-upper}
  A_{U,\alpha}(\rho)\le C_W\sqrt\rho
 \end{equation}
 for all sufficiently small $\rho$;
 \item a sequence of contained arcs with widths
 $A_j/\sqrt{\rho_j}\to\infty$ is incompatible with such an extension
 chart;
 \item for the single-arc length mechanism, the scale is sharp:
 $A(\rho)/\sqrt\rho\to\infty$ makes the lengths diverge,
 $A(\rho)\asymp\sqrt\rho$ keeps them of constant order, and
 $A(\rho)=o(\sqrt\rho)$ makes them tend to zero.  More precisely,
 \begin{equation}\label{eq:critical-width-limit}
  \frac{A(\rho)}{\sqrt\rho}\longrightarrow k\in[0,\infty)
  \quad\Longrightarrow\quad
  L(\rho)\longrightarrow k\sqrt{2M}\,|a|.
 \end{equation}
\end{enumerate}
\end{theorem}

\begin{proof}
At the negative equator, \eqref{eq:orbit-components} reduces to
\eqref{eq:negative-axial}.  Since the metric coefficients are
constant along the axial orbits, their length is given by
\eqref{eq:axial-length}.  Orienting these arcs to the future proves
(a): their lengths grow as $\sqrt{2M}|a|A/\sqrt\rho$.
For an admissible arc in $W$, Proposition~\ref{prop:finite-tdiam} gives
\begin{equation}
 A\sqrt{\frac{2Ma^2}{\rho}-a^2-\rho^2}
 \le D_W\equiv\tdiam^+(W,\widetilde g).
\end{equation}
For sufficiently small $\rho$, the square root is at least
$\sqrt{M}|a|/\sqrt\rho$.  Taking the supremum over the admissible
arcs proves \eqref{eq:critical-width-upper} with
$C_W\equiv D_W/(\sqrt M|a|)$.  Part~(c) follows by contraposition.
Finally,
\begin{equation}
 L(\rho)=\sqrt{2M}|a|\,
 \frac{A(\rho)}{\sqrt\rho}\bigl(1+O(\rho)\bigr),
\end{equation}
which gives (d) and \eqref{eq:critical-width-limit}.
For a fixed width $A_0$, the error in the leading coefficient can
be written exactly.  With $c_*\equiv2Ma^2$ and $b_*\equiv a^2$, we
have
\begin{equation}\label{eq:sharp-normalized-axial-length}
 A_0\sqrt{c_*}-\sqrt\rho\,L(\rho)
 =\frac{A_0(b_*\rho+\rho^3)}
 {\sqrt{c_*}+\sqrt{c_*-b_*\rho-\rho^3}}.
\end{equation}
In particular the normalized error is $O(\rho)$; for $0<\rho\le1$
with positive radicand it lies between $0$ and
$A_0(b_*+1)\rho/\sqrt{c_*}$.
\end{proof}

\begin{corollary}[Vanishing normalized axial width in nested charts]
\label{cor:nested-width}
Let $(\widetilde M,\widetilde g,\iota)$ be a $C^0$ extension and let
$q\in\partial\iota(\KerrSpace)$.  Choose nested coordinate
neighborhoods $W_n$ of $q$ such that
\begin{equation}
 D_n\equiv \tdiam^+(W_n,\widetilde g)\longrightarrow0,
\end{equation}
as supplied by Proposition~\ref{prop:finite-tdiam}, and put
\begin{equation}
 U_n\equiv\iota^{-1}\bigl(W_n\cap\iota(\KerrSpace)\bigr).
\end{equation}
For every fixed base ring collar $B_\alpha$,
\begin{equation}\label{eq:nested-width}
 \limsup_{\rho\downarrow0}
 \frac{A_{U_n,\alpha}(\rho)}{\sqrt\rho}
 \le \frac{D_n}{\sqrt{2M}\,|a|},
 \qquad
 \lim_{n\to\infty}\limsup_{\rho\downarrow0}
 \frac{A_{U_n,\alpha}(\rho)}{\sqrt\rho}=0.
\end{equation}
In particular, there cannot be a constant $c>0$, radii
$\rho_n\downarrow0$, and admissible arcs of Definition~\ref{def:axial-width} contained in
$U_n$ (lifted arcs on a cover) whose widths $A_n$ satisfy $A_n\ge c\sqrt{\rho_n}$ for all sufficiently
large $n$.
\end{corollary}

\begin{proof}
Every admissible arc in $U_n$ has length at most $D_n$, so
\begin{equation}
 \frac{A}{\sqrt\rho}
 \le
 \frac{D_n}{\sqrt{2Ma^2-a^2\rho-\rho^3}}
\end{equation}
for small $\rho$.  Taking the supremum over the admissible arcs and then
$\limsup_{\rho\downarrow0}$ proves the first inequality.  The
second follows from $D_n\to0$.
The same bound with $A=A_n$, $\rho=\rho_n$ proves the last assertion.
\end{proof}

\begin{remark}
On a base model, if $U$ contains a full negative-sheet axial circle in $B_\alpha$, that
circle is timelike and may be traversed repeatedly.  Thus
$\tdiam^+(U,g)=\infty$ already at one fixed radius.  The width estimate applies also to local inverse images containing
only fragments of an axial circle.
\end{remark}

\subsection{Positive-sheet helical directions}

\begin{theorem}[Positive-sheet helical minimum]\label{thm:positive-minimum}
On the positive equatorial sheet,
\begin{equation}\label{eq:positive-helical}
 g(K_\Omega,K_\Omega)
 =-1+\Omega^2(r^2+a^2)+\frac{2M}{r}(1-a\Omega)^2.
\end{equation}
For every fixed $\Omega\ne a^{-1}$ this tends to $+\infty$ as $r\downarrow0$, while for $\Omega=a^{-1}$ it equals $r^2/a^2>0$.  The pointwise minimizer satisfies
\begin{equation}\label{eq:positive-minimum}
 \min_\Omega g(K_\Omega,K_\Omega)
 =-\frac{r(r^2-2Mr+a^2)}{r^3+a^2r+2Ma^2}
 =-\frac r{2M}+O(r^2).
\end{equation}
For all sufficiently small $r>0$ this minimum is negative, so the
timelike set in the following supremum is nonempty.  Consequently the
largest proper-time density among timelike vectors
with unit $T$ coefficient is
\begin{equation}\label{eq:positive-proper-density}
 \sup_{\Omega:\,g(K_\Omega,K_\Omega)<0}
 \sqrt{-g(K_\Omega,K_\Omega)}
 =\sqrt{\frac r{2M}}\,(1+O(r)).
\end{equation}
More generally, there are $\epsilon,C>0$ such that on
\begin{equation}
 0<r<\epsilon,\qquad \sin^2\theta\ge\frac12,
\end{equation}
one has
\begin{equation}
 \min_{\Omega\in\mathbb R}g(K_\Omega,K_\Omega)\ge-C.
\end{equation}
\end{theorem}

\begin{proof}
On the equator, \eqref{eq:helical-form} gives
\eqref{eq:positive-helical}, from which the assertions for fixed
$\Omega$ follow.  Substituting $g_{\phi\phi}=r^2+a^2+2Ma^2/r$
into \eqref{eq:helical-general-minimum} and expanding at $r=0$
proves \eqref{eq:positive-minimum}.  Taking the square root of the
negative of this expression proves \eqref{eq:positive-proper-density}.
For $r>0$ and $\sin^2\theta\ge1/2$, one has
$g_{\phi\phi}\ge a^2/2$, while $\Delta\sin^2\theta$ is bounded
on a small collar.  The same minimum formula proves the uniform
lower bound.
\end{proof}

Thus the axial direction whose norm diverges is timelike on the
negative sheet and spacelike on the positive sheet.  The minimum on
the positive sheet depends on the normalization of the Killing fields
and, by itself, does not give an invariant obstruction to a $C^0$
extension.  On both sheets the equatorial orbit determinant is
$-\Delta\to-a^2$.

\subsection{A cusp counterexample}

\begin{proposition}[A smooth timelike cusp with $g_{\phi\phi}=-1/r$]\label{prop:flat-cusp}
Let $(\R^{1,3},\eta)$ have coordinates $(T,X,Y,Z)$ and metric
\begin{equation}
 \eta\equiv-dT^2+dX^2+dY^2+dZ^2.
\end{equation}
Fix $\beta>0$ and consider the open cusp
\begin{equation}
 \Omega_\beta\equiv\{0<X<\varepsilon,\ |T|<X^\beta,\ |Y|<1,\ |Z|<1\}.
\end{equation}
On $\Omega_\beta$ introduce
\begin{equation}
 r\equiv X,
 \qquad
 \phi\equiv\sqrt r\,T.
\end{equation}
Then
\begin{align}
 \eta={}&dY^2+dZ^2
 +\left(1-\frac{\phi^2}{4r^3}\right)dr^2
 +\frac{\phi}{r^2}\,dr\,d\phi
 -\frac1r\,d\phi^2,
 \label{eq:cusp-metric}
\end{align}
while the visible coordinate width obeys
\begin{equation}\label{eq:cusp-width}
 |\phi|<r^{\beta+1/2}=o(\sqrt r).
\end{equation}
The spacetime $(\Omega_\beta,\eta)$ is smoothly extendible in the
original Cartesian coordinates.  Here $\phi$ is not periodic, and
$\partial_\phi$ is not a Kerr-type axial Killing generator.
\end{proposition}

\begin{proof}
Substituting $dT=r^{-1/2}d\phi-\phi\,dr/(2r^{3/2})$ and
$dX=dr$ gives \eqref{eq:cusp-metric}, and the defining inequality
of the cusp gives \eqref{eq:cusp-width}.  The metric extends in
Cartesian coordinates.  In the singular coordinates its coefficients
depend on $\phi$, so this coordinate is neither periodic nor a
Killing coordinate.
\end{proof}

\begin{corollary}[Failure of a H\"older width inference]
\label{cor:no-automatic-thickness}
Fix any $0<\alpha\le1$.  The following data do not imply an angular
lower width exceeding $\sqrt r$: openness of an inverse-image region,
accumulation at $r=0$, finite timelike diameter of an extension-chart
neighborhood, local $C^{0,\alpha}$ regularity of the ambient metric, and
a timelike coefficient satisfying $-g_{\phi\phi}\asymp r^{-1}$ in the
singular coordinates.  The assertion remains false even when the
ambient extension metric is smooth.  Consequently there is no H\"older
exponent at which these data alone yield the supercritical-width bound.
Such a width estimate must therefore come from additional geometry of
the Kerr extension.
\end{corollary}

\begin{proof}
In a Cartesian box $W_\delta$ about the origin, every future
timelike curve has $L_\eta\le\Delta T\le2\delta$.
For small $r>0$, the section $Y=Z=0$ of
$W_\delta\cap\Omega_\beta$ has the full width
$A(r)\equiv2r^{\beta+1/2}=o(\sqrt r)$.
Proposition~\ref{prop:flat-cusp} supplies all the other stated
properties, including smooth ambient regularity and
$g_{\phi\phi}=-1/r$, for every $\alpha\in(0,1]$.
\end{proof}

\subsection{Kerr generators at an extension boundary}

We now consider continuous extensions of open Kerr subdomains.
Their embeddings are not required to extend to the full realization.

\begin{definition}[Continuous extension of an interior vector field]\label{def:vector-extension}
Let $(\widetilde M,\widetilde g,\iota)$ be a $C^0$ extension of $(N,h)$ and let $X$ be a smooth vector field on $N$.  We say that $X$ extends continuously at $q\in\partial\iota(N)$ if there are a neighborhood $W\ni q$ and a continuous vector field $\widetilde X$ on $W$ such that
\begin{equation}
 \widetilde X=\iota_*X
 \qquad\text{on }W\cap\iota(N).
\end{equation}
\end{definition}

A $C^1$ local extension of the corresponding one-parameter action
has a continuously extended generator.  Definition~\ref{def:vector-extension}
requires only the latter property.

\begin{theorem}[Killing combinations at accessible boundary points]\label{thm:symmetry-breaking}
Let $U$ be an open Kerr subdomain containing a terminal tail of a
unit-speed timelike geodesic $\gamma:[0,\Lambda)\to\KerrSpace$,
$\Lambda<\infty$, which is inextendible in $\KerrSpace$ at $s=\Lambda$.
Let $(\widetilde M,\widetilde g,\iota)$ be a $C^0$ extension of
$(U,g_{M,a}|_U)$, and suppose $q\in\partial\iota(U)$ satisfies
\begin{equation}
 \iota(\gamma(s))\longrightarrow q.
\end{equation}
For every $\alpha,\beta\in\mathbb R$ with $\alpha\ne a\beta$, the
field $\alpha T+\beta\Phi$ does not extend continuously at $q$ in the
sense of Definition~\ref{def:vector-extension}.  In particular, neither
$T$ nor $\Phi$ extends continuously at $q$.
\end{theorem}

\begin{proof}
Theorem~\ref{thm:terminal-geodesics} gives $\theta=\pi/2$ and
$r\to0^+$ along the tail.  For $X\equiv\alpha T+\beta\Phi$,
\begin{equation}
 g(X,X)=-\alpha^2+\beta^2(r^2+a^2)
 +\frac{2M}{r}(\alpha-a\beta)^2\longrightarrow+\infty.
\end{equation}
If a continuous extension $\widetilde X$ existed, then
$\widetilde g(\widetilde X,\widetilde X)$ would be finite and
continuous at $q$.  This contradicts the displayed divergence.
\end{proof}

\begin{remark}[The exceptional Killing direction]
The proof uses no curvature bound and does not require the isometry
group to extend.  Nonsymmetric embeddings and the exceptional field
$aT+\Phi$, whose terminal equatorial norm is $r^2$, are not excluded
by this argument.
\end{remark}

\section{Conclusions and Future Work}\label{sec:conclusion}

The ring obstruction is determined by the causal structure of Kerr.
A closed timelike detour can be inserted and repeated in a
fixed-endpoint causal homotopy class. It has positive length for every
positive continuous conformal factor, also on covers where the factor
does not descend to the base. Confinement in an extension chart
contradicts the resulting unbounded lengths. For the exact full
realizations considered here, geodesic continuation reduces every
possible continuous extension to an extension at the ring. This
proves $C^0$-inextendibility in all three rotating regimes and on
every connected cover.

The classification of terminal geodesics and the closest-approach,
curvature, and polar Jacobi calculations describe the geometry
associated with this causal obstruction.

Continuous stationary and axisymmetric perturbations preserve the ring
obstruction under the small $O(\sqrt\Sigma)$ bound stated in the
paper.  For the stronger assertion of infinite length with every
conformal factor, the sharp strict coefficient is $1/(1+a^2)$ at
order $\Sigma$ in the specified coordinate norm.  The pointwise
critical bound suffices, but equality of the limiting ratio does not.
At the next order, the strict threshold is $M^2/(2a^4)$.  At equality,
the obstruction persists for $a^2\ge2M^2$; for $a^2<2M^2$, the
rank-one example has finite conformal lengths in every causal homotopy
class with fixed endpoints.  Cone widening preserves the obstruction
along Kerr-causal approaches, including null curves, without symmetry
assumptions.

A global inextendibility theorem for a perturbed metric requires
control of its other ends, including other possible conformal
boundaries. The finite-length examples here do not provide geometric
extensions. Determining the optimal
value of the sufficient perturbation constant $\eta_{M,a}$ is a
further problem. A holonomy argument based on the tidal curvature
estimates would require controlled shrinking loops and uniform
remainders in an arbitrary extension chart. Extensions to dynamical
vacuum perturbations, and their implications for strong cosmic
censorship, require further study of the perturbed geometry.

\appendix

\section{Carter realizations and graph classification}
\label{app:explicit-carter-atlas}

Using the crossing maps of Section~\ref{sec:regular-horizons}, we
construct an explicit member of $\CarterClass$ and classify the
realizations obtained by gluing full collars.  All angular variables
take values in $\mathbb S^2$.

The allowed Carter
gluings have line or alternating-cycle incidence cores, explicit
holonomy, fundamental groups, and integral homology.

\subsection{The countable graph and the quotient manifold}

Take crossing vertices $C_n$,
$n\in\mathbb Z$, with
\begin{equation}
 h_n\equiv\begin{cases}r_+,&n\ \text{even},\\ r_-,&n\ \text{odd}.
 \end{cases}
\end{equation}
Place one copy $X_n$ of $\mathcal X_{h_n}$ at $C_n$.  For each
$n\in\mathbb Z$, take a copy $W_n$ of the intermediate block
$\mathcal B_2$, joining $C_n$ to $C_{n+1}$.  Attach two copies
$P_n^+$ and $P_n^-$ of $\mathcal B_3$ at every even vertex, and two
copies $N_n^+$ and $N_n^-$ of $\mathcal B_1$ at every odd vertex.  The
four quadrants are assigned as follows:
\begin{equation}\label{eq:app-quadrant-table}
\begin{array}{c|cccc}
 &Q_{++}&Q_{--}&Q_{+-}&Q_{-+}\\ \hline
 n\ \mathrm{even}&P_n^+&P_n^-&W_n&W_{n-1}\\
 n\ \mathrm{odd}&W_{n-1}&W_n&N_n^+&N_n^-
\end{array}
\end{equation}
These quadrant assignments respect
$\operatorname{sgn}(UV)=\operatorname{sgn}(r-h_n)$.

Let $\mathscr Y$ be the topological disjoint union of the blocks and
all the $X_n$.  For each entry in \eqref{eq:app-quadrant-table}, set
$\boldsymbol\lambda_e\equiv(0,0)\in\mathsf A$ and identify the corresponding
block collar with the indicated quadrant through
$\Xi_{h_n}^{\epsilon\eta}$.  The two collars of each $W_n$ are disjoint by
\eqref{eq:app-disjoint-collars}, so each point lies in at most one
gluing domain.  Let
\begin{equation}\label{eq:app-quotient}
 \mathscr K_{M,a}^{\mathrm{lin}}\equiv\mathscr Y/\!\sim,
 \qquad \pi:\mathscr Y\longrightarrow
 \mathscr K_{M,a}^{\mathrm{lin}}
\end{equation}
be the resulting quotient.

\begin{proposition}[Properties of the explicit realization]
\label{prop:explicit-carter-realization}
The space $\mathscr K_{M,a}^{\mathrm{lin}}$ is a connected, Hausdorff,
second-countable, time-oriented real-analytic four-manifold.  The block
metrics and the metrics $g_{h_n}$ determine an analytic Lorentzian
metric $g_{M,a}^{\mathrm{lin}}$.  The functions $r$ and $\mu$, the
fields $T$ and $\Phi$, and the tensor $\mathsf K$ descend analytically.
With the images of the blocks and crossing models as its marked atlas,
the resulting tuple is a member of $\CarterClass$.  In particular,
$\CarterClass\ne\varnothing$.
\end{proposition}

\begin{proof}
Each equivalence class has at most two representatives.  The relation
is closed: in a fixed block-crossing pair, a limit on $UV=0$ would
force $r\to h_n$, which has no limit in the open block; a limit in a
quadrant remains in the transition graph by
\eqref{eq:app-sector-inverse}.  The outer collar boundary
$|UV|=\delta_{h_n}$ is absent from the crossing model.  Since each summand is open and closed, every convergent net is
covered by these cases.
The saturation of an open set is open under the gluing homeomorphisms.
Thus $\pi$ and $\pi\times\pi$ are open quotient maps, and
\begin{equation}
 (\pi\times\pi)^{-1}
 \bigl(\operatorname{diag}(\mathscr K_{M,a}^{\mathrm{lin}})\bigr)
 =\{(y,z):y\sim z\}
\end{equation}
is closed.  The diagonal of the quotient is therefore closed, and the quotient
is Hausdorff.  In particular, the construction has no doubled horizon
points.

Each model maps injectively and openly to the quotient.  Its analytic
charts, including the polar caps, therefore give a four-manifold
atlas with transitions
\eqref{eq:app-sector-map}-\eqref{eq:app-sector-inverse}.
Lemma~\ref{lem:app-crossing-metric} and
\eqref{eq:app-killing-fields}, \eqref{eq:app-ky-crossing} show that the
metrics, functions, and tensors agree on overlaps.  Since $\pi$ is open, a countable base of the disjoint union
descends to a countable base of the quotient.  The incidence graph is a bi-infinite chain with
two leaves at each crossing.  Each model and each overlap is connected.  To see this for
$\mathcal B_1$, one may join points while avoiding the deleted ring
by moving first to $r\ne0$ and then on the sphere.  It follows from
the connectedness of the incidence graph that the quotient is
connected.

On each $\Delta>0$ block the field
\begin{equation}
 Z\equiv A(r)T+a\Phi,
 \qquad g(Z,Z)=-\Delta\Sigma<0,
\end{equation}
is timelike; on an intermediate block use $\partial_r$, whose norm is
$\Sigma/\Delta<0$.  Each crossing model is simply connected, since
its $(U,V)$ domain is star-shaped and its angular factor is
$\mathbb S^2$, so its time-orientation double cover is trivial.
An orientation chosen on $X_0$ propagates uniquely through the
connected overlaps of the incidence tree and defines a global time
orientation.

Removing the crossing axes leaves exactly the four quadrants already
identified with full block collars.  Hence the components of
$\{\Delta\ne0\}$ are the marked blocks.  Both axes, all four open
horizon branches, and the bifurcation sphere occur in each crossing;
\eqref{eq:app-ingoing-from-kruskal}-\eqref{eq:app-outgoing-from-kruskal}
give the full ingoing and outgoing charts.

A nonhorizon point has a neighborhood meeting only its block and,
within a collar, the adjacent crossing.  A horizon point has a
crossing neighborhood meeting only that crossing and its four
incident block closures.  The disjoint collar intervals prevent this
neighborhood from reaching another crossing through an intermediate
block.  These neighborhoods prove local finiteness.  Taking a finite
subcover of a compact set shows that it meets only finitely many
pieces.  Distinct block
closures meet only on the axes of a common crossing.

Finally, $r_->0$ and \eqref{eq:app-disjoint-collars} keep every gluing
away from $r=0$.  Each negative block retains the distinct analytic
disc branches
\begin{equation}
 D_{n,\mathrm{up}}^\pm\equiv\{r=0,\ \mu>0\},\qquad
 D_{n,\mathrm{down}}^\pm\equiv\{r=0,\ \mu<0\},
\end{equation}
each diffeomorphic to $\mathbb R_t$ times an open two-disc.  Only the
equatorial ring is removed; no points on either branch are identified.
With $\boldsymbol\lambda_e\equiv(0,0)$, all conditions
\textup{(C1)}-\textup{(C6)} follow.
\end{proof}

\subsection{Graph realizations, coverings, and cyclic holonomy}
\label{app:graph-realizations}

The construction above extends to all the full-collar gluings in
Definition~\ref{def:carter-spacetime}.  Their classification is
described by the incidence graph and the stationary-axial gluing
parameters.

Let $\mathscr C_+$ and $\mathscr C_-$ be disjoint finite or countable sets
of crossing vertices, carrying copies of $\mathcal X_{r_+}$ and
$\mathcal X_{r_-}$, respectively.  Let $\mathscr B_j$ be a finite or
countable set of copies of $\mathcal B_j$.  Pair the block endpoints
with crossing quadrants through the four bijections
\begin{align}
 \mathscr B_3\times\{r_+\}
 &\longleftrightarrow
 \mathscr C_+\times\{Q_{++},Q_{--}\},
 \label{eq:graph-pairing-outer-positive}\\
 \mathscr B_2\times\{r_+\}
 &\longleftrightarrow
 \mathscr C_+\times\{Q_{+-},Q_{-+}\},
 \label{eq:graph-pairing-outer-negative}\\
 \mathscr B_2\times\{r_-\}
 &\longleftrightarrow
 \mathscr C_-\times\{Q_{++},Q_{--}\},
 \label{eq:graph-pairing-inner-positive}\\
 \mathscr B_1\times\{r_-\}
 &\longleftrightarrow
 \mathscr C_-\times\{Q_{+-},Q_{-+}\}.
 \label{eq:graph-pairing-inner-negative}
\end{align}
Assign $\boldsymbol\lambda_e\in\mathsf A$ to each paired endpoint.
The blocks and crossings form a connected incidence multigraph
$\mathscr G$, with one edge per pairing; parallel intermediate blocks
are allowed.

Let $\mathscr Y_{\mathscr G}$ be the disjoint union of the indicated
blocks and crossing models.  When $e$ pairs a block collar with the
quadrant $Q_{\epsilon\eta}$ at a crossing of type $h$, impose the
identification
\begin{equation}\label{eq:graph-gluing-map}
 x\sim
 \widehat{\mathscr S}_{h,\boldsymbol\lambda_e}
 \bigl(\Xi_h^{\epsilon\eta}(x)\bigr).
\end{equation}
Write the resulting quotient as
\begin{equation}
 \mathscr K(\mathscr G,\boldsymbol\lambda)
 \equiv\mathscr Y_{\mathscr G}/\!\sim.
\end{equation}

\begin{theorem}[Classification of Carter graph realizations]
\label{thm:graph-realizations}
Every connected gluing datum
$(\mathscr G,\boldsymbol\lambda)$ described above has the following
properties.  Conversely, every member of $\CarterClass$ is
marked-isometric to the quotient determined by its incidence graph and
the incidence, exact-overlap, and local-finiteness data specified in
conditions \textup{(C4)}-\textup{(C6)}.
\begin{enumerate}[label=\textup{(\roman*)}]
 \item The quotient $\mathscr K(\mathscr G,\boldsymbol\lambda)$ is a
 connected, Hausdorff, second-countable, time-oriented real-analytic
 four-manifold.  The local Kerr metrics and the functions and tensors
 $r,\mu,T,\Phi,\mathsf K$ descend analytically.  The marked quotient is
 a member of $\CarterClass$.
 \item Remove the vertices in $\mathscr B_1\cup\mathscr B_3$ and their
 incident edges.  The remaining core $\mathscr G_0$ is either a
 bi-infinite line or a finite cycle.  Along it the $r_+$ and $r_-$
 crossing vertices alternate.  Consequently every finite core cycle
 has length divisible by four when both block and crossing vertices are
 counted.
 \item If $\mathscr G_0$ is a finite cycle, the graph universal cover
 $\widetilde{\mathscr G}$ is the bi-infinite line with the two prescribed
 leaves at every crossing.  Pulling back the gluing parameters gives an
 analytic covering
 \begin{equation}\label{eq:graph-covering-map}
  \mathscr K(\widetilde{\mathscr G},
   \widetilde{\boldsymbol\lambda})
  \longrightarrow
  \mathscr K(\mathscr G,\boldsymbol\lambda)
 \end{equation}
 whose deck group is infinite cyclic.  The deck action is free and
 properly discontinuous and consists of analytic Kerr isometries.  In
 particular, the realization with cyclic core is the quotient of its
 line realization by this action.
 \item For a fixed port-labelled graph, changes of stationary-axial
 coordinates on the individual pieces replace the edge parameters by
 \begin{equation}\label{eq:graph-gauge-change}
  \boldsymbol\lambda_e'=
  \boldsymbol b_{c(e)}+\boldsymbol\lambda_e-
  \boldsymbol b_{b(e)},
 \end{equation}
 where $b(e)$ and $c(e)$ are the block and crossing endpoints of $e$.
 If the core is a line, all parameters can be made zero.  If the core is
 a finite cycle, all but one can be made zero; the remaining invariant is
 the stationary-axial holonomy
 \begin{equation}\label{eq:graph-holonomy}
  \boldsymbol H\equiv\sum_{e\subset\mathscr G_0}
  \sigma_e\boldsymbol\lambda_e\in\mathsf A,
 \end{equation}
 where $\sigma_e=1$ when the chosen circuit traverses $e$ from its block
 endpoint to its crossing endpoint and $\sigma_e=-1$ otherwise.
 Reversing the circuit replaces $\boldsymbol H$ by
 $-\boldsymbol H$.
 \item The line realization
 $\mathscr K_{M,a}^{\mathrm{lin}}$ constructed in
 Proposition~\ref{prop:explicit-carter-realization} is a universal
 graph model in the following precise sense.  Every realization whose
 core is a line is marked-isometric to it after applying the quadrant
 involutions \eqref{eq:graph-quadrant-involution}, stationary-axial
 changes of coordinates, and a relabelling of identical block copies.
 If a finite core contains $2m$ crossing vertices, its realization is
 marked-isometric to a screw quotient
 \begin{equation}\label{eq:graph-screw-quotient}
  \mathscr K_{M,a}^{\mathrm{lin}}/
  \langle\mathscr D_{m,\boldsymbol H}\rangle,
  \qquad
  \mathscr D_{m,\boldsymbol H}
  \equiv\mathscr S_{-\boldsymbol H}\circ\mathscr T_{2m}.
 \end{equation}
 Here $\mathscr T_{2m}$ shifts the indices of all pieces by $2m$, and
 $\mathscr S_{\boldsymbol H}$ is the global stationary-axial action
 given on each piece by
 \eqref{eq:graph-block-action} or
 \eqref{eq:graph-crossing-action}.  The generator in
 \eqref{eq:graph-screw-quotient} preserves a global time orientation
 and acts freely and properly discontinuously.
\end{enumerate}
\end{theorem}

\begin{proof}
The stationary-axial actions preserve $UV$, $r$, and the Kerr
tensors.  The construction in Proposition~\ref{prop:explicit-carter-realization}
therefore applies.  Since the collars are disjoint, each equivalence
class has at most two representatives.  The transition relation is
closed and the quotient map is open.  Thus each model embeds, the
quotient is Hausdorff, and the analytic objects descend.  The countable
collection of pieces gives a countable atlas.  Since the models and
overlaps are connected and $\mathscr G$ is connected, the quotient
is connected.  The block and crossing neighborhoods used above meet
only finitely many marked pieces, proving local finiteness and the
assertion for compact sets.

To establish time orientability, we treat the line and cyclic cases
together.  On a
block with $\Delta>0$, use the timelike field
\begin{equation}
 Z\equiv A(r)T+a\Phi,
 \qquad g(Z,Z)=-\Delta\Sigma,
\end{equation}
as the reference orientation.  On an intermediate block, use
$\partial_r$, for which $g(\partial_r,\partial_r)=\Sigma/\Delta<0$.
Choose on $\mathcal X_{r_+}$ the time cone containing
$-\partial_U+\partial_V$ at its bifurcation sphere, and on
$\mathcal X_{r_-}$ the cone containing
$\partial_U+\partial_V$.  Their timelike character follows from
\eqref{eq:app-bifurcation-metric}.  The orientations extend over the
simply connected crossing models.
The adjacent block references induce the signs
\begin{equation}\label{eq:graph-time-signs}
\begin{array}{c|cccc}
 &Q_{++}&Q_{--}&Q_{+-}&Q_{-+}\\ \hline
 r_+&+&-&+&-\\
 r_-&+&-&+&-
\end{array}
\end{equation}
where at $r_+$ the first two entries refer to $Z$ and the last two to
$\partial_r$, while at $r_-$ the first two refer to $\partial_r$ and
the last two to $Z$.  Indeed, on each quadrant,
\begin{align}
 (\Xi_h^{\epsilon\eta})_*Z
 &=-\kappa_hA(r)U\partial_U
   +\kappa_hA(r)V\partial_V
   +(a-A(r)\Omega_h)\Phi,\label{eq:graph-Z-pushforward}\\
 (\Xi_h^{\epsilon\eta})_*\partial_r
 &=\frac{\kappa_hA(r)}{\Delta(r)}
   (U\partial_U+V\partial_V)+\chi_h'(r)\Phi.
 \label{eq:graph-r-pushforward}
\end{align}
The angular coefficients stay bounded at $r=h$.  The signs of
$\kappa_h$ and $\Delta$ give, up to positive factors, the double-null
components near $r_+$
\begin{equation}
 Z\sim(-U,V),\qquad
 \partial_r\sim(-U,-V),
\end{equation}
on the exterior and intermediate sides, respectively.  Near $r_-$ one
has
\begin{equation}
 \partial_r\sim(U,V),\qquad
 Z\sim(U,-V)
\end{equation}
on the intermediate and negative sides.  This proves \eqref{eq:graph-time-signs}.  The signs are constant
on each connected overlap and are preserved by the connected
stationary-axial group.

Orientations propagate consistently precisely when each circuit has
comparison-sign product $+1$.  Every circuit lies in the core.
Each crossing contributes $-1$ from its two core quadrants, and the
alternating crossing types give an even number of factors.  The product is therefore $+1$.  In the line case there is no
circuit.

Every quadrant and finite endpoint occurs once in the four pairings,
so all horizon branches and bifurcation spheres are present.  The
collars lie in $r>0$, leaving both regular disc branches distinct and
removing precisely the rings.  This proves \textup{(i)}.

Conversely, the marked maps of a given realization induce
\begin{equation}
 \mathscr K(\mathscr G,\boldsymbol\lambda)
 \longrightarrow\KerrSpace
\end{equation}
which is bijective by the exact-overlap condition.  It is open
because its restriction to each piece is an open embedding.  With the
prescribed atlas, it is a marked analytic isometry.

Deleting the degree-one exterior and negative vertices leaves a
connected core in which every vertex has degree two.  Following an edge in both directions gives either a finite cycle or
a bi-infinite line.  Connectedness shows that this exhausts the core.  Each intermediate block
joins an $r_+$ crossing to an $r_-$ crossing, so the types alternate.
This proves \textup{(ii)}.

Suppose now that the core is a finite cycle.  Lift the incidence graph,
the port labels, and the edge parameters to its graph universal cover.
The local identity maps from the lifted pieces to their corresponding
base pieces commute with every gluing map and therefore induce
\eqref{eq:graph-covering-map}.  A small block chart away from collars lifts to separate block
copies.  A sector chart lifts with the same two-piece transition
\eqref{eq:graph-gluing-map}.  At a horizon or bifurcation point, a
small crossing chart, with polar caps when necessary, lifts to
separate crossing copies containing all incident sectors.  These evenly covered neighborhoods prove that the induced map is
a covering.  The graph deck generator is an analytic isometry and
fixes no piece.  Since a compact set meets finitely many pieces, only
finitely many of its translates can meet that set.  The action is
therefore free and properly discontinuous, and its quotient is the
original realization.  This proves \textup{(iii)}.  The covering
four-manifold itself need not be simply connected.

The equivariance of the gluing maps under vertex translations gives
\eqref{eq:graph-gauge-change}.  Fix one vertex parameter and solve
successively along a spanning tree.  This makes every tree-edge
parameter zero.  For a line, every edge belongs to the tree.  For a
cycle, one edge remains, with parameter given by the signed sum in
\eqref{eq:graph-holonomy}; the vertex parameters cancel in this
sum.
Lemma~\ref{lem:marked-isometry-rigidity} gives the converse after the
discrete relabellings are fixed.  This proves \textup{(iv)}.

To prove (v), consider the map
\begin{equation}\label{eq:graph-quadrant-involution}
 \mathscr J_h(U,V,\omega)\equiv(-U,-V,\omega)
\end{equation}
This is an analytic isometry of $\mathcal X_h$, as follows either from
\eqref{eq:app-analytic-crossing-metric} or from
\begin{equation}
 \mathscr J_h\circ\Xi_h^{\epsilon\eta}
 =\Xi_h^{-\epsilon,-\eta}.
\end{equation}
The isometry preserves $r,\mu,T,\Phi$, and $\mathsf K$; it exchanges
the two core ports at a crossing and, at the same time, the two
identical leaf blocks.  Choose an orientation of the core.  At a
crossing $c$, assign to each incident core edge a quadrant sign
$s_{c,e}\in\{1,-1\}$ according to the convention
\begin{equation}
 \begin{array}{c|cc}
  &s_{c,e}=+1&s_{c,e}=-1\\ \hline
  h_c=r_+&Q_{+-}&Q_{-+}\\
  h_c=r_-&Q_{--}&Q_{++}.
 \end{array}
 \label{eq:graph-core-port-signs}
\end{equation}
Under this convention, the two endpoints of every intermediate block
in \eqref{eq:app-quadrant-table} have equal signs.  Applying
$\mathscr J_h$ at $c$ reverses both of them.  We seek signs
$x_c\in\{1,-1\}$ such that the transformed endpoint signs agree on
every intermediate block:
\begin{equation}\label{eq:graph-port-normalization}
 x_cs_{c,e}=x_{c'}s_{c',e}.
\end{equation}
On a line, these equations are solved recursively.  For a finite
cycle, compatibility requires the product of all endpoint signs around
the circuit to be $+1$.  The two signs at each crossing have product
$-1$, while the number of crossings is even, so this condition is
satisfied.  Thus \eqref{eq:graph-port-normalization} also has a solution
in the cyclic case.

After applying the corresponding involutions and relabelling the
identical blocks, one recovers exactly the quadrant table
\eqref{eq:app-quadrant-table}.  Part (iv) removes every continuous edge
parameter on the line, proving the first assertion of (v).  In the
cyclic case, pass to the line and remove all lifted parameters.  Choose
$\mathscr T_{2m}$ to follow the circuit used in
\eqref{eq:graph-holonomy}; let $\delta$ be the corresponding graph deck
shift, and let $\boldsymbol b_v$ denote the lifted vertex gauge.  Along
an edge traversed from its block endpoint to its crossing endpoint,
\eqref{eq:graph-gauge-change}, with zero transformed parameter, changes
$\boldsymbol b$ by $-\boldsymbol\lambda_e$; traversing the edge in the
opposite direction changes it by the opposite amount.  Hence, for every
lifted vertex $v$,
\begin{equation}
 \boldsymbol b_{\delta v}-\boldsymbol b_v
 =-\sum_{e\subset\mathscr G_0}
   \sigma_e\boldsymbol\lambda_e
 =-\boldsymbol H.
\end{equation}
The graph deck generator is therefore conjugate to the index shift
$\mathscr T_{2m}$ followed by the global symmetry
$\mathscr S_{-\boldsymbol H}$, which gives
\eqref{eq:graph-screw-quotient}.  The two factors commute by
\eqref{eq:graph-equivariance}.  The nontrivial index shift gives
freeness and proper discontinuity, while the sign calculation
\eqref{eq:graph-time-signs} shows that a time orientation is preserved.
This proves (v).
\end{proof}

\begin{corollary}[Inextendibility of graph realizations]
\label{cor:graph-inextendibility}
Every spacetime $\mathscr K(\mathscr G,\boldsymbol\lambda)$ in
Theorem~\ref{thm:graph-realizations} is
$C^0$-inextendible.  Every finite-proper-time terminal
timelike geodesic in it has the endpoint behavior and sharp asymptotics
proved in Theorems~\ref{thm:terminal-geodesics} and
\ref{thm:sharp-terminal-asymptotics}.
\end{corollary}

\begin{proof}
By part (i) of Theorem~\ref{thm:graph-realizations}, the marked quotient
belongs to $\CarterClass$.  The stated conclusions follow from
Theorems~\ref{thm:main-c0}, \ref{thm:terminal-geodesics}, and
\ref{thm:sharp-terminal-asymptotics}.
\end{proof}

\begin{remark}
The geodesic and causal homotopy arguments apply directly to these
quotients; no descent principle for inextendibility is needed.
\end{remark}

\begin{corollary}[Fundamental groups of Carter realizations]
\label{cor:graph-fundamental-group}
Let $\mathscr K\equiv\mathscr K(\mathscr G,\boldsymbol\lambda)$ be a
realization in Theorem~\ref{thm:graph-realizations}.  Its fundamental
group is free, with one generator for every negative block and one
additional generator when the core is cyclic.  In particular,
\begin{equation}
 \pi_1(\mathscr K)\cong
 \begin{cases}
  F_{\aleph_0},&\text{if the core is a bi-infinite line},\\
  F_{2m+1},&\text{if the core has $2m$ crossing vertices},
 \end{cases}
\end{equation}
where $F_I$ denotes the free group on a set of cardinality $I$.
These groups are independent of the stationary-axial holonomy.
Thus cyclic realizations with different numbers of crossing vertices
are not homeomorphic.
\end{corollary}

\begin{proof}
Every crossing model $\mathcal X_h$ is simply connected: its
$(U,V)$ domain is star-shaped and its angular factor is $\mathbb S^2$.
The blocks $\mathcal B_2$ and $\mathcal B_3$ are simply connected as
well.  Each block-crossing intersection is a full collar,
diffeomorphic to $\mathbb R\times J\times\mathbb S^2$ for an open
interval $J$, and is therefore connected and simply connected.
Distinct block images are disjoint, as are distinct crossing images;
there are no triple intersections.

We compute the group of $\mathcal B_1$ explicitly.  The radial
diffeomorphism
\begin{equation}
 (r,\omega)\longmapsto\frac{r_-}{r_--r}\,\omega
\end{equation}
maps $I_1\times\mathbb S^2$ onto $\mathbb R^3\setminus\{0\}$ and
maps the removed equatorial circle to the standard unit circle $C$.
Consequently
\begin{equation}
 \mathcal B_1\cong
 \mathbb R\times\bigl(\mathbb R^3\setminus(C\cup\{0\})\bigr).
\end{equation}
The complement of $C$ in $\mathbb S^3$ is an open solid torus.
Removing a point from a three-manifold preserves its fundamental
group, by van Kampen with a coordinate ball and its simply connected
punctured ball.  Removing first $\infty$ and then $0$ therefore gives
$\pi_1(\mathcal B_1)\cong\mathbb Z$.

Van Kampen along a spanning tree gives the free product of the
negative-block groups; each remaining edge adds one generator:
\begin{equation}
 \pi_1(\mathscr K)\cong
 \Bigl(\mathop{\ast}_{\mathscr B_1}\mathbb Z\Bigr)
 *\pi_1(\mathscr G).
\end{equation}
Local finiteness places every loop and homotopy in finitely many
pieces, so the countable case follows by direct limit.  A cyclic core
has $m$ inner crossings, $2m$ negative blocks, and one cycle; a line
has countably many negative blocks and no cycle.  These counts give the asserted ranks, independently of the collar
maps.  Their abelianizations have different ranks in the finite cases,
so the corresponding realizations cannot be homeomorphic.
\end{proof}

\begin{corollary}[Homology of Carter realizations]
\label{cor:graph-homology}
For every realization $\mathscr K\equiv\mathscr K(\mathscr G,\boldsymbol\lambda)$,
\begin{equation}
 \begin{aligned}
 H_2(\mathscr K;\mathbb Z)
 &\cong\mathbb Z\oplus\bigoplus_{\mathscr B_1}\mathbb Z,\\
 H_3(\mathscr K;\mathbb Z)&\cong H_1(\mathscr G;\mathbb Z),\qquad
 H_k(\mathscr K;\mathbb Z)=0\quad(k\ge4).
 \end{aligned}
\end{equation}
In particular, a cyclic core with $2m$ crossing vertices gives
$H_2(\mathscr K;\mathbb Z)\cong\mathbb Z^{2m+1}$ and
$H_3(\mathscr K;\mathbb Z)\cong\mathbb Z$, whereas a line core gives
$H_2(\mathscr K;\mathbb Z)\cong\bigoplus_{\mathbb N}\mathbb Z$
and $H_3(\mathscr K;\mathbb Z)=0$.
All integral homology groups are free abelian and independent of the
stationary-axial holonomy.
\end{corollary}

\begin{proof}
The homeomorphism used in
Corollary~\ref{cor:graph-fundamental-group} identifies a negative block
with $\mathbb R\times Y$, where
\begin{equation}
 Y\equiv\mathbb R^3\setminus(C\cup\{0\})
  =\mathbb S^3\setminus(C\cup\{0,\infty\})
\end{equation}
and $C$ is the unit equatorial circle.
Alexander duality \cite[Section~3.3]{Hatcher2002} gives $H_1(Y;\mathbb Z)\cong\mathbb Z$,
$H_2(Y;\mathbb Z)\cong\mathbb Z^2$, and $H_k(Y;\mathbb Z)=0$
for $k\ge3$.  To identify the collar map, represent
\begin{equation}
 \widetilde H^0(C\cup\{0,\infty\};\mathbb Z)
 \cong\mathbb Z^2
\end{equation}
by the values on $C$ and $\{0\}$ after subtracting the value at
$\infty$.  With outward orientations, the duality isomorphism sends
an enclosing sphere to its winding numbers about these two components.
A sphere $\{r=r_0\}\times\mathbb S^2$ in the positive gluing collar
has Euclidean radius $r_-/(r_--r_0)>1$, and its class therefore maps
to $(1,1)$.  A small sphere about the origin maps to $(0,1)$.
These form an integral basis.  Write the corresponding classes as
$u_b,v_b$, with $u_b$ the collar class.

Let $U$ be the union of all block images and $V$ the union of all
crossing images.  They form an open cover of $\mathscr K$.
The connected components of $U\cap V$ are precisely the full collars,
one for each edge of $\mathscr G$.
Every collar, every crossing model, and every block of type
$\mathcal B_2$ or $\mathcal B_3$ is homotopy equivalent to
$\mathbb S^2$.  The angular sphere defines its $H_2$ generator.
For a negative block use the basis $u_b,v_b$ just constructed.
On every sphere with $t$ and $r$ fixed, the angular part of each
gluing map is an axial rotation, which has degree $+1$.
Thus, with each incidence edge oriented from
its block endpoint to its crossing endpoint, the Mayer-Vietoris map
on $H_2$ is
\begin{equation}
 D:\bigoplus_{e\in E(\mathscr G)}\mathbb Z
   \longrightarrow
   \left(\bigoplus_{v\in V(\mathscr G)}\mathbb Z u_v\right)
   \oplus\left(\bigoplus_{b\in\mathscr B_1}\mathbb Z v_b\right),
 \qquad
 D(e)\equiv u_{b(e)}-u_{c(e)}.
\end{equation}
Its first component is the incidence differential of the graph,
up to the choice of edge orientation, and its second component is
zero.  Hence
\begin{equation}
 \ker D\cong H_1(\mathscr G;\mathbb Z),
 \qquad
 \operatorname{coker}D
 \cong H_0(\mathscr G;\mathbb Z)
       \oplus\bigoplus_{\mathscr B_1}\mathbb Z.
\end{equation}
The higher homology of the pieces vanishes and
$H_1(U\cap V;\mathbb Z)=0$.  The Mayer-Vietoris sequence \cite[Section~2.2]{Hatcher2002} therefore
identifies $H_3(\mathscr K;\mathbb Z)$ with $\ker D$ and
$H_2(\mathscr K;\mathbb Z)$ with $\operatorname{coker}D$, and gives
the asserted vanishing in degrees at least four.
Connectedness gives $H_0(\mathscr G;\mathbb Z)=\mathbb Z$;
$H_1(\mathscr K;\mathbb Z)$ is the free abelian group obtained by
abelianizing Corollary~\ref{cor:graph-fundamental-group}.
All groups in the argument are direct sums.  Since singular chains
have finite support, the same argument applies to a countable cover.
Finally, the collar maps above have degree $+1$ for every choice of
$\boldsymbol\lambda$, so neither torsion nor dependence on the
stationary-axial holonomy occurs.
\end{proof}

\section{Scalar asymptotics at a simple radial zero}\label{app:scalar-normal-form}

The exponent $2/5$ in the Kerr asymptotics comes from a simple
radial zero and a quadratically vanishing coefficient of $\dot r$.
We give the scalar calculation, including the two correction terms
used in the main text.

\begin{lemma}[Asymptotics at a simple radial zero]
\label{lem:simple-zero-terminal-law}
Let $H,R:(0,\varepsilon)\to(0,\infty)$ satisfy, as $r\downarrow0$,
\begin{align}
 H(r)
 &=h_0r^q\left(1+h_1r+h_2r^2+O(r^3)\right),
 \label{eq:simple-H}\\
 R(r)
 &=\kappa_0r\left(1+\kappa_1r+\kappa_2r^2+O(r^3)\right),
 \label{eq:simple-R}
\end{align}
where
\begin{equation}
 h_0>0,\qquad \kappa_0>0,\qquad p\equiv q+\frac12>0.
\end{equation}
Let $r:[s_0,\Lambda)\to(0,\varepsilon)$ be $C^1$, strictly
decreasing, and satisfy
\begin{equation}\label{eq:simple-radial-equation}
 H(r(s))^2\dot r(s)^2=R(r(s)),
 \qquad \lim_{s\uparrow\Lambda}r(s)=0,
\end{equation}
where initially $\Lambda\le\infty$.  Define
\begin{align}
 \eta_1&\equiv h_1-\frac{\kappa_1}{2},
 \label{eq:simple-eta1}\\
 \eta_2&\equiv h_2-\frac{h_1\kappa_1}{2}
 +\frac{3\kappa_1^2}{8}-\frac{\kappa_2}{2},
 \label{eq:simple-eta2}\\
 k_1&\equiv\frac{\eta_1}{p+1},
 \label{eq:simple-k1}\\
 k_2&\equiv\frac{p+3}{2(p+1)^2}\eta_1^2-\frac{\eta_2}{p+2}.
 \label{eq:simple-k2}
\end{align}
Then $\Lambda<\infty$.  If
\begin{equation}\label{eq:simple-rho}
 u\equiv\Lambda-s,
 \qquad
 \rho(u)\equiv\left(\frac{p\sqrt{\kappa_0}}{h_0}u\right)^{1/p},
\end{equation}
then
\begin{equation}\label{eq:simple-r-expansion}
 r(s)=\rho-k_1\rho^2+k_2\rho^3+O(\rho^4).
\end{equation}

Suppose that a real-valued $C^1$ function $z$ along the curve satisfies,
for some $\ell<p$,
\begin{equation}\label{eq:simple-zdot}
 \dot z=r^{-\ell}\left(z_0+z_1r+z_2r^2+O(r^3)\right).
\end{equation}
Set
\begin{equation}
 n\equiv p-\ell,
 \qquad b_1\equiv z_1+\eta_1z_0,
 \qquad b_2\equiv z_2+\eta_1z_1+\eta_2z_0.
\end{equation}
Then $z$ has a finite limit $z_\star$ at $s=\Lambda$, and
\begin{align}
 z_\star-z(s)
 =\frac{h_0}{\sqrt{\kappa_0}}
 \bigg[&\frac{z_0}{n}\rho^n
 +\left(\frac{b_1}{n+1}-z_0k_1\right)\rho^{n+1}
 \nonumber\\
 &+\left(
 \frac{b_2}{n+2}-b_1k_1
 +z_0\left(k_2+\frac{n-1}{2}k_1^2\right)
 \right)\rho^{n+2}
 +O(\rho^{n+3})\bigg].
 \label{eq:simple-z-expansion}
\end{align}

Finally, suppose that a scalar function $\mathcal J$ along the curve
has the expansion
\begin{equation}\label{eq:simple-invariant-r}
 \mathcal J=r^{-m}\left(j_0+j_1r+j_2r^2+O(r^3)\right).
\end{equation}
Then
\begin{align}
 \mathcal J
 =\rho^{-m}\bigg[&j_0+(j_1+mk_1j_0)\rho\nonumber\\
 &+\left\{
 j_2+(m-1)k_1j_1
 +m\left(\frac{m+1}{2}k_1^2-k_2\right)j_0
 \right\}\rho^2+O(\rho^3)\bigg].
 \label{eq:simple-invariant-rho}
\end{align}
In particular, if $m>0$ and $j_0\ne0$, then
\begin{equation}\label{eq:simple-invariant-u}
 \mathcal J(s)=j_0
 \left(\frac{p\sqrt{\kappa_0}}{h_0}\right)^{-m/p}
 u^{-m/p}\left(1+O(u^{1/p})\right).
\end{equation}
\end{lemma}

\begin{proof}
By positivity and monotonicity,
\begin{equation}
 \dot r=-\frac{\sqrt{R(r)}}{H(r)}.
\end{equation}
Separating variables, for $s<S<\Lambda$, gives
\begin{equation}\label{eq:simple-separated-time}
 S-s=\int_{r(S)}^{r(s)}\frac{H(y)}{\sqrt{R(y)}}\,dy.
\end{equation}
Equations~\eqref{eq:simple-H}-\eqref{eq:simple-R} imply
\begin{equation}\label{eq:simple-time-integrand}
 \frac{H(r)}{\sqrt{R(r)}}
 =\frac{h_0}{\sqrt{\kappa_0}}r^{p-1}
 \left(1+\eta_1r+\eta_2r^2+O(r^3)\right).
\end{equation}
Since $p>0$, the integrand is integrable at zero.
Equation~\eqref{eq:simple-separated-time} therefore implies
$\Lambda<\infty$ and gives
\begin{equation}\label{eq:simple-u-integral}
 u=\int_0^{r(s)}\frac{H(y)}{\sqrt{R(y)}}\,dy.
\end{equation}
Termwise integration yields
\begin{equation}\label{eq:simple-u-r}
 u=\frac{h_0}{p\sqrt{\kappa_0}}r^p
 \left(1+\frac{p\eta_1}{p+1}r
 +\frac{p\eta_2}{p+2}r^2+O(r^3)\right).
\end{equation}
Taking the $p$-th root gives
\begin{align}
 \rho=r&+\frac{\eta_1}{p+1}r^2\nonumber\\
 &+\left(\frac{\eta_2}{p+2}
 +\frac{1-p}{2(p+1)^2}\eta_1^2\right)r^3+O(r^4).
 \label{eq:simple-rho-r}
\end{align}
If $\rho=r+\alpha r^2+\beta r^3+O(r^4)$, direct substitution gives
\begin{equation}
 r=\rho-\alpha\rho^2+(2\alpha^2-\beta)\rho^3+O(\rho^4).
\end{equation}
Applying this identity to \eqref{eq:simple-rho-r}, with
$\alpha=k_1$ and $2\alpha^2-\beta=k_2$, proves
\eqref{eq:simple-r-expansion}.

Multiplying \eqref{eq:simple-zdot} by
\eqref{eq:simple-time-integrand} gives
\begin{equation}
 \dot z\,\frac{H(r)}{\sqrt{R(r)}}
 =\frac{h_0}{\sqrt{\kappa_0}}r^{n-1}
 \left(z_0+b_1r+b_2r^2+O(r^3)\right).
\end{equation}
Since $n>0$, the expression is absolutely integrable at zero.
Hence $z_\star$ exists, and
\begin{equation}\label{eq:simple-z-before-inversion}
 z_\star-z(s)=\frac{h_0}{\sqrt{\kappa_0}}
 \left[\frac{z_0}{n}r^n+\frac{b_1}{n+1}r^{n+1}
 +\frac{b_2}{n+2}r^{n+2}+O(r^{n+3})\right].
\end{equation}
For every fixed real $\nu$, \eqref{eq:simple-r-expansion} gives
\begin{equation}
 r^\nu=\rho^\nu\left[1-\nu k_1\rho+
 \left(\nu k_2+\frac{\nu(\nu-1)}2k_1^2\right)\rho^2
 +O(\rho^3)\right].
\end{equation}
Using $\nu=n,n+1,n+2$ in \eqref{eq:simple-z-before-inversion}
proves \eqref{eq:simple-z-expansion}.

For the invariant expansion, use the same identity with $\nu=-m$
together with
\begin{equation}
 j_0+j_1r+j_2r^2
 =j_0+j_1\rho+(j_2-k_1j_1)\rho^2+O(\rho^3).
\end{equation}
Multiplication gives \eqref{eq:simple-invariant-rho}, and
\eqref{eq:simple-rho} gives \eqref{eq:simple-invariant-u}.
\end{proof}

\begin{corollary}[A simple-zero curvature criterion]
\label{cor:simple-zero-curvature-criterion}
Let $(N,h)$ be a connected smooth time-oriented spacetime, and let
$\mathcal I_1,\ldots,\mathcal I_m$, with $m<\infty$, be curvature contractions of the
kind occurring in Theorem~\ref{thm:general-curvature-criterion}.
Suppose that, for each choice of time orientation, every
future-inextendible unit-speed timelike geodesic of finite proper time
admits, on a terminal interval, a decreasing function $r$ and data
$H,R,h_0,\kappa_0,q$ satisfying
\eqref{eq:simple-H}-\eqref{eq:simple-radial-equation}, with
$q+\tfrac12>0$.  Suppose also that, for at least one index
$j=j(\gamma)$,
\begin{equation}
 \mathcal I_j(\gamma(s))
 =r(s)^{-m_\gamma}\left(j_{0,\gamma}+O(r(s))\right),
 \qquad m_\gamma>0,\qquad j_{0,\gamma}\ne0.
\end{equation}
The data and the index may depend on the geodesic.  Then $(N,h)$ is
$C^{1,1}_{\mathrm{loc}}$-inextendible.  Along each such terminal branch,
\begin{equation}
 |\mathcal I_j(\gamma(s))|
 \asymp(\Lambda-s)^{-m_\gamma/(q_\gamma+1/2)},
\end{equation}
where $q_\gamma$ is the exponent in \eqref{eq:simple-H} for that
geodesic.
\end{corollary}

\begin{proof}
Put $p_\gamma\equiv q_\gamma+\tfrac12$.  Equations
\eqref{eq:simple-rho}-\eqref{eq:simple-r-expansion} give
$r=\rho(1+O(\rho))$, with
$\rho=(p_\gamma\sqrt{\kappa_{0,\gamma}}
(\Lambda-s)/h_{0,\gamma})^{1/p_\gamma}$.  Combining this with the
assumed $r^{-m_\gamma}(j_{0,\gamma}+O(r))$ expansion gives
\begin{equation}
 |\mathcal I_j(\gamma(s))|
 =|j_{0,\gamma}|
 \left(\frac{p_\gamma\sqrt{\kappa_{0,\gamma}}}
 {h_{0,\gamma}}\right)^{-m_\gamma/p_\gamma}
 (\Lambda-s)^{-m_\gamma/p_\gamma}
 \left(1+O((\Lambda-s)^{1/p_\gamma})\right).
\end{equation}
This proves the stated estimate and divergence, so
Theorem~\ref{thm:general-curvature-criterion} applies.
\end{proof}

\begin{remark}
Application of the criterion to a separable metric requires a
reduction of terminal geodesics to the stated radial equation, a
simple radial zero, and a nonzero leading curvature coefficient.
For Kerr, these properties follow respectively from
Theorem~\ref{thm:terminal-geodesics}, \eqref{eq:terminal-R-exact},
and Proposition~\ref{prop:weyl-invariants}.  They are additional
geometric properties, rather than consequences of separation alone.
\end{remark}

\section*{Acknowledgments}
The author thanks his family for their support and encouragement.

\section*{Funding}
This work is funded by the Indonesian Endowment Fund for Education (LPDP), on behalf of the Indonesian Ministry of Higher Education, Science and Technology, and managed under the EQUITY Program (Contract No. 4298/B3/DT.03.08/2025).

\section*{Data availability}
No datasets were generated or analyzed in this work.

\end{document}